\documentclass{article}

\usepackage[utf8]{inputenc}
\usepackage[margin = 1in]{geometry}
\usepackage{amsmath}
\usepackage{amsfonts}
\usepackage{amssymb}
\usepackage{amsthm}
\usepackage{mathtools}
\usepackage{physics}
\usepackage{overpic}
\usepackage{epstopdf}
\usepackage{graphicx}
\usepackage[hidelinks]{hyperref}
\usepackage{bm}
\usepackage[font = small]{caption}
\usepackage[font = small]{subcaption}

\newcommand{\eps}{\ensuremath{\varepsilon}}

\newcommand{\bs}{\boldsymbol}
\newcommand{\mbf}{\mathbf}
\newcommand{\jac}{J_{\mbf u} \, }

\newcommand{\ol}{\overline}

\newcommand*\circled[1]{\tikz[baseline=(char.base)]{\node[shape = circle, draw, inner sep = 2pt] (char) {#1};}}

\allowdisplaybreaks

\title{Pattern formation of bulk-surface reaction-diffusion systems in a ball.}

\author{Edgardo Villar-Sep\'ulveda\thanks{Department of Engineering Mathematics, University of Bristol, Bristol, UK (edgardo.villar-sepulveda@bristol.ac.uk).}
\and Alan R.~Champneys\thanks{Department of Engineering Mathematics, University of Bristol, Bristol, UK (a.r.champneys@bristol.ac.uk).}
\and Davide Cusseddu\thanks{CMAT – Centro de Matemática, Universidade do Minho, Braga, Portugal
(davide.cusseddu@gmail.com).}
\and Anotida Madzvamuse\thanks{University of British Columbia, Department of Mathematics, 1984 Mathematics Road, Vancouver, V6T 1Z2, British Columbia, Canada; University of Pretoria and University of Johannesburg, Department of Mathematics, South Africa, and University of Zimbabwe, Department of Mathematics and Computational Science, Mt Pleasant, Harare, Zimbabwe
(am823@math.ubc.ca, \url{https://www.math.ubc.ca/user/3665}).}}
\date{}

\usepackage{amsopn}
\DeclareMathOperator{\diag}{diag}

\newcommand{\rtext}[1]{\textcolor{black}{#1}}

\newcommand{\s}{\\}
\newcommand{\innerp}[1]{\left \langle #1 \right \rangle}

\usepackage{amssymb}
\usepackage{bbm}
\usepackage{subcaption}
\usepackage{tikz}
\usepackage{amsthm}

\usepackage{mathtools}
\usepackage{ulem}
\usepackage{physics}
\usepackage{xcolor}

\definecolor{DCbrown}{RGB}{123,63,0}

\newcommand{\evs}[1]{\textcolor{black}}

\ifpdf
\hypersetup{
  pdftitle={Pattern formation of bulk-surface reaction-diffusion systems in a ball.},
  pdfauthor={Edgardo Villar-Sep\'ulveda, Alan Champneys, Davide Cusseddu, and Anotida Madzvamuse}
}
\fi

\begin{document}

\maketitle

\begin{abstract}
    Weakly nonlinear amplitude equations are derived for the onset of spatially extended patterns on a general class of $n$-component bulk-surface reaction-diffusion systems in a ball, under the assumption of linear kinetics in the bulk \rtext{and coupling Robin-type boundary conditions}. Linear analysis shows conditions under which various pattern modes can become unstable to either generalised pitchfork or transcritical bifurcations depending on the parity of the spatial wavenumber. Weakly nonlinear analysis is used to derive general expressions for the multi-component amplitude equations of different patterned states. These reduced-order systems are found to agree with prior normal forms for pattern formation bifurcations with $O(3)$ symmetry and \rtext{they} provide information on the stability of bifurcating patterns of different symmetry types. The analysis is complemented with numerical results using a dedicated \rtext{bulk-surface} finite-element method. The theory is illustrated in two examples: a bulk-surface version of the Brusselator and a four-component \rtext{bulk-surface} cell-polarity model.
\end{abstract}



\section{Introduction}
    Recently, there have been substantial developments in extending the ground-breaking theory for pattern formation proposed by Alan Turing in his seminal paper in 1952 \cite{turing1952chemical}; see, for example, \cite{Krause,madzvamuse2010stability} and references therein. One such development, driven primarily by experimental observations in biology and material science, is the formulation of bulk-surface reaction-diffusion models (BS-RDEs), see e.g \cite{madzvamuse2016bulk,madzvamuse2015,ratz2015turing, ratz2014symmetry}. \rtext{A bulk-surface reaction-diffusion system is one posed} on two different, yet compatible geometries; one denoted the {\it bulk}, which describes the interior of the domain; and the {\it surface}, which describes the enclosing smooth manifold representing its boundary. These models naturally describe many biological and physical processes, such as symmetry-breaking and polarity formation in cell biology \cite{brauns2021bulk,cusseddu2019,cusseddu2022numerical}, surfactant dynamics \cite{chen2014conservative,fernandez2016existence,hahn2014modelling}, and spatiotemporal electrodeposition in batteries \cite{frittelli2023turing}, to give just a few examples. 
    
    The complex coupled nature of such models tends to involve not only nonlinear reaction kinetics but also different spatial operators (Laplace and Laplace-Beltrami) for fields posed on the interior and boundary of the domain. The mathematical analysis of the resulting so-called bulk-surface partial differential equations is still in its infancy, and we do not attempt a full existence and uniqueness theory here. Instead, we conduct linear, weakly nonlinear, and numerical analysis of a general class bulk-surface reaction-diffusion systems (BS-RDEs) in 3D. \rtext{Analytical studies on the existence, uniqueness, and uniform boundedness of solutions of bulk-surface reaction-diffusion systems have been recently carried out (see \cite{morgan2023global,morgan2017uniform,sharma2016global} for specific details). In \cite{sharma2016global}, criteria for guaranteeing uniform boundedness of solutions in time for BS-RDEs was established. The analysis was based on the assumption that the reaction kinetics are locally Lipschitz and quasi-positive. Global well-posedness and global existence of solutions for BS-RDEs were subsequently shown in \cite{morgan2017uniform,morgan2023global}. The reaction kinetics considered in this study fulfill these assumptions, thereby guaranteeing the existence and uniqueness of solutions for this generalised BS-RDEs.}
     
    The class of BS-RDEs we consider consists of an arbitrary number $n \geq 2$ of components whose dynamics on the bulk and surface are coupled through linear mixed Robin (inhomogeneous Neumann) boundary conditions. \rtext{To date, most of the analysis has been restricted to the case where there is a natural compatibility between the eigenfunctions of the {\it ground state} steady solution with respect to the Laplace and Laplace-Beltrami operators.} The simplest geometries with natural compatibility are a circle in 2D (the bulk being the interior of the circle and the surface being its boundary) and a ball in 3D. \rtext{Hence, in this study} we restrict \rtext{our} analysis to the case of the ball with linear kinetics in the bulk; see Sec.~\ref{sec:model} for more details.
    
    \rtext{This study aims} to understand the nature of patterned states that can form under parameter variation as a bifurcation from a trivial radially-symmetric ground state. One inspiration for the present work is due to Paquin-Lefebvre {\it et al} \cite{paquin-lefebvre}, who follow a similar procedure on a disk, where rotational symmetry is exploited to carry out the analysis. However, the case of a bulk-surface reaction-diffusion model posed on a ball is more complex due to the extra symmetries involved. \rtext{The interested reader is referred to} \cite{charette2018,matthews,matthews2,callahan} and references therein for analogous bifurcation results for reaction-diffusion systems posed on the surface of a sphere.
     
    The main result (see Sec.~\ref{sec:main}) is to derive general expressions for amplitude equations for small-amplitude steady patterns that can be used to determine the bifurcation and stability of different patterned states. We find that these amplitude equations take the same form as those describing Turing bifurcations of PDE systems posed on the surface of the sphere \cite{callahan}, owing to the underlying $O(3)$ symmetry of both problems.
    
    We complement this analytical work with numerical solutions to the initial value problem to illustrate patterns emerging from the fully nonlinear system and to show their transient dynamics. These results are developed using a state-of-the-art finite element method for nonlinear semilinear parabolic partial differential equations (PDEs) of which BS-PDEs are one such class \cite{alphonse2018coupled,elliott2013finite, elliott2017coupled, madzvamuse2016bulk,madzvamuse2015}.
    
    To illustrate our results, we choose two examples. One is a bulk-surface version of the well-studied Brusselator model \cite{Brusselator} in which there is an activator and an inhibitor field. The other is a four-component system that is inspired by problems in cell biology, specifically the reaction-diffusion kinetics of small GTPases; these are key molecular proteins that exhibit co-localisation of their molecular concentrations and interactions in both the bulk and surface geometries as active and inactive species. \rtext{The reader is referred to \cite{beta,MeronIssue} for general introductions to models describing} such biological processes and  \cite{borgqvist2021cell,cusseddu2019,cusseddu2022numerical,singh2022sensing} for \rtext{work} on specific bulk-surface versions of them. 
    
    For each example, we identify bifurcation curves in a parameter plane, corresponding to instabilities of the ground state for different spatial wavenumbers. From there, we derive the normal forms and give predictions for the bifurcation and stability of different kinds of steady patterns. These are corroborated by numerical simulations, which suggest the existence of more spatially localised patterns that are at least meta-stable. However, for higher wavenumbers, the amplitude equations have many components, and their full analysis is unknown. A complete analysis of their spatiotemporal dynamics, including the analogs of possible localised patterns after subcritical bifurcation \rtext{(see e.g.,~\cite{bramburger-review,WoodsReview})}, will be left for future work.
            
    The rest of the article is structured as follows. In Section \ref{sec:model} we present the general class of BS-RDEs \rtext{under study}, its linearised analysis together with the main result on how to construct the appropriate weakly nonlinear normal forms at pattern formation instabilities of different wavenumbers. The normal forms that arise have been analysed before in the context of systems with $O(3)$ symmetry, and we outline some of their dynamics in Section \ref{sec:explanation}. The proof of the main result is contained in Section \ref{sec:main}. Section \ref{sec:examples} presents two examples where we show how the results of the explicit computation of normal form coefficients explain what is observed in the finite-element simulations. Finally, Section \ref{sec:conclusions} contains concluding remarks and suggests avenues for future work.

\section{General theory} \label{sec:model}
    We consider a general class of BS-RDEs posed in a ball of radius $R > 0$. For analytical purposes, we restrict the class of PDEs in the bulk to be linear. We then find a general expression for its ground state and categorise its spectrum using modified spherical Bessel functions and Legendre polynomials. This leads to a statement of our main result, which is to provide closed-form expressions for the amplitude equations valid in a neighbourhood of pattern formation bifurcations for different wavenumbers. 

    \subsection{A class of bulk-surface reaction-diffusion systems in a ball}
        Consider the ball $\Omega = \left\{\mbf x \in \mathbb R^3 : \lVert \mbf x \rVert \leq R \right\}$ which we refer to as the \textit{bulk} domain; its boundary $\Gamma = \partial\Omega = \left\{\mbf x \in \mathbb R^3 : \lVert \mbf x \rVert = R \right\}$ will be called the \textit{surface}. \rtext{Let $r = \norm{\mbf x}$, and co}nsider an $n$-component field of spatiotemporal unknowns, designating its bulk components using capital letters,  $ \rtext{U_1 (\mbf x, t), \ldots, U_n (\mbf x, t): \Omega \times \mathbb R_0^+ \to \mathbb R}$, and surface components using lowercase letters $\rtext{u_1(\mbf x, t), \ldots, u_n(\mbf x, t)}: \Gamma \times \mathbb R_0^+ \to \mathbb R$. In compact form, \rtext{following the presentation given in \cite{paquin-lefebvre},} we suppose $\mbf U = \left(U_1, \ldots, U_n\right)^\intercal$ and $\mbf u = \left(u_1, \ldots, u_n\right)^\intercal$ satisfy the following system
        \begin{align}
            \begin{cases}
                \begin{array}{rcl}
                    \partial_t \mbf U &=& - B \, \mbf U + \mathbb D_U \nabla^2_\Omega \, \mbf U, \qquad \rtext{\text{for } (\mbf x, t) \in \Omega \times \mathbb R_0^+},
                    \\[1ex]
                    \rtext{\mathbb D_U \left. \dfrac{\partial \mbf U}{\partial \mbf{\hat n}}\right|_{r = R}} &=& K\left(\mbf u-\left.\mbf U\right|_{\rtext{r = R}}\right), \qquad \rtext{\text{for } (\mbf x, t) \in \Gamma \times \mathbb R_0^+}
                    \\[1ex]
                    \partial_t \mbf u &=& \mbf g(\mbf u) - K \left(\mbf u - \left. \mbf U\right|_{\rtext{r = R}}\right) + \mathbb D_u \, \nabla^2_\Gamma \, \mbf u, \qquad \rtext{\text{for } (\mbf x, t) \in \Gamma \times \mathbb R_0^+}
                \end{array}
            \end{cases}
            \label{generalsystem}
        \end{align}
        where $B = \diag\left(\sigma_1, \ldots, \sigma_n\right) \in \left(\mathbb R_0^+\right)^{n \times n}$, $K = \diag\left(K_1, \ldots, K_n\right) \in \left(\mathbb R_0^+\right)^{n\times n}$ \rtext{is a matrix formed out of so-called Langmuir-rate constants {\it cf.}~\cite{paquin-lefebvre}}, $\mbf g: \mathbb R^n \to \mathbb R^n$ is a $\mathcal C^3$ nonlinear vector function, $\mbf{\hat n}$ is the outward unit vector to $\Omega$ on $\Gamma$, $\mathbb D_U = \diag\left(D_1, \ldots, D_n\right)$, $\mathbb D_u = \diag\left(d_1, \ldots, d_n\right)$\rtext{, $\bs \nabla_\Omega^2$ is the Laplace operator in $\Omega$, and $\bs \nabla_\Gamma^2$ is the Laplace-Beltrami operator on $\Gamma$}. We allow $B$, $\mathbb D_U$, $K$, and $\mbf g$ to depend on a parameter, $\mu$, defined as the unfolding parameter in the sequel. For simplicity \rtext{of notation}, we omit the parameter dependence of these functions. Note that the presence of the same factor $K \left(\mbf u - \left. \mbf U\right|_{\rtext{r = R}}\right)$ in the 2nd and 3rd equations of \eqref{generalsystem} leads to \rtext{a} consistent \rtext{use of the Langmuir rate equation \cite{Langmuir} in our} \rtext{Robin-type boundary conditions}, which couple \rtext{the} bulk and surface fields.

    \subsection{The ground state solution} 
        We seek a radially symmetric steady state for the coupled system. \rtext{We denote this symmetric steady state as the {\it ground state,} which we define as a time-independent, spatially inhomogeneous solution. To proceed,} let $r = \norm{\mbf x}$, and seek \rtext{a ground state solution} $\mbf P = \left(\mbf U^*, \mbf u^*\right)^\intercal = \left(\mbf U^*(r), \mbf u^*\right)^\intercal$ of the system, \rtext{which satisfies}  $\partial_t \mbf U^* = \mbf 0$ and $\partial_t \mbf u^* = \mbf 0$. In particular, passing to spherical coordinates, $\mbf U^*(r)$ must satisfy
        \begin{align}
    		r^2 \, \frac{\partial^2 \mbf U^*}{\partial r^2} + \rtext{2 r} \, \frac{\partial \mbf U^*}{\partial r} - r^2 \, \mathbb D_U^{-1} \, B \, \mbf U^*&=\mbf 0, \qquad \text{for } r\in (0,R). \label{eq:radial}
    	\end{align}
        We seek a solution to \eqref{eq:radial} of the form $\mbf U^* = i_0\left(\bs \omega \, r\right) \, \mbf U^\dagger$, where $\mbf U^\dagger\in \mathbb R^n$ is a constant vector, $\bs \omega = \diag\left(\omega_1, \ldots, \omega_n\right)$, with $\omega_p = \sqrt{\sigma_p/D_p}$, for $p = 1, \ldots, n$, and $i_0\left(\bs \omega \, r\right) = \diag\left(i_0\big(\omega_1 \, r\big), \ldots, i_0\big(\omega_n \, r\big)\right)$, where $i_0$ is a modified spherical Bessel function of the first kind \cite[eq. 10.2.12]{besselfunctions}. 
        
        To determine $\mbf U^\dagger$ and $\mbf u^*$, we use the \rtext{coupling} boundary conditions (the second set of equations on \eqref{generalsystem}). Thus, we obtain $\mathbb D_U \, i_0'(\bs \omega R) \, \bs \omega \, \mbf U^\dagger =K\left(\mbf u^* - i_0(\bs \omega R) \, \mbf U^\dagger\right)$, which can be rearranged into the form 
    	\begin{align}
            \mbf U^\dagger = S_0 \, \mbf u^*, \quad \mbox{where } S_0\coloneqq \left(\mathbb D_U \, i_0'(\bs \omega R) \, \bs \omega + K \, i_0(\bs \omega R) \right)^{- 1} K. \label{eq:S0}
        \end{align}
        Here, we use $i_0'\left(\bs \omega \, r\right)$ to denote the diagonal matrix with entries $i_0'\left(\omega_p \, r\right)$, $p = 1, \ldots, n$, and a prime means differentiation with respect to $r$.
        
        Finally, we use the third \rtext{set of equations on} \eqref{generalsystem} to obtain:
    	\begin{align}
    		\mbf g(\mbf u^*) - K\left(I - i_0(\bs \omega R) \, S_0\right)\mbf u^* = \mbf 0, \label{nonlineareq}
        \end{align}
        where $I$ is \rtext{the $n \times n$ identity matrix}. Therefore, a radially-symmetric steady state of the form $\mbf P = \left(\mbf U^*(r), \mbf u^*\right)^\intercal$ can be found whenever there is a solution $\mbf u^*$ to \eqref{nonlineareq}, by setting $\mbf U^* = i_0(\bs \omega r) \, S_0 \, \mbf u^*$, with $S_0$ given by \eqref{eq:S0}.
	
    \subsection{Linearised analysis about the ground state}
        Now, let $\jac \mbf f(\mbf P)$ denote the Jacobian matrix of \eqref{generalsystem} at $\mbf P$ in the absence of diffusion, i.e., $\jac \mbf f(\mbf P) = \begin{pmatrix}
            -B & \mbf 0_{n\times n}
            \s 
            K \, \mathbbm I|_{\rtext{r = R}} & -K + \jac \mbf g\left(\mbf u^*\right)
        \end{pmatrix}$. Here, $\jac \mbf g\left(\mbf u^*\right)$ denotes the Jacobian matrix of $\mbf g$ at $\mbf u^*$ and $\mathbbm I|_{\rtext{r = R}}$ means restriction onto \rtext{ $\Gamma$} such that $\mathbbm I|_{\rtext{r = R}}(\mbf U) = \mbf U|_{\rtext r = R}$. Also, let $\mathbb D$ denote the diffusion matrix of \eqref{generalsystem} given by $\mathbb D = \begin{pmatrix}
            \mathbb D_U & \mbf 0_{n\times n}
            \s 
            \mbf 0_{n\times n} & \mathbb D_u
	    \end{pmatrix}$. This implies that the Jacobian \rtext{operator} of the full \rtext{ linearised PDE} system at $\mbf P$ is given by 
        $$
            \mathcal L:= \jac \mbf f(\mbf P) + \mathbb D \, \bs \nabla^2, \quad \mbox{ where } \bs \nabla^2 = \left(\nabla_\Omega^2, \nabla_\Gamma^2\right)^\intercal.
        $$
        In what follows, we shall be interested in bifurcations \rtext{from the} ground state, \rtext{$\mbf P$}. \rtext{For the linear part of} such \rtext{an} analysis, \rtext{we need to find values of a free parameter at which the kernel of $\mathcal L$ has a simple zero eigenvalue.} \rtext{The problem is similar to the equivalent bifurcation problem for} reaction-diffusion systems on the sphere with a constant steady state $\mbf u^*$, \rtext{with the added complication that we need to deal with the compatibility of the Laplace and Laplace-Beltrami operators, and need to solve the Helmholz} equation in the bulk.

        \rtext{Specifically,} we are interested in the space of eigenfunctions of the operator $\bs \nabla^2$, which satisfy 
        \begin{align}
            \nabla_\Omega^2 \, \rtext{\mbf V}(r, \theta, \phi) = - k^2 \, \rtext{\mbf V}(r, \theta, \phi), \qquad \text{and} \qquad \nabla_\Gamma^2 \, \rtext{\mbf v}(\theta, \phi) = - \frac{\ell (\ell + 1)}{R^2} \, \rtext{\mbf v}(\theta, \phi), \label{eigenvalueproblem}
        \end{align}
        with eigenvalues $- k^2$ and $- \ell (\ell + 1)/R^2$ in the bulk and on the sphere, respectively. Specifically, we find that the eigenfunctions satisfy
    	\begin{align}
    		\begin{pmatrix}
    		    \mbf U(r, \theta, \phi)
                \\
                \mbf u(\theta, \phi)
    		\end{pmatrix} &= Y_\ell^m(\theta, \phi) \begin{pmatrix}
    		    j_\ell(kr) \, \bs{\hat \varphi}
                \\
                \bs \varphi
    		\end{pmatrix}, \label{eigenfunctions}
    	\end{align}
        where $Y_\ell^m(\theta, \phi) = P_\ell^m(\cos(\theta)) \, e^{im\phi}$, $j_\ell$ is a spherical Bessel function of the first kind, $P_\ell^m$ is an associated Legendre polynomial (normalized as in \cite[eqns.~(2.1.20-21)]{legendre-poly} so that $\int_0^\pi \left(P_\ell^m(\cos(\theta))\right)^2 \, \dd \theta = 1$, and $\bs{\hat \varphi}, \, \bs{\varphi}\in \mathbb R^n$ are constant vectors \cite{courant}. \rtext{We highlight that the eigenfunctions of the Laplace-Beltrami operator correspond to the unit sphere. Due to this scaling, the eigenvalue of the Laplace-Beltrami operator will need to be scaled by $R^2$ \cite{chaplain2001spatio}.}

        \rtext{Our task in what follows is to establish conditions under which there is a regular bifurcation point that is of codimension-one in the parameter space and to describe the nature of the non-trivial steady solutions that emerge. We shall find that such bifurcations may either be of transcritical or pitchfork type, see e.g.~\cite{Iooss}. As it is well known, in the case of a pitchfork, these may be super or subcritical, and the bifurcating solutions can be stable or unstable as solutions to the initial-value problem. As we shall see in the next section, there is additional complexity to the structure of the bifurcating solutions, owing to the $O(3)$-symmetry of the underlying bulk-surface PDEs} 
        
        \rtext{To study system \eqref{generalsystem}, based on similar studies on related problems \cite{Fahad,tirapegui,paquin-lefebvre,edgardodegenerate}, we use the approach of developing amplitude equations.} Specifically, we consider perturbations to the ground state of the form
        \begin{align}
    		\begin{pmatrix}
    			\mbf U
    			\\
    			\mbf u
    		\end{pmatrix} = \mbf P + \sum_{m = - \ell}^\ell A_m \begin{pmatrix}
    			i_\ell\left(\bs \omega \, r\right) \, \bs{\hat \varphi}
    			\\
    			\bs \varphi
    		\end{pmatrix}  Y_\ell^{\rtext{m}}(\theta,\phi)
            \label{ansatz}
    	\end{align}
        where $A_m = A_m(t)\in \mathbb C$ is a complex variable such that $A_{-m} = (-1)^m \, \overline{A_m}$ for each $-\ell \leq m \leq \ell$\rtext{, and $i_\ell(\cdot) = i^{- \ell} \, \rtext{ j_\ell(i \, \cdot)}$ is a modified spherical Bessel function of the first kind \cite{besselfunctions}}.

    \subsection{The main result}
        As is common for bifurcation problems on a sphere \cite{callahan,chossat,matthews2}, we shall find that odd (respectively, even) values of $\ell$ give rise to pitchfork (respectively, transcritical) pattern-formation bifurcations for the amplitudes of the bulk-surface \rtext{reaction-diffusion} system \eqref{generalsystem}. We shall now formalise this result and give expressions for amplitude equations we obtain at each bifurcation in terms of the coefficients of the underlying system. It is convenient to introduce some notation. For each positive integer (wavenumber) $\ell$, we define the matrix
        \begin{equation}
            \mathbb B_\ell : = K\left(i_\ell\left(\bs \omega \, R\right) \, S_\ell - I\right) + \jac \mbf g\left(\mbf u^*\right) - \frac{\ell (\ell + 1)}{R^2} \, \mathbb D_u, \label{eq:Bdef}
        \end{equation}
        where \: $S_\ell\coloneqq \left(\mathbb D_U \, i_\ell'(\bs \omega \, R) \, \bs \omega + K \, i_\ell(\bs \omega \, R)\right)^{-1} K$. Moreover, we introduce the following multilinear symmetric vector functions:
    	\begin{align}
    	    \mbf G_{1, 1}(\mbf a) &= \rtext{\left. \left(\sum_{p = 1}^n a_p \, \frac{\partial}{\partial \mu} \left(\left. \frac{\partial \rtext{\mbf g}}{\partial u_p}(\mbf u)\right|_{\mbf u = \mbf u^*(\mu)}\right)\right)\right|_{\mu = \mu^*},} \label{eqn:G11}
    	    \s
    		\mbf G_2\left(\mbf a, \mbf b\right) &= \rtext{\frac{1}{2} \, \left. \left(\sum_{1 \leq p, q \leq n} a_p \, b_q \, \left. \frac{\partial^2 \rtext{\mbf g}}{\partial u_p \, \partial u_q}(\mbf u)\right|_{\mbf u = \mbf u^*(\mu)}\right)\right|_{\mu = \mu^*},} \label{eqn:G2}
    		\s 
    		\mbf G_3\left(\mbf a, \mbf b, \mbf c\right) &= \rtext{\frac{1}{6} \, \left.\left(\sum_{1 \leq p, q, s \leq n} a_p \, b_q \, c_s \, \left. \frac{\partial^3 \rtext{\mbf g}}{\partial u_p \, \partial u_q \, \partial u_s}(\mbf u)\right|_{\mbf u = \mbf u^*(\mu)}\right)\right|_{\mu = \mu^*},} \label{eqn:G3}
    	\end{align}
        where $\mbf a = \left(a_1, \ldots, a_n\right)^\intercal \in \mathbb R^n$, $\mbf b = \left(b_1, \ldots, b_n\right)^\intercal \in \mathbb R^n$, $\mbf c = \left(c_1, \ldots, c_n\right)^\intercal \in \mathbb R^n$, and $\mu^*\in \mathbb R$. The main result of this study can  be stated as follows:
        \newtheorem{theorem}{Theorem}
    	\begin{theorem} \label{th:maintheorem}
    		Let $\mbf P = \left(\mbf U^*(r), \mbf u^*\right)^\intercal$ be a radially-symmetric steady state of \eqref{generalsystem} and $\mathbb B_\ell$, $\mbf G_{1,1}$, $\mbf G_2$, $\mbf G_3$ as defined in \eqref{eq:Bdef}-\eqref{eqn:G3}. If there exists $\mu^*\in \mathbb R$ and a \textit{critical wavenumber} $\ell$ such that \rtext{the following three conditions hold}
            \begin{equation}
                \det\left(\mathbb B_\ell\right) = 0, \quad \ker\left(\mathbb B_\ell\right) \cap \Im\left(\mathbb B_\ell\right) = \left\{\mbf 0\right\}, \quad \mbox{and} \quad C_{1, 1} \neq 0, \label{kerim}
            \end{equation}
            where
            \begin{align*}
                & C_{1,1} = 
                    \frac{1}{I_\ell} \left(\int_0^R \left(- \frac{\partial B}{\partial \mu} \, i_\ell(\bs \omega \, r) \, S_\ell \, \bs \varphi\right) \cdot \left(i_\ell(\bs \omega \, r) \, S_\ell \, \bs \psi \right) r^2 \, \dd r \right. \notag
                    \\
                    & \hspace{1cm} + R^2 \, \left(\mbf G_{1,1}(\bs \varphi) - \frac{\partial K}{\partial \mu} \, \left(\bs \varphi - i_\ell(\bs \omega \, R) \, \bs{\hat \varphi}\right) \right) \cdot \bs \psi \notag
                    \\
                    & \hspace{0.2cm} \left. - \frac{\ell(\ell + 1)}{R^2} \left(\int_0^R \left(\frac{\partial \mathbb D_U}{\partial \mu} \, i_\ell(\bs \omega \, r) \, S_\ell \, \bs \varphi\right) \cdot \left(i_\ell(\bs \omega \, r) \, S_\ell \, \bs \psi\right) r^2 \, \dd r + R^2 \, \left(\frac{\partial \mathbb D_u}{\partial \mu} \, \bs \varphi\right) \cdot \bs \psi\right)\right),
                \\
                &I_\ell = \int_0^R i_\ell(\bs \omega r) \, S_\ell \, \bs \varphi \cdot i_\ell(\bs \omega r) \, S_\ell \, \bs \psi \, r^2 \, \dd r + R^2 \, \bs \varphi \cdot \bs \psi, \notag
            \end{align*}
            with $\bs \varphi \in \ker\left(\mathbb B_\ell\right)\setminus \{\mbf 0\}$ and $\bs \psi \in \ker\left(\mathbb B_\ell^\intercal\right)\setminus \{\mbf 0\}$, then
    		\begin{enumerate}
                \item If $\ell$ is even, and the quadratic coefficients of the amplitude equations \eqref{quadratic_coef} are different from zero, then the system undergoes a codimension-one (Turing) transcritical bifurcation as $\mu$ varies through $\mu^*$.
                \item If $\ell$ is even and the quadratic coefficients of the amplitude equations \eqref{quadratic_coef} are equal to zero, then the system undergoes a codimension-two (Turing) transcritical bifurcation as $\mu$ varies through $\mu^*$.
                \item If $\ell$ is odd, and $\det \left(\mathbb B_p\right) \neq 0$ for every even integer $0 \leq p\leq 2 \ell$, then the system undergoes a (Turing) pitchfork bifurcation as $\mu$ varies through $\mu^*$.
    		\end{enumerate}
    		Furthermore, all of these bifurcations are described by a vector \\ \rtext{$\mbf A_\ell = \big(A_{- \ell}, A_{- \ell + 1}, \ldots, A_{\ell - 1}, A_\ell\big)^\intercal$ of $2 \ell + 1$} amplitudes that fulfill the following:
    		\begin{enumerate}
    			\item If $\ell$ is even and the quadratic coefficients of the amplitude equations, \eqref{quadratic_coef}, are different from zero, then
    			\begin{align}
    				\dot A_m &= C_{1, 1} \, \left(\mu - \mu^*\right) \, A_m + \sum_{\substack{-\ell\leq q_1 \leq q_2\leq \ell \s q_1 + q_2=m}} C_{q_1,q_2,m}^{[2]} A_{q_1}A_{q_2} + \mathcal O\big(\norm{\rtext{\mbf A_\ell}}^3\big), \label{eq:nf2}
    			\end{align}
    			for each  $m\in [-\ell, \ell]\cap \mathbb Z$, where
    			\begin{align}
    				C_{q_1, q_2, m}^{[2]} &= \frac{R^2}{I_\ell} \, \delta_{m, q_1 + q_2} \, d_{\ell, q_1, q_2}^{[2]} \, \mbf G_2\left(\bs \varphi, \bs \varphi\right) \cdot \bs \psi, \quad \rtext{\text{with}} \label{quadratic_coef}
                    \\
                    \rtext{d_{\ell, q_1, q_2}^{[2]}} &\rtext{= \int_{- 1}^1 P_\ell^{q_1}(x) \, P_\ell^{q_2}(x) \, P_\ell^{q_1 + q_2}(x) \, \dd x,} \notag
    			\end{align}
                and $\delta_{m, q_1 + q_2}$ is the \textit{Kronecker delta function} that is equal to 1 if $m = q_1 + q_2$ and 0 otherwise.

    			\item If $\ell$ is odd or $\ell$ is even and the quadratic coefficients of the amplitude equations \eqref{quadratic_coef} are equal to zero, then
    			\begin{align}
    				\dot A_m &= C_{1,1} \, \left(\mu - \mu^*\right) \, A_m + \sum_{\substack{-\ell\leq q_1\leq q_2\leq q_3\leq \ell \s q_1+q_2+q_3 = m}} C_{q_1,q_2,q_3,m}^{[3]} A_{q_1} A_{q_2} A_{q_3} + \mathcal O\big(\norm{\rtext{\mbf A_\ell}}^4\big), \label{eq:nf3}
    			\end{align}
    			for each $m\in [-\ell, \ell]\cap \mathbb Z$, where 
    			\begin{align*}
    				C_{q_1, q_2, q_3, m}^{[3]} &=\begin{multlined}[t][0cm]
        			    \frac{R^2}{I_\ell} \, \delta_{m, q_1 + q_2 + q_3} \left(2 \sum_{p = \abs{q_1 + q_2}}^{2 \, \ell} \tilde d^{[2]}_{p, q_1, q_2} \mbf G_2\left(\mbf u_{p, q_1, q_2}^{[2]}, \bs \varphi\right) \cdot \bs \psi \right.
                        \\
            			+ \left. \tilde d^{[3]}_{q_1, q_2, q_3} \mbf G_3\left(\bs \varphi, \bs \varphi, \bs \varphi\right) \cdot \bs \psi \vphantom{\sum_{p = \abs{q_1 + q_2}}^{2 \, \ell}}\right),
        			\end{multlined}
                \end{align*}
                with
                $$
                    \tilde d_{p, q_1, q_2}^{[2]} = \int_{-1}^1 P_\ell^{q_3}(x) \, P_p^{q_1 + q_2}(x) \, P_\ell^m(x) \, \dd x, 
                $$
                $$
                    \tilde d_{q_1, q_2, q_3}^{[3]} = \int_{-1}^1 P_\ell^{q_1}(x) \, P_\ell^{q_2}(x) \, P_\ell^{q_3}(x) \, P_\ell^m(x) \, \dd x,
                $$
                and $\mbf u_{p, q_1, q_2}^{[2]}$ solves
        		\begin{align*}
        			\mathbb B_p \, \mbf u_{p, q_1, q_2}^{[2]} = - d_{p, q_1, q_2}^{[2]} \, \mbf G_2\left(\bs \varphi,\bs \varphi\right),
                    \quad \text{with}
                    \s
        			d_{p, q_1, q_2}^{[2]} = \int_{-1}^1 P_\ell^{q_1}(x) \, P_\ell^{q_2}(x) \, P_p^{q_1 + q_2}(x) \, \dd x,
        		\end{align*}
        		for each set of integers $- \ell \leq p, m, q_1, q_2, q_3 \leq \ell$.
            \end{enumerate}
            In the above, all the coefficients $C_{1, 1}$, $C_{q_1, q_2, m}^{[2]}$, $C_{q_1, q_2, q_3, m}^{[3]}$ are evaluated at $\mu = \mu^*$.
    	\end{theorem}
        The proof of Theorem \ref{th:maintheorem} \rtext{, contained in Section \ref{sec:main},  provides a method for the explicit construction of each of the normal form coefficients for any particular bulk-surface PDE of the general form \eqref{generalsystem}. However, we note that the form of the amplitude equations \eqref{eq:nf2} or \eqref{eq:nf3} is identical to those that are derived for analogous bifurcations of reaction-diffusion equations on the sphere. There, it is also known that due to the symmetry of the sphere, many of the normal form coefficients must actually be zero or trivially related to others. This is spelled out in the next section.}

\section{The dynamics of the normal forms}  \label{sec:explanation}
    Equations \eqref{eq:nf2} and \eqref{eq:nf3} have precisely the same form as the normal forms derived for Turing bifurcations in systems with $O(3)$-symmetry. We note that such normal forms were previously derived for PDEs posed on the surface of a sphere.  Callahan (\cite{callahan} (see also the earlier results by Chossat {\it et al} \cite{chossat}) has shown how the coefficients of the normal forms for different values of $\ell$ are consequent on the symmetry of the problem. From this, he was able to apply the equivariant branching Lemma \cite{golubitsky} to find the criticality (direction of bifurcation curves) and stability of each different kind of bifurcating pattern up to $\ell = 6$. This leads to the conjecture in \cite{callahan} that only 10 topologically distinct stable steady patterns can arise for $\ell \leq 6$; see Fig.~\ref{fig:Callahan_spherical_harmonics}.
    \begin{figure}
        \centering
        \begin{tikzpicture}
            \node (image) at (0,0) {
                \includegraphics[width = 0.98 \textwidth]{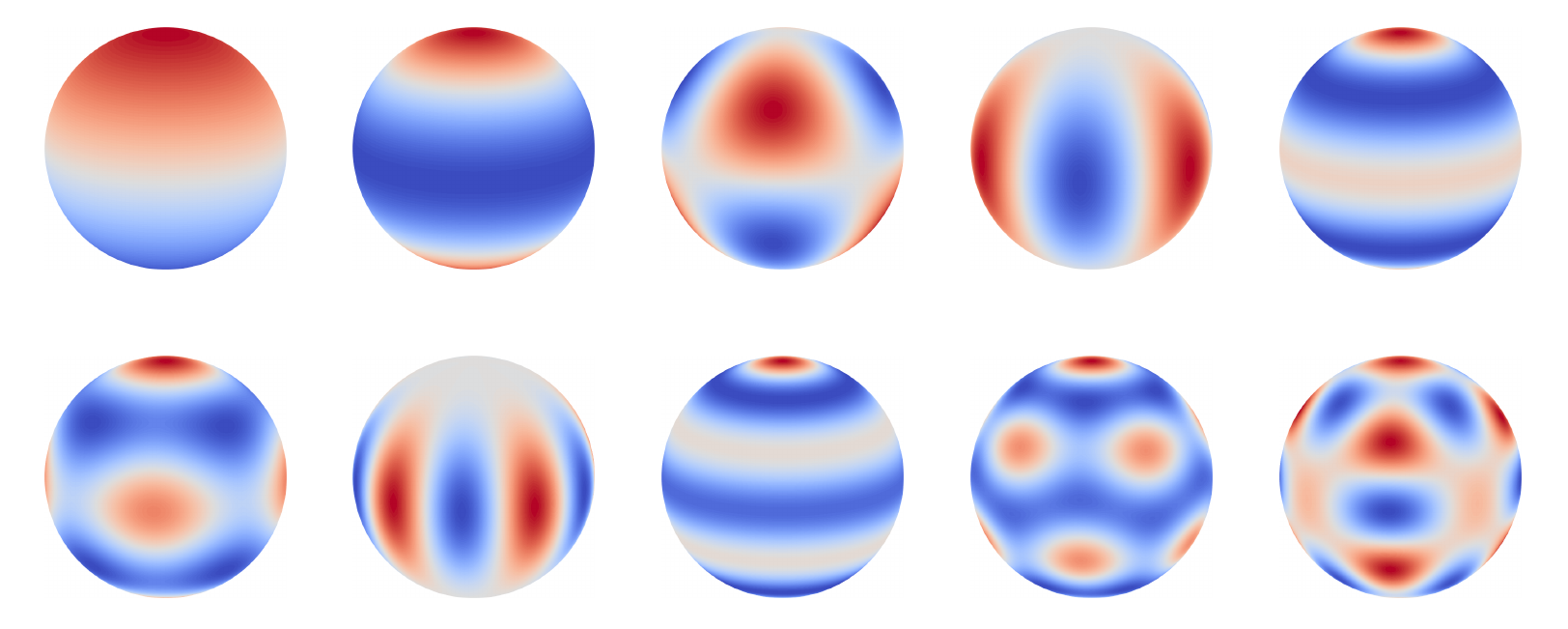}
            };
            \node (text) at (- 6.4, 3.2) {$(a)$};
            \node (text) at (- 3.25, 3.2) {$(b)$};
            \node (text) at (- 0.1, 3.2) {$(c)$};
            \node (text) at (3.15, 3.2) {$(d)$};
            \node (text) at (6.3, 3.2) {$(e)$};
            \node (text) at (- 6.4, - 0.1) {$(f)$};
            \node (text) at (- 3.25, - 0.1) {$(g)$};
            \node (text) at (- 0.1, - 0.1) {$(h)$};
            \node (text) at (3.15, - 0.1) {$(i)$};
            \node (text) at (6.3, - 0.1) {$(j)$};
        \end{tikzpicture}
        \caption{A collection of spherical harmonics and some of their linear combinations. Top row, from left to right: $Y_1^0(\theta, \phi)$,  $Y_2^0(\theta, \phi)$,  $Y_3^2(\theta, \phi)$, $Y_3^3(\theta, \phi)$, $Y_4^0(\theta, \phi)$. Bottom row, from left to right: $\sqrt{14} \, Y_4^0(\theta, \phi) + \sqrt{5} \, Y_4^4(\theta, \phi)$, $Y_5^5(\theta, \phi)$, $Y_6^0(\theta, \phi)$, $\sqrt{11} \, Y_6^0(\theta, \phi) + \sqrt{7} \, Y_6^5(\theta, \phi)$, $\sqrt{2} \, Y_6^0(\theta, \phi) + \sqrt{7} \, Y_6^4(\theta, \phi)$. These functions are reported in Figure 1 of \cite{callahan}, respectively, as Solution 1, 2, 4, 5, 6, 7, 9, 16, 17, and 18.}
        \label{fig:Callahan_spherical_harmonics}
    \end{figure}
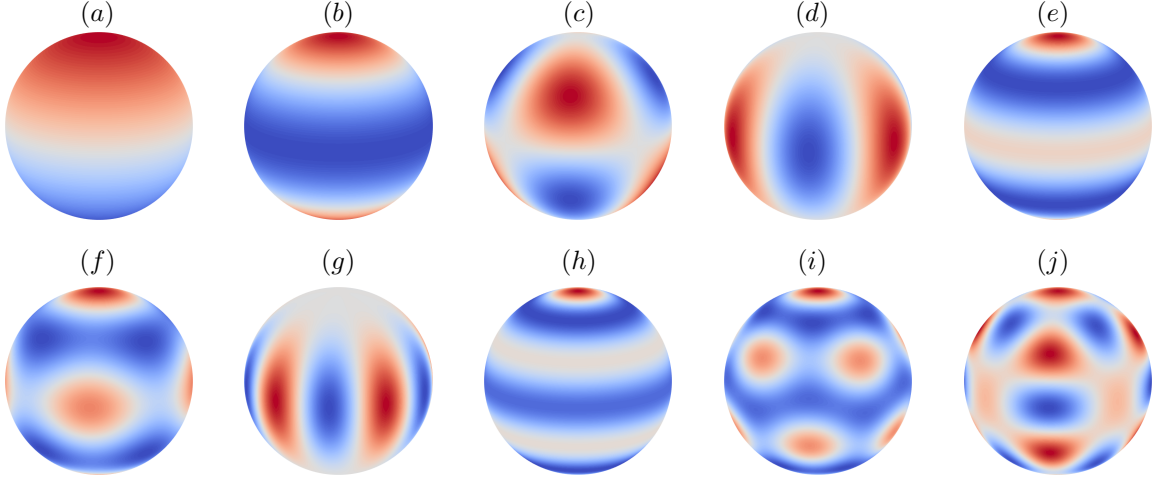
    Other information is available in \cite{matthews,matthews2} for higher $\ell$, including numerical indications of complex spatio-temporal dynamics. 
    
    Nevertheless, a full \rtext{description} of the dynamics of the normal forms, even for the lowest values of $\ell$ is, as far as we are aware, unknown. \rtext{For later use when we compute the amplitude equations for example models, we summarise here what is already known for some low values of $\ell$, giving a few new insights and highlighting what remains to be investigated.}

    \subsection{The case $\ell = 1$}
        Here,  $2 \ell + 1 = 3$, so there are three \rtext{components} in the amplitude vector $\rtext{\mbf A_1} = \left(A_{- 1}, A_0, A_1\right)^\intercal$, such that, as in definition \eqref{ansatz}, $A_{- 1} = - \ol{A_1}$. Following the theory developed in \cite{callahan}, we expect to find the specific form of the amplitude equations:
        \begin{equation}
            \begin{cases}
                 \dot A_0 = A_0 \left(\varepsilon + a \, A_0^2 - 2 \, a \, A_{- 1} \, A_1\right) + \mathcal O\left(\norm{\rtext{\mbf A_1}}^4\right), 
                \\
                \dot A_1 = A_1 \left(\varepsilon + a \, A_0^2 - 2 \, a \, A_{- 1} \, A_1\right) + \mathcal O\left(\norm{\rtext{\mbf A_1}}^4\right), 
            \end{cases} \label{examplel1}
        \end{equation}
        with the conjugate equation for $A_{- 1}$. Here, there are two real parameters to be determined, a single criticality coefficient $a \in \mathbb R$ and an unfolding parameter $\varepsilon = C_{1, 1} \left(\mu - \mu^*\right)$, defined so that the origin is stable whenever $\varepsilon < 0$.

        Neglecting higher-order terms, \eqref{examplel1} has a family of steady states defined by $A_0 = A_1 = A_{- 1} = 0$ and $\varepsilon + a \, A_0^2 - 2 \, a \, A_{- 1} \, A_1 = 0$, and at least two invariant subspaces given by $A_0 = 0$, and  $A_{- 1} = A_1 = 0$. Only the latter is analyzed carefully in \cite{callahan}, referred to as the `primary branch'. However, the former is also an invariant manifold and, when $A_0 = 0$, \eqref{examplel1} becomes $\dot A_1 = A_1 \left(\varepsilon + 2 \, a \, \abs{A_1}^2\right)$, whilst when $A_{- 1} = A_1 = 0$, \eqref{examplel1} becomes $\dot A_0 = A_0 \left(\varepsilon + a \, A_0^2\right)$. These systems are the normal forms of simple rotational pitchfork bifurcations, which turn out to be supercritical (resp.~subcritical) if and only if the third-order coefficient is negative (resp.~positive) (see e.g.~\cite{strogatz}).
        
        On the other hand, looking at the structure of \eqref{examplel1}, even for the planar reduction in which $A_{- 1} = - \ol{A_1}$, we have a one-dimensional continuum of steady states given by $\varepsilon + a \, A_0^2 - 2 \, a \, A_{- 1} \, A_1 = 0$. To analyze \rtext{the} stability of this family, let $z = \varepsilon + a \, A_0^2 - 2 \, a \, A_{-1} \, A_1 = \varepsilon + a \, A_0^2 + 2 \, a \, \abs{A_1}^2$. Therefore, \rtext{when differentiating this new variable with respect to $t$, we obtain}
        \begin{align*}
            \dot z &= 2 \, a \, A_0 \, \dot A_0 - 2 \, a \, \left(\dot A_{-1} \, A_1 + A_{-1} \, \dot A_1\right) \: = 2 \, a \, z \left(A_0^2 + 2 \, \abs{A_1}^2\right) + \mathcal O\left(\norm{\rtext{\mbf A_1}}^5\right)
            \\
            &= 2 \, z  (z - \eps) + \mathcal O\left(\norm{\rtext{\mbf A_1}}^5\right).
        \end{align*}
        \rtext{Thus}, this family of steady states is stable for $\eps > 0$, when the origin is unstable, but has a small basin of attraction in the direction of positive $z$. Thus, the full dynamics of \eqref{examplel1} appears to be non-trivial.

    \subsection{The case $\ell = 2$}
        We now have $2 \ell + 1 = 5$ components in the amplitude vector $\rtext{\mbf A_2} = \left(A_{- 2}, A_{- 1}, A_0, A_1, A_2\right)^\intercal$, with $A_{- 1} = - \ol{A_1}$, and $A_{- 2} = \ol{A_2}$. Again, by the theory developed in \cite{callahan}, the amplitude equations should have the form 
        \begin{equation}
            \begin{cases}
                \dot A_0 = A_0 \left(\varepsilon - 2 \, a \, A_0\right) + 2 \, a \, A_{- 1} \, A_1 + 4 \, a \, A_{- 2} \, A_2 + A_0 \, \mathcal P_2\rtext{\left(\mbf A_2\right)} + \mathcal O\left(\norm{\rtext{\mbf A_2}}^4\right), 
                \\
                \dot A_1 = A_1 \left(\varepsilon - 2 \, a \, A_0\right) + 2 \, \sqrt{6} \, a \, A_{- 1} \, A_2 + A_1 \, \mathcal P_2\rtext{\left(\mbf A_2\right)} + \mathcal O\left(\norm{\rtext{\mbf A_2}}^4\right), 
                \\
                \dot A_2 = A_2 \left(\varepsilon + 4 \, a \, A_0\right) - \sqrt{6} \, a \, A_1^2 + A_2 \, \mathcal P_2\rtext{\left(\mbf A_2\right)} + \mathcal O\left(\norm{\rtext{\mbf A_2}}^4\right),
            \end{cases}
            \label{examplel2}
        \end{equation}
        where $\mathcal P_2\rtext{\left(\mbf A_2\right)} = - b \, A_0^2 + 2 \, b A_{- 1} \, A_1 - 2 \, b \, A_{- 2} \, A_2$, in addition to the conjugate equations for $A_{- 1}$ and $A_{- 2}$. The equations now depend on two criticality parameters $a, b \in \mathbb R$, apart from the unfolding parameter $\varepsilon$. Finding all invariant subspaces is trickier than \eqref{examplel1}. However, some general results can still be obtained.
        \newtheorem{lemma}{Lemma}
        \begin{lemma} \label{Lemma1}
            When neglecting cubic and higher-order terms in \eqref{examplel2}, the following are invariant subspaces
            \begin{align}
                A_{- 1} = A_1 = 0 \qquad \text{and} \qquad \abs{A_2}^2 = \frac{3}{2} \, A_0^2, \label{firstinvman}
                \\
                A_0 = 0 \qquad \text{and} \qquad \abs{A_1}^2 = 2 \, \abs{A_2}^2, \label{secondinvman}
            \end{align}
            while
            \begin{align}
                4 \, A_0 \left(- \, A_0 + \frac{\varepsilon}{a}\right) + 24 \, \abs{A_2}^2 = \frac{\varepsilon^2}{a^2}, \qquad 2 \, A_0 \left(4 \, A_0 - \frac{\varepsilon}{a}\right) + 12 \, \abs{A_1}^2 = \frac{\varepsilon^2}{a^2}, \notag
                \\
                \text{and} \qquad A_2 \left(\varepsilon + 4 \, a \, A_0\right) - \sqrt{6} \, a \, A_1^2 = 0 \label{steadystateman}
            \end{align}
            give a manifold composed of steady states of  \eqref{examplel2}.
        \end{lemma}
        The proof of this lemma is given in Supplementary Materials \ref{appendix_proof_of_Lemma}.
    
        Rather than consider the full implications of these invariant manifolds, we can focus on the existence and stability of `primary' steady states. Setting $A_1 = A_2 = 0$ in System \eqref{examplel2}, we have $\dot A_0 = A_0 \left(\varepsilon - 2 \, a \, A_0 - b \, A_0^2\right)$. Thus, nontrivial states are given by $A_{0, \pm} = \frac{- a \pm \sqrt{a^2 + b \, \varepsilon}}{b}$. Now, we can easily delineate stability based on the signs of $a$ and $b$ \cite{callahan}. If $b < 0$, there is no stable small-amplitude pattern. If $b > 0$, then either $A_{0, +}$ or $A_{0, -}$ emerges stably from the bifurcation point, according to whether $a < 0$ or $a > 0$, respectively. See Figure \ref{fig:bifdiagl2}.

        \begin{figure}
            \centering
            \begin{tikzpicture}
                    \node[anchor=south west,inner sep=0] (image) at (0,0) {\includegraphics[width = \textwidth]{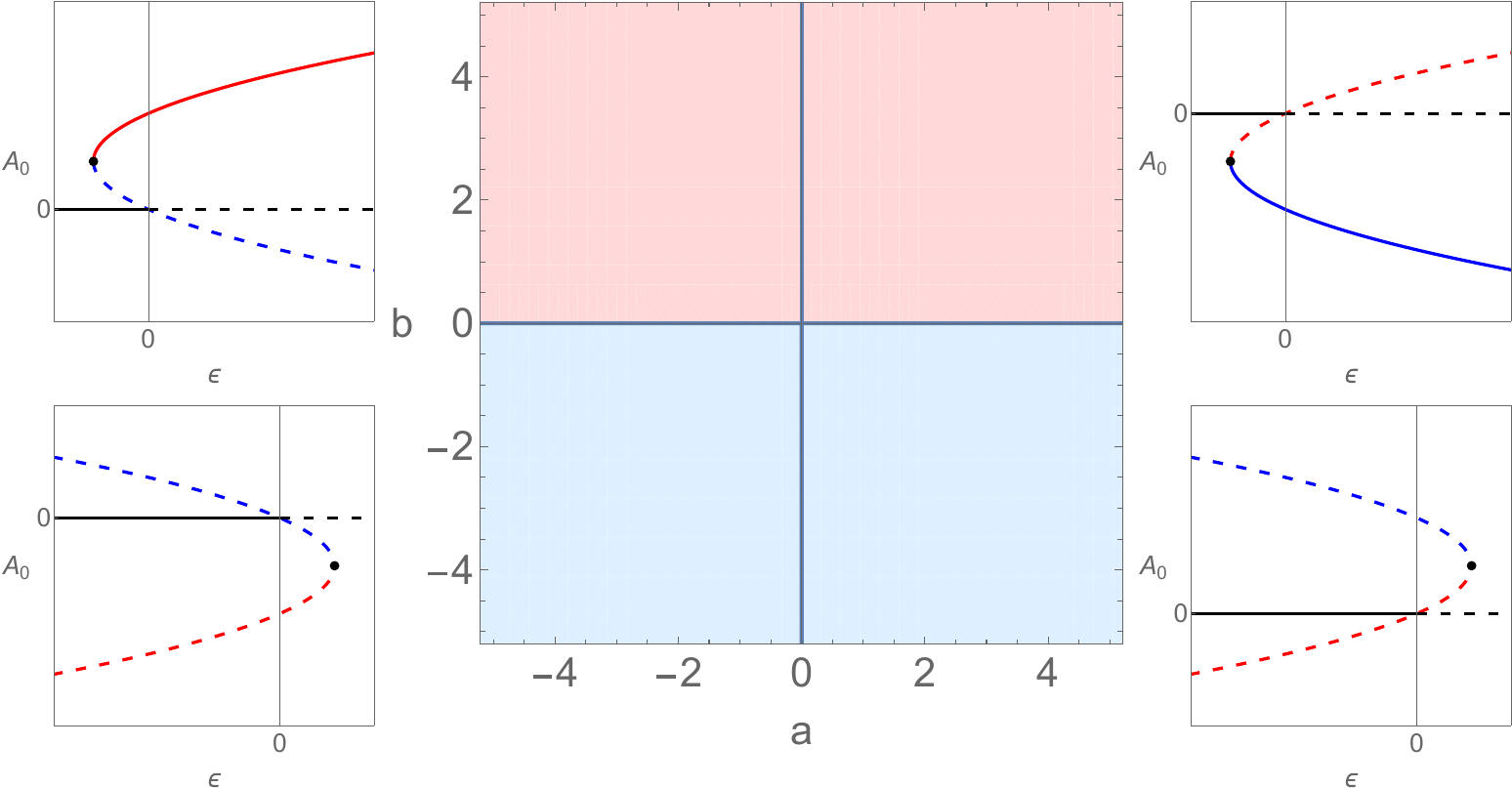}};
                    \node at (6.9, 6.9) {\circled{1}};
                    \node at (1.1, 5.7) {\small \circled{1}};
                    \node at (10.6, 6.9) {\circled{2}};
                    \node at (13.5, 7.95) {\small \circled{2}};
                    \node at (6.9, 3.2) {\circled{3}};
                    \node at (3.6, 3.6) {\small \circled{3}};
                    \node at (10.6, 3.2) {\circled{4}};
                    \node at (15.95, 1.3) {\small \circled{4}};
                \end{tikzpicture}
            \caption{Partial bifurcation diagram for the transcritical bifurcation of the origin of the normal form \eqref{examplel2} that arises when $\ell = 2$. Blue (respectively red) curves represent $A_{0, -}$ (resp. $A_{0, +}$). Moreover, dashed (resp.~continuous) lines represent unstable (resp.~stable) steady states.}
            \label{fig:bifdiagl2}
        \end{figure}
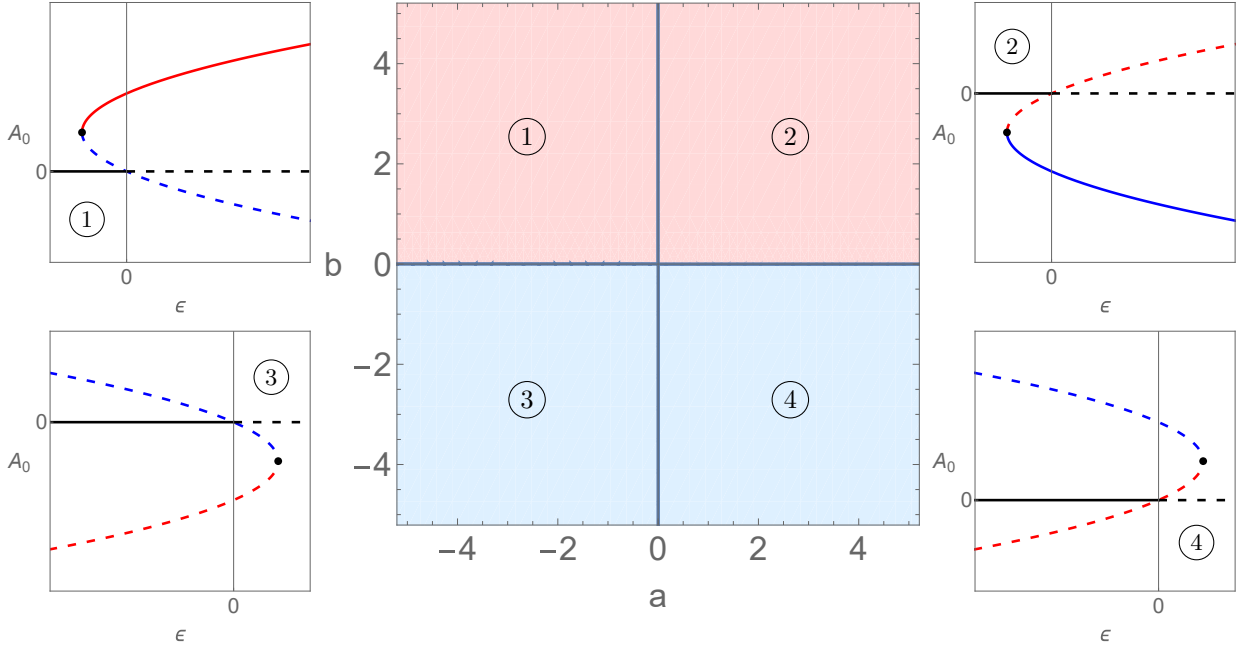
        
    \subsection{The case $\ell = 3$}
        Here, there are $2 \ell + 1 = 7$ components in the amplitude vector $\rtext{\mbf A_3} = \left(A_{- 3}, A_{- 2}, A_{- 1}, A_0, A_1, A_2, A_3\right)^\intercal$, with $A_{- 3} = - \ol{A_3}$, $A_{- 2} = \ol{A_2}$, and $A_{- 1} = - \ol{A_1}$. Furthermore, in accordance with the theory developed in \cite[Table B.2]{chossat} (with the correction of typos in the coefficients of $A_{- 2} A_1^2$ and $A_{- 2} A_0 A_3$ reported in \cite{callahan}), we expect the amplitude equations to take the form
        \begin{align}
            &\dot A_0 = A_0 \left(\varepsilon + \left(- 18 \, a + b\right) \, A_0^2 + \left(36 \, a - 2 \, b\right) \, A_{- 1} \, A_1 \right. \notag
            \\
            & \hspace{1cm} \left. + \left(-60 \, a + 2 \, b\right) \, A_{- 2} \, A_2 - 2 \, b \, A_{- 3} \, A_3\right) + 15 \, \sqrt{2} \, a \, \left(A_{- 3}\, A_1 \, A_2 + A_{- 2} \, A_{- 1} \, A_3\right) \notag
            \\
            & \hspace{1cm} + \sqrt{30} \, a \left(A_{- 2} \, A_1^2 + A_{- 1}^2 \, A_2\right) + \mathcal O\left(\norm{\rtext{\mbf A_3}}^4\right), \notag
            \\
            & \dot A_1 = A_1 \left(\varepsilon + \left(- 18 \, a + b\right) A_0^2 + \left(41 \, a - 2 \, b\right) A_{- 1} \, A_1 + \left(- 35 \, a + 2 \, b\right) A_{- 2} \, A_2 \right. \notag
            \\ 
            & \hspace{1cm} \left. + \left(15 \, a - 2 \, b\right) A_{- 3} \, A_3\right) \notag
            + 5 \, \sqrt{15} \, a \, A_{- 3} \, A_2 \, A_2 - 15 \, \sqrt{2} \, a \, A_{- 2} \, A_0 \, A_3
            \\
            & \hspace{1cm} + 6 \, \sqrt{15} \, a \, A_{- 1}^2 \, A_3 - 2 \, \sqrt{30} \, a \, A_{- 1} \, A_0 \, A_2 + \mathcal O\left(\norm{\rtext{\mbf A_3}}^4\right), \notag
            \\
            & \dot A_2 = A_2 \left(\varepsilon + \left(- 30 \, a + b\right) \, A_0^2 + \left(35 \, a - 2 \, b\right) A_{- 1} \, A_1 + \left(- 20 \, a + 2 \, b\right) A_{- 2} \, A_2 \right. \label{examplel3}
            \\ 
            &  \hspace{1cm} \left. + \left(45 \, a - 2 \, b\right) A_{- 3} \, A_3\right)  + \sqrt{30} \, a \, A_0 \, A_1^2 + 15 \, \sqrt{2} \, a \, A_{- 1} \, A_0 \, A_3  \notag
            \\
            & \hspace{1cm} - 10 \, \sqrt{15} \, a \, A_{- 2} \, A_1 \, A_3 + \mathcal O\left(\norm{\rtext{\mbf A_3}}^4\right), \notag
            \\
            & \dot A_3 = A_3 \left(\varepsilon + b \, A_0^2 + \left(15 \, a - 2 \, b\right) A_{- 1} \, A_1 + \left(- 45 \, a + 2 \, b\right) A_{- 2} \, A_2 \right. \notag
            \\
            & \hspace{1cm} \left. + \left(45 \, a - 2 \, b\right) A_{- 3} \, A_3\right)  + 5 \, \sqrt{15} \, a \, A_{- 1} \, A_2^2 - 15 \, \sqrt{2} \, a \, A_0 \, A_1 \, A_2 \notag
            \\
            & \hspace{1cm} + 2 \, \sqrt{15} \, a \, A_1^3 + \mathcal O\left(\norm{\rtext{\mbf A_3}}^4\right). \notag
        \end{align}
        We now have two cubic parameters $a, b \in \mathbb R$, in addition to the unfolding parameter $\varepsilon$. In \cite{callahan,chossat}, analysis can be found for some particular bifurcation branches according to different symmetry subgroups of $O(3)$. Specifically, according to the theory developed by these authors, three steady states become relevant to study. Those are the ones that follow $O(2)^-$, $\mbf O$, and $D_6^d$ symmetries, which correspond to the cases in which $A_{- 3} = A_{- 2} =  A_{- 1} = A_1 = A_2 = A_3 = 0$, $A_{- 3} =  A_{- 1} = A_0 = A_1 = A_3 = 0$, and $A_{- 2} =  A_{- 1} = A_0 = A_1 = A_2 = 0$, respectively \cite{callahan}. These steady states are defined, respectively, by $\varepsilon = (18 \, a - b) \, A_0^2$, $\varepsilon = 2 \, (10 \, a - b) \, A_{- 1} \, A_1$, and $\varepsilon = (45 \, a - 2 \, b) \, A_{- 3} \, A_3$. Furthermore, the steady-state associated with the symmetry $O(2)^-$ is always unstable, whilst the second one can be stable only if $10 \, a - b > 0$, and $a > 0$. Last but not least, the third one bifurcates stably only if $45 \, a - 2 \, b > 0$, and $a < 0$. Figure \ref{fig:bifdiagl3} shows the different regions one can find depending on the values of $a$ and $b$ in Equation \eqref{examplel3}. The full nonlinear dynamics of the 7-dimensional normal form is likely to be somewhat complicated.

        \begin{figure}
            \centering
            \begin{tikzpicture}
                \node[anchor=south west,inner sep=0] (image) at (0,0) {\includegraphics[width = \textwidth]{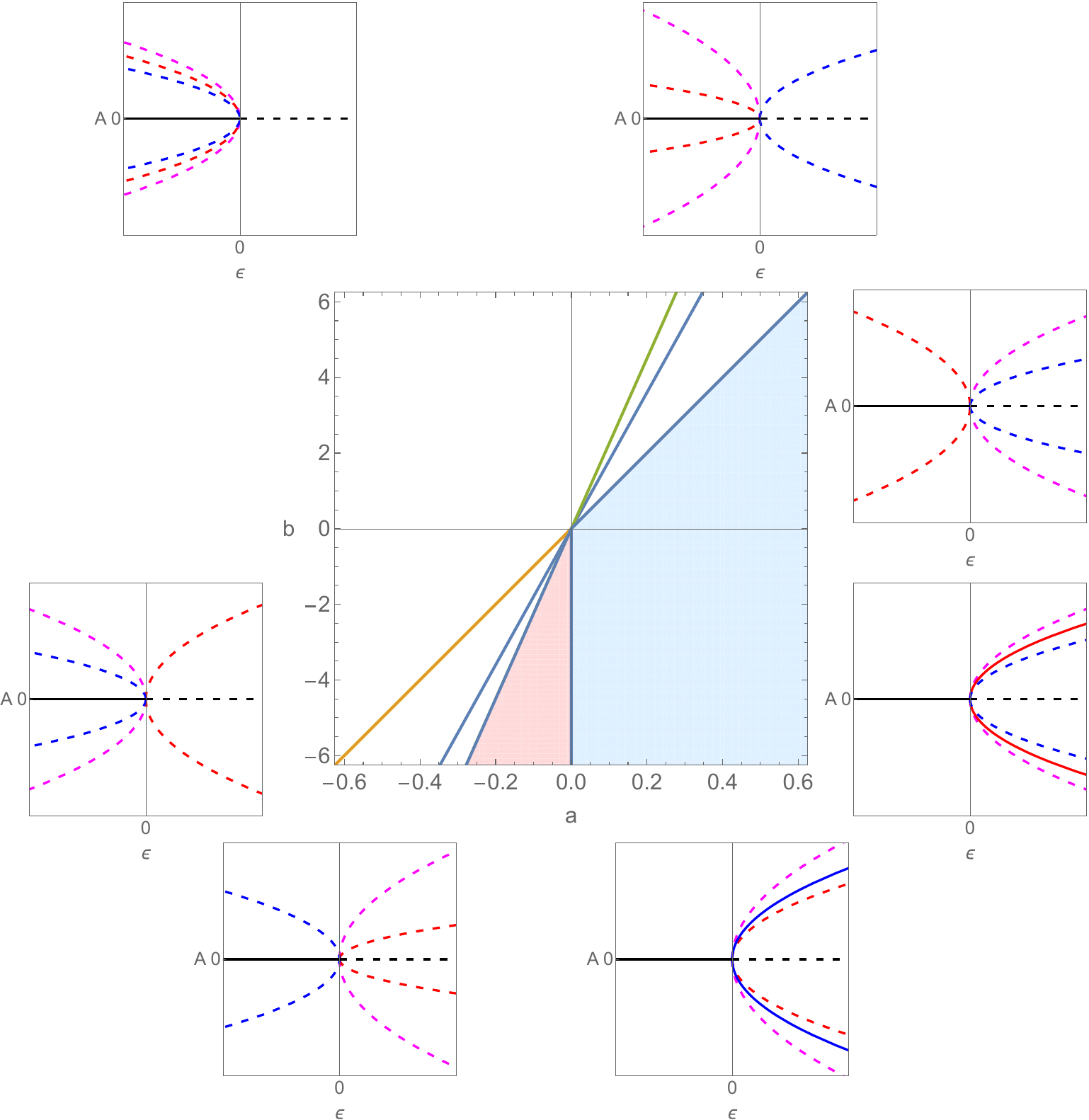}};
                \node at (10.5, 7.0) {\circled{1}};
                \node at (13.85, 7.25) {\small \circled{1}};
                \node at (10.7, 11.6) {\small \circled{2}};
                \node at (14.2, 12.2) {\small \circled{2}};
                \draw [-stealth](10.3, 12.2) -- (10.8, 13.3);
                \node at (12.3, 16.5) {\small \circled{3}};
                \node at (6.8, 10.7) {\small \circled{4}};
                \node at (4.5, 16.1) {\small \circled{4}};
                \node at (6.3, 5.9) {\small \circled{5}};
                \node at (2.9, 7.75) {\small \circled{5}};
                \draw [stealth-](6.35, 4.25) -- (7.0, 5.65);
                \node at (4.3, 3.8) {\small \circled{6}};
                \node at (8.1, 6.1) {\small \circled{7}};
                \node at (10.2, 3.4) {\small \circled{7}};
            \end{tikzpicture}
            \caption{Partial bifurcation diagram for the pitchfork bifurcation of the normal form \eqref{examplel3} that arises when $\ell = 3$. Magenta, red, and blue curves represent the steady states with symmetries $O(2)^-$, $\mbf O$, and $D_6^d$, respectively. Moreover, dashed (resp.~continuous) lines represent unstable (resp.~stable) steady states. In the $a-b$ bifurcation diagram, the straight lines are given by: $18 \, a - b = 0$, $ 10 \, a - b = 0$ and $45 \, a - 2 \, b = 0$.}
            \label{fig:bifdiagl3}
        \end{figure}
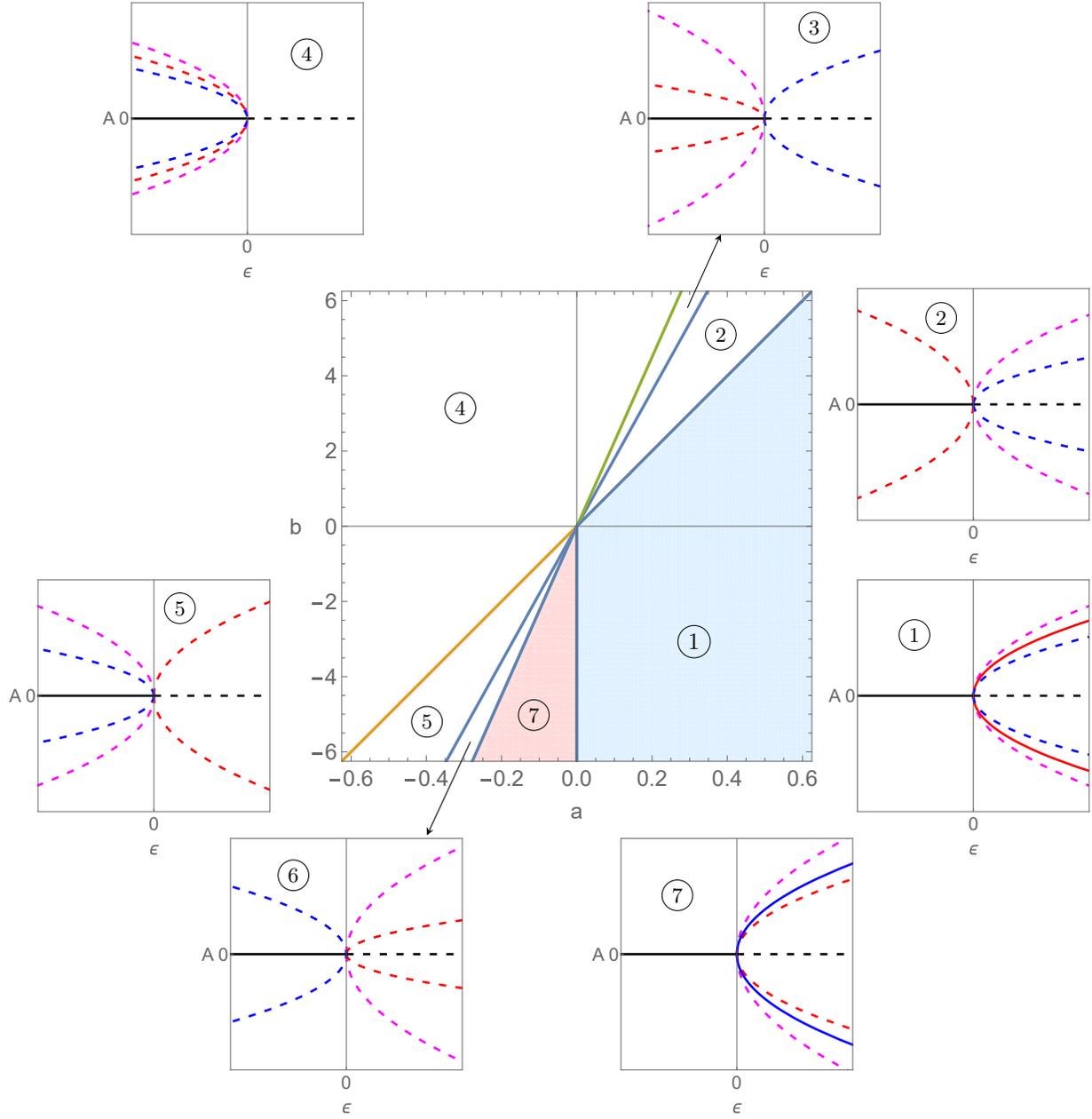
        
    \subsection{The case $\ell \geq 4$}
        When $\ell \geq 4$ we have that the amplitude vector has $2 \ell + 1$ components, $\rtext{\mbf A_\ell} = \left(A_{- \ell}, A_{-\ell + 1}, \ldots, A_{- 1}, A_0, A_1, \ldots, A_{\ell - 1}, A_\ell\right)^\intercal$, with $A_{- j} = (- 1)^j \, \ol{A_j}$, for each $j \in \{-\ell, - \ell + 1, \ldots, -1, 0, 1, \ldots, \ell - 1, \ell\}$, which requires analysis of dynamics in a $2\ell + 1$-dimensional phase space.
        
        For the specific case $\ell = 4$, \rtext{according to \cite[Table B.1]{chossat},} even just up \rtext{to} quadratic terms we have
        \begin{equation}
            \label{examplel4}
            \begin{cases}
                \dot A_0 = A_0 \left(\varepsilon - 18 \, a \, A_0\right) + 18 \, a \, A_{-1} \, A_1 + 22 \, a \, A_{-2} \, A_2
                \\
                \hspace{2cm} - 42 \, a \, A_{-3} \, A_3 -28 \, a \, A_{-4} \, A_4 
                 + \mathcal O\left(\norm{\rtext{\mbf A_4}}^3\right), 
                \\
                \dot A_1 = A_1 \left(\varepsilon - 18 \, a \, A_0\right) + 12 \, \sqrt{10} \, a \, A_{-1} \, A_2
                \\ 
                \hspace{2cm} - 2 \, \sqrt{70} \, a \, A_{-2} \, A_3 - 14 \, \sqrt{10} \, a \, A_{-3} \, A_4 
                 + \mathcal O\left(\norm{\rtext{\mbf A_4}}^3\right), 
                \\
                \dot A_2 = A_2 \left(\varepsilon + 22 \, a \, A_0\right) - 6 \, \sqrt{10} \, a \, A_1^2
                \\ 
                \hspace{2cm} + 2 \, \sqrt{70} \, a \, A_{-1} \, A_3 - 6 \, \sqrt{70} \, a \, A_{-2} \, A_4 
                 + \mathcal O\left(\norm{\rtext{\mbf A_4}}^3\right),
                \\
                \dot A_3 = A_3 \left(\varepsilon + 42 \, a \, A_0\right) - 14 \, \sqrt{10} \, a \, A_{-1} \, A_4 - 2 \, \sqrt{70} \, a \, A_1 \, A_2 + \mathcal O\left(\norm{\rtext{\mbf A_4}}^3\right), 
                \\
                \dot A_4 = A_4 \left(\varepsilon - 28 \, a \, A_0\right) + 14 \, \sqrt{10} \, a \, A_1 \, A_3 - 3 \, \sqrt{70} \, a \, A_2^2 + \mathcal O\left(\norm{\rtext{\mbf A_4}}^3\right).
            \end{cases}
        \end{equation}
        The full equations up to and including cubic terms are given in the Supplementary Materials \rtext{\ref{ap:order3ell=4}}, where the cubic terms depend on two real parameters, $b$ and $c$. Again, a complete analysis of such a 9-dimensional normal form is beyond the scope of this article.

        Callahan \cite{callahan} goes further and supplies expressions for all the coefficients up to  $\ell = 6$. He also establishes the general principles for deciding which steady states can bifurcate stably. Figure \ref{fig:Callahan_spherical_harmonics} shows the only primary steady states that Callahan argues can be stable for low amplitude, up to $\ell = 6$. The stability of these patterns depends explicitly on the values of the normal form coefficients, for which he derives specific conditions. In addition, Matthews \cite{matthews,matthews2} finds explicit conditions for certain dihedral-symmetry states with large $\ell$ to bifurcate stably.
        
        It is important here to understand one consequence of Theorem \ref{th:maintheorem} (as already shown in the context of $O(3)$-equivariant bifurcation theory in \cite{matthews2}); when $\ell$ is even and the second-order coefficients are different from zero, we need only to compute one scalar to obtain the leading-order unfolding of the transcritical bifurcation. On the other hand, when the second-order coefficients are zero, independently of the parity of $\ell$, the cubic terms depend on $\left \lfloor \ell/3\right \rfloor + 1$ independent coefficients that must be computed to determine the unfolding of the pitchfork bifurcation \cite{chossat,callahan}. The main contribution of the Theorem is to derive expressions that enable us to compute such coefficients explicitly for the class of BS-RDEs considered here, and thus make predictions about which small-amplitude patterns should appear under parameter variation.

\section{Calculation of the amplitude equations} \label{sec:main}
    We provide a proof of Theorem \ref{th:maintheorem} using \rtext{the adaptation of the method of multiple scales for computation of normal forms at bifurcations} introduced by Elphick \textit{et al}~\rtext{and adapted to the context of infinite-dimensional spatial operators by a number of authors \cite{tirapegui,paquin-lefebvre,Fahad,edgardodegenerate}. The method supposes that for each $\ell$, the normal form can be expressed using the ansatz \eqref{ansatz} as an ordinary differential equation (ODE) for the complex amplitude vectors $\mbf A_\ell$. The right-hand side of these ODEs is determined as an asymptotic expansion in a small book-keeping parameter that captures the powers of components of $|\mbf A_\ell|$. \rtext{The} coefficients of these expansion\rtext{s} are obtained order-by-order using the solvability conditions derived using the Fredholm alternative. The first three orders are dealt with in Subsections} \ref{sub:order1}, \ref{sub:order2}, and \ref{sub:order3} respectively. There, we suppose that all functions are evaluated precisely at the bifurcation point $\mu = \mu\rtext{^*}$. Subsection \ref{sec:unfolding} \rtext{then} deals with the unfolding of the bifurcation. That is, the computation of the leading-order coefficient of $\abs{\mu - \mu^*}$ in the normal form. 
 
    Suppose that, at the bifurcation point $\mu = \mu^*$, the ground state $\mbf P$ fulfills the conditions stated in Theorem \ref{th:maintheorem}. We proceed to build the normal form using ansatz \eqref{ansatz}.
    In particular, we shall prove that $\rtext{\mbf A_\ell}$ solves an equation of the form
	\begin{align}
		\partial_t \rtext{\mbf A_\ell} &= \sum_{m_1 = -\ell}^\ell \mbf h_{m_1}^{[1]}(A_{m_1}) + \sum_{-\ell \leq m_1, m_2\leq \ell} \mbf h_{m_1, m_2}^{[2]}(A_{m_1}, A_{m_2}) + \ldots, \label{eqA}
	\end{align}
	where each function $\mbf h_{m_1, \ldots, m_p}^{[p]}\left(A_{m_1}, \ldots, A_{m_p}\right)$ corresponds to a linear function in $\prod_{j = 1}^p A_{m_j}$. To find these functions, we use the ansatz
	\begin{align*}
		\begin{pmatrix}
			\mbf U
			\s 
			\mbf u
		\end{pmatrix} = \mbf P + \sum_{q_1 = -\ell}^\ell \mbf W_{q_1}^{[1]}(A_{q_1}) + \sum_{-\ell\leq q_1, q_2\leq \ell} \mbf W_{q_1, q_2}^{[2]}(A_{q_1}, A_{q_2}) + \ldots,
	\end{align*}
    where we need to find the vector functions $\mbf W_{q_1}^{[1]}, \mbf W_{q_1, q_2}^{[2]}, \ldots$ such that $\rtext{\mbf A_\ell}$ solves an equation of the form \eqref{eqA}. We proceed order by order.
    
	\subsection{Order 1} \label{sub:order1}
        At this order, \eqref{generalsystem} becomes
		\begin{align}
			\sum_{q_1 = -\ell}^\ell \partial_{A_{q_1}} \mbf W_{q_1}^{[1]} \sum_{m_1 = -\ell}^\ell \, h_{m_1, q_1}^{[1]} &= \sum_{q_1 = -\ell}^\ell \left(\jac \mbf f(\mbf P) + \mathbb D \, \bs \nabla^2\right) \mbf W_{q_1}^{[1]}, \label{firstordereq}
		\end{align}
		where $q_1$ in the sub-script of $h_{m_1, q_1}^{[1]}$ indicates the component of $\mbf h_{m_1}^{[1]}$. Therefore, \rtext{as our analysis is based on an ansatz and $\det\left(\mathbb B_\ell\right) = 0$, then we note that taking $h_{m_1, q_1}^{[1]}\equiv 0$ for every $- \ell \leq m_1, q_1 \leq \ell$ will let us define a nontrivial solution ---as explained below--- and keep carrying out our analysis. With this,} \eqref{firstordereq} becomes:
		\begin{align}
			\sum_{q_1 = - \ell}^\ell \left(\jac \mbf f(\mbf P) + \mathbb D \, \bs \nabla^2\right) \mbf W_{q_1}^{[1]} &= \mbf 0. \label{jacobiankernel}
		\end{align}
		Here, note that the sum on the left-hand side is taken over all the amplitudes, which multiply orthogonal functions. Therefore, we must set each component of the sum to zero independently. Thus, for each $q_1$, we must find $\mbf W_{q_1}^{[1]} = \left(\mbf U(r, \theta, \phi), \mbf u(\theta, \phi)\right)^\intercal$ such that
		\begin{align*}
			\left(\jac \mbf f(\mbf P) + \mathbb D \, \bs \nabla^2\right) \begin{pmatrix}
			    \mbf U
                \\
                \mbf u
			\end{pmatrix} &= \mbf 0.
		\end{align*}
		Now, note that the first $n$ equations can be written as $- B \, \mbf U + \mathbb D_U \, \nabla_\Omega^2 \, \mbf U = \mbf 0$, from where we can choose the following non-zero solution $\mbf U = i_\ell\left(\bs \omega \, r\right) \, P_\ell^m(\cos(\theta)) \, e^{im\phi} \, \bs{\hat \varphi}$, where $\bs{\hat \varphi}\in \mathbb R^n\setminus \{\mbf 0\}$ is a constant vector, and $i_\ell\left(\bs \omega \, r\right) = \diag\left(i_\ell\big(\omega_1 \, r\big), \ldots, i_\ell\big(\omega_n \, r\big)\right)$, with $i_\ell$ being a modified spherical Bessel function of the first kind. Therefore, if we take $\mbf u = P_\ell^m(\cos(\theta)) \, e^{im\phi} \, \bs \varphi$, the boundary conditions imply that $\bs{\hat \varphi} = S_\ell \, \bs \varphi$. Thus, the last $n$ equations of \eqref{jacobiankernel} become
		\begin{align*}
            \rtext{\mathbb B_\ell} \, \bs \varphi &= \mbf 0.
		\end{align*}
		This equation has a one-parameter family of nontrivial solutions as, by assumption, \rtext{$\det\left(\mathbb B_\ell\right) = 0$}.
		We can choose $\bs \varphi$ to be any non-zero vector in the kernel \rtext{of $\mathbb B_\ell$}. With this, we define the non-zero vector
		\begin{align*}
			\mbf W_{q_1}^{[1]} = \begin{pmatrix}
				\mbf W_{q_1, 1}^{[1]}
				\\[1.5ex]
				\mbf W_{q_1, 2}^{[1]}
			\end{pmatrix} = A_{q_1}\begin{pmatrix}
				i_\ell\left(\bs \omega \, r\right) \, \bs{\hat \varphi}
				\s 
				\bs \varphi
			\end{pmatrix} P_\ell^{q_1}(\cos(\theta)) \, e^{i q_1 \phi},
		\end{align*}
		where
		\begin{align*}
			\rtext{\mathbb B_\ell} \, \bs \varphi &= \mbf 0, \qquad \text{and} \qquad \bs{\hat \varphi} = S_\ell \, \bs \varphi.
		\end{align*}
    
	\subsection{Order 2} \label{sub:order2}
        At this order, \eqref{generalsystem} becomes
		\begin{align}
			\sum_{b = -\ell}^\ell \partial_{A_b} \mbf W_b^{[1]} \sum_{- \ell \leq q_1, q_2 \leq \ell} h_{q_1, q_2, b}^{[2]} &= \sum_{-\ell\leq q_1, q_2\leq \ell} \left(\jac \mbf f(\mbf P) + \mathbb D \, \bs \nabla^2\right)\mbf W_{q_1, q_2}^{[2]}	\notag
            \\
            & + \sum_{- \ell \leq q_1, q_2 \leq \ell} \begin{pmatrix}
				\mbf 0_n
				\s
				\mbf G_2\left(\mbf W_{q_1, 2}^{[1]}, \mbf W_{q_2, 2}^{[1]}\right)
			\end{pmatrix}, \label{secondordereq}
		\end{align}
        which is equivalent to
		\begin{multline*}
			\sum_{- \ell \leq q_1, q_2 \leq \ell}\left(\jac \mbf f(\mbf P) + \mathbb D \, \nabla^2\right)\mbf W_{q_1, q_2}^{[2]} = \sum_{b = - \ell}^\ell \partial_{A_b} \mbf W_b^{[1]} \sum_{- \ell \leq q_1, q_2 \leq \ell} h_{q_1, q_2, b}^{[2]}
			\s 
			- \sum_{- \ell \leq q_1, q_2 \leq \ell} A_{q_1} A_{q_2} \, P_\ell^{q_1}(\cos(\theta)) \, P_\ell^{q_2}(\cos(\theta)) \, e^{i\left(q_1 + q_2\right)\phi} \, \begin{pmatrix}
				\mbf 0_n
				\s 
				\mbf G_2\left(\bs \varphi, \bs \varphi\right)
			\end{pmatrix}.
		\end{multline*}
        Here, the subscript $b$ on $h$ refers to the coefficient of the corresponding quadratic term $A_{q_1} A_{q_2}$. 
  
        Before continuing, we need to choose an inner product to establish solvability conditions using the Fredholm alternative. We make the following natural choice:  suppose for $\mbf a_1$, $\mbf a_2$, $\mbf b_1$, $\mbf b_2\in \mathbb C^n$, $f_1, g_1\in \mathcal L^2\left(\Omega;\mathbb R\right)$, and $f_2, g_2\in \mathcal L^2\left(\Gamma;\mathbb R\right)$, then we define
		\begin{align}
            & \left \langle \begin{pmatrix}
				f_1 \, \mbf a_1
				\s 
				f_2 \, \mbf a_2
    		\end{pmatrix}, \begin{pmatrix}
    				g_1 \, \mbf b_1
    				\s 
    				g_2 \, \mbf b_2
    	    \end{pmatrix} \right \rangle := \frac{1}{2 \, \pi} \left(\iiint_\Omega f_1(x, y, z) \, \mbf a_1 \, { \cdot} \, \overline{g_1(x, y, z) \, \mbf b_1} \, \dd \Omega\right. \notag
            \\
            & \hspace{1cm} \left. + R^2 
            { \iint_\Gamma f_2(x,y,z) \, \mbf a_2 \cdot \, \overline{g_2(x,y,z) \, \mbf b_2} \, \dd \Gamma}\right) \notag
            \\
			& \hspace{1cm} = \frac{1}{2 \, \pi} \left(\int_0^R \int_0^{2\pi}\int_0^\pi f_1(r,\theta,\phi) \, \mbf a_1 \cdot \overline{g_1(r,\theta,\phi) \, \mbf b_1} \, r^2 \sin(\theta) \, \dd \theta \, \dd \phi \, \dd r\right.
			\notag \\
			& \hspace{1cm} \left. + R^2 \int_0^{2\pi}\int_0^\pi f_2(\theta,\phi) \, \mbf a_2 \cdot \overline{g_2(\theta,\phi) \, \mbf b_2} \, \sin(\theta) \, \dd \theta \, \dd \phi\right), \label{innerp}
		\end{align}
		where $\cdot$ denotes the  usual scalar product of $n$-component vectors.
		
        Now, we want to expand $P_\ell^{q_1}(\cos(\theta)) \, P_\ell^{q_2}(\cos(\theta))$ in terms of the family of functions $\left\{P_p^{q_1 + q_2}(\cos(\theta))\right\}_{p = \abs{q_1 + q_2}}^\infty$. To do this, consider the following result.
		\begin{lemma} \label{lemma:linearexpansion}
			For every $- \ell \leq q_1, q_2 \leq \ell$, the function $P_\ell^{q_1}(x) \, P_\ell^{q_2}(x)$ can be written as a linear combination of $\left\{P_p^{q_1 + q_2}(x)\right\}_{p = \abs{q_1 + q_2}}^{2 \ell}$.
		\end{lemma}
		\begin{proof}
			As the associated Legendre polynomials fulfill $P_\ell^{- m}(x) = (- 1)^m \, P_\ell^m(x)$, then we assume, without loss of generality, that $q_1, q_2 > 0$. By definition, for $\ell \in \mathbb Z$ with $0 \leq m \leq \ell$, $P_\ell^m(x) = \mathcal N_{\ell, m} \, \left(1 - x^2\right)^{m/2}\frac{\dd^{\ell + m}}{\dd x^{\ell + m}}\left(x^2 - 1\right)^\ell$, where $\mathcal N_{\ell, m} \in \mathbb R$ is a constant that normalizes the associated Legendre polynomial. Note that
	      \begin{align*}
				\frac{P_\ell^{q_1}(x) \, P_\ell^{q_2}(x)}{\left(1 - x^2\right)^{\left(q_1 + q_2\right)/2}} = \mathcal N_{\ell, q_1} \, \mathcal N_{\ell, q_2} \, \frac{\dd^{\ell + q_1}}{\dd x^{\ell + q_1}}\left(x^2 - 1\right)^\ell \, \frac{\dd^{\ell + q_2}}{\dd x^{\ell + q_2}}\left(x^2 - 1\right)^\ell,
			\end{align*}
			is a polynomial of degree $2 \ell - \left(q_1 + q_2\right)$, and
			\begin{align}
				\frac{P_p^{q_1 + q_2}(x)}{\left(1 - x^2\right)^{\left(q_1 + q_2\right)/2}} &= \mathcal N_{p, q_1 + q_2} \, \frac{\dd^{p + q_1 + q_2}}{\dd x^{p + q_1 + q_2}}\left(x^2 - 1\right)^p, \label{addition}
			\end{align}
			is a polynomial of degree $p - \left(q_1 + q_2\right)$, for $p \geq q_1 + q_2$. Moreover, note that $P_p^{q_1 + q_2} \equiv 0$ when $\abs{q_1 + q_2} > p$. Hence, this implies that the family of polynomials $\left\{\frac{P_p^{q_1 + q_2}(x)}{\left(1 - x^2\right)^{\left(q_1 + q_2\right)/2}}\right\}_{p = \abs{q_1 + q_2}}^{2 \ell}$ forms a basis of the space of polynomials of degree $2 \ell - \left(q_1 + q_2\right)$. Thus, $P_\ell^{q_1}(x) \, P_\ell^{q_2}(x)$ can be expanded as a linear combination of $\left\{P_p^{q_1 + q_2}(x)\right\}_{p = \abs{q_1 + q_2}}^{2 \ell}$, which concludes the proof.
		\end{proof}
		Now, thanks to Lemma \ref{lemma:linearexpansion} and the inner product \eqref{innerp}, we have
		\begin{align*}
			P_\ell^{q_1}(\cos(\theta)) \, P_\ell^{q_2}(\cos(\theta)) &= \sum_{p = \abs{q_1 + q_2}}^{2 \ell} d_{p, q_1, q_2}^{[2]} \, P_p^{q_1 + q_2}(\cos(\theta)),
		\end{align*}
		where
		\begin{align*}
			d_{p, q_1, q_2}^{[2]} &= \int_0^\pi P_\ell^{q_1}(\cos(\theta)) \, P_\ell^{q_2}(\cos(\theta)) \, P_p^{q_1 + q_2}(\cos(\theta))\sin(\theta) \, \dd \theta
            \\
            &= \int_{-1}^1 P_\ell^{q_1}(x) \, P_\ell^{q_2}(x) \, P_p^{q_1 + q_2}(x) \, \dd \theta,
		\end{align*}
        for every integer $\abs{q_1 + q_2} \leq p \leq 2 \ell$. Owing to the parity of the associated Legendre polynomials, it follows that $d_{p, q_1, q_2}^{[2]} = 0$ for every odd $p$.
		
		With this, equation \eqref{secondordereq} becomes
		\begin{multline}
			\sum_{-\ell\leq q_1, q_2\leq \ell} \left(\jac \mbf f(\mbf P) + \mathbb D \, \nabla^2\right)\mbf W_{q_1, q_2}^{[2]} = \sum_{b = -\ell}^\ell \partial_{A_b} \mbf W_b^{[1]} \sum_{- \ell \leq q_1, q_2 \leq \ell} h_{q_1, q_2, b}^{[2]}
			\s 
			- \sum_{-\ell\leq q_1, q_2\leq \ell} A_{q_1} A_{q_2} \sum_{p = \abs{q_1 + q_2}}^{2 \ell} d_{p, q_1, q_2}^{[2]} \, P_p^{q_1 + q_2}(\cos(\theta)) \, e^{i\left(q_1 + q_2\right)\phi} \, \begin{pmatrix}
				\mbf 0_n
				\s 
				\mbf G_2\left(\bs \varphi,\bs \varphi\right)
			\end{pmatrix}. \label{secondordersimplified}
		\end{multline}
		Here, there might be secular terms if $\ell$ is even \cite{matthews}, in which case there is a solvability condition that can be used to determine $\mbf h_{q_1, q_2}^{[2]}$. To find it, we can use the Fredholm alternative with the vector functions $\bs \psi_m^*$, for $m\in \mathbb Z$, that span
        \\
        $\ker \left(\left(\jac \mbf f(\mbf P) + \mathbb D \, \bs \nabla^2\right)^*\right)$. We have the following result.
        \begin{lemma}\label{lemma:adjoint}
			With the inner product defined by \eqref{innerp}, we have that
			\begin{align*}
				\left(\jac \mbf f(\mbf P) + \mathbb D \, \bs \nabla^2\right)^* = \begin{pmatrix}
					-B + \mathbb D_U \, \nabla_\Omega^2 & \mbf 0_{n\times n}
					\s 
					K \, \rtext{\left. \mathbbm I\right|_{\rtext{r = R}}} & - K + \jac \mbf g\left(\mbf u^*\right)^\intercal + \mathbb D_u^\intercal \, \nabla_\Gamma^2
				\end{pmatrix}.
			\end{align*}
		\end{lemma}
		\begin{proof}
			By the definition of the adjoint matrix
			\begin{multline*}
				\innerp{\left(\jac \mbf f(\mbf P) + \mathbb D \, \bs \nabla^2\right)^* \begin{pmatrix}
						\mbf{\hat U}
						\s 
						\mbf{\hat u}
					\end{pmatrix}, \begin{pmatrix}
						\mbf U
						\s 
						\mbf u
				\end{pmatrix}} = \frac{1}{2 \, \pi} \left(\iiint_\Omega \left(- B \, \mbf{\hat U}\cdot \mbf{\bar U} + \mathbb D_U \, \nabla_\Omega^2 \, \mbf{\hat U} \cdot \mbf{\bar U}\right) \, \dd \Omega\right.
				\s 
				\left. + R^2 \iint_\Gamma \left(K \left.\mbf{\hat U}\right|_{\rtext{r = R}} \cdot \mbf{\bar u} - K \, \mbf{\hat u} \cdot \mbf{\bar u} + \jac \mbf g\left(\mbf u^*\right)^\intercal \, \mbf{\hat u} \cdot \mbf{\bar u} + \mathbb D_u^\intercal \, \nabla_\Gamma^2 \, \mbf{\hat u} \cdot \mbf{\bar u}\right) \dd \Gamma\right),
			\end{multline*}
			and
			\begin{multline*}
				\innerp{\begin{pmatrix}
						\mbf{\hat U}
						\s 
						\mbf{\hat u}
					\end{pmatrix}, \left(\jac \mbf f(\mbf P) + \mathbb D \, \bs \nabla^2\right)\begin{pmatrix}
						\mbf U
						\s 
						\mbf u
				\end{pmatrix}} = \frac{1}{2 \, \pi} \left(\iiint_\Omega \left(- B \, \mbf{\bar U} \cdot \mbf{\hat U} + \mathbb D_U \, \nabla_\Omega^2 \, \mbf{\bar U} \cdot \mbf{\hat U}\right)\dd \Omega\right.
				\s 
				\left. + R^2 \iint_\Gamma \left(K \left. \mbf{\bar U}\right|_{\rtext{r = R}} \cdot \mbf{\hat u} - K \, \mbf{\bar u} \cdot \mbf{\hat u} + \jac \mbf g\left(\mbf u^*\right) \, \mbf{\bar u} \cdot \mbf{\hat u} + \mathbb D_u \, \nabla_\Gamma^2 \, \mbf{\bar u} \cdot \mbf{\hat u}\right) \dd \Gamma\right).
			\end{multline*}
			Therefore,
			\begin{multline*}
				\innerp{\left(\jac \mbf f(\mbf P) + \mathbb D \, \bs \nabla^2\right)^* \begin{pmatrix}
						\mbf{\hat U}
						\s 
						\mbf{\hat u}
					\end{pmatrix}, \begin{pmatrix}
						\mbf U
						\s 
						\mbf u
				\end{pmatrix}} - \innerp{\begin{pmatrix}
						\mbf{\hat U}
						\s 
						\mbf{\hat u}
					\end{pmatrix}, \left(\jac \mbf f(\mbf P) + \mathbb D \, \bs \nabla^2\right)\begin{pmatrix}
						\mbf U
						\s 
						\mbf u
				\end{pmatrix}}
				\s 
				= \frac{1}{2 \, \pi}\left(\iint_\Gamma \left(\mathbb D_U  \left. \frac{\dd \mbf{\hat U}}{\dd \nu} \right|_{\rtext{r = R}} \cdot \left. \mbf{\bar U} \right|_{\rtext{r = R}} - \mathbb D_U \left. \frac{\dd \mbf{\bar U}}{\dd \nu}\right|_{\rtext{r = R}} \cdot \left. \mbf{\hat U}\right|_{\rtext{r = R}}\right) \dd \Omega\right.
				\s 
				\left. + R^2 \iint_\Gamma \left(K \left. \mbf{\hat U}\right|_{\rtext{r = R}} \cdot \mbf{\bar u} - K \left. \mbf{\bar U}\right|_{\rtext{r = R}} \cdot \mbf{\hat u}\right) \dd \Gamma\right) = 0,
			\end{multline*}
			which proves the result.
		\end{proof}
        Now, for $m\in \mathbb Z$, $\bs \psi_m^*$ solves the following equation $\left(\jac \mbf f(\mbf P) + \mathbb D \, \bs \nabla^2\right)^* \bs \psi_m^*=\mbf 0$. Furthermore, by Lemma \ref{lemma:adjoint}, we know that $\bs \varphi_m^*$ has the following form\rtext{:}
        \\
        $\bs \psi_m^* = \begin{pmatrix}
            i_\ell\left(\bs \omega r\right) \, \bs{\hat \psi}
            \s 
            \bs \psi
        \end{pmatrix}P_\ell^m(\cos(\theta)) \, e^{im\phi}$, for each integer $- \ell \leq m \leq \ell$, where
		\begin{align*}
			\left(K \left(i_\ell\left(\bs \omega R\right) \, S_\ell - I \right) + \jac \mbf g\left(\mbf u^*\right)^\intercal - \frac{\ell (\ell + 1)}{R^2} \, \mathbb D_u^\intercal\right) \bs \psi &= \mbf 0, \qquad \text{and} \qquad \bs{\hat \psi}= S_\ell \, \bs \psi.
		\end{align*}
		Therefore, if we split equation \eqref{secondordereq} according to the amplitudes and apply the inner product with respect to a single function $\bs \psi_m^*$, we \rtext{obtain} the following expression
        \begin{multline*}
			\innerp{\left(\jac \mbf f(\mbf P) + \mathbb D \, \bs \nabla^2\right) \mbf W_{q_1, q_2}^{[2]}, \, \bs{\psi}_m^*} = \innerp{\begin{pmatrix}
                i_\ell\left(\bs \omega r\right) \, \bs{\hat \varphi}
                \s 
                \bs \varphi
            \end{pmatrix} P_\ell^m(\cos(\theta)) \, e^{i m \phi} \, h_{q_1, q_2, m}^{[2]}, \, \bs{\psi}_m^*}
			\s 
			- A_{q_1} A_{q_2} \, \innerp{\sum_{p = \abs{q_1 + q_2}}^{2 \ell} d_{p, q_1, q_2}^{[2]} \, P_p^{q_1 + q_2}(\cos(\theta)) \, e^{i\left(q_1 + q_2\right)\phi} \, \begin{pmatrix}
				\mbf 0_n
				\s 
				\mbf G_2\left(\bs \varphi,\bs \varphi\right)
			\end{pmatrix}, \, \bs{\psi}_m^*},
		\end{multline*}
		which can be expanded as
        \begin{align*}
            & \innerp{\left(\jac \mbf f(\mbf P) + \mathbb D  \bs \nabla^2\right) \mbf W_{q_1, q_2}^{[2]}, \bs \psi_m^*}
            \\
			& =
            \left(\int_0^R i_\ell(\bs \omega r) S_\ell \, \bs \varphi \cdot i_\ell(\bs \omega r)  S_\ell  \bs \psi \, r^2  \dd r + R^2  \bs \varphi \cdot \bs \psi\right)
			- A_{q_1} A_{q_2} \, R^2 \, \delta_{m, q_1 + q_2} \; 
            \times
            \\
            & \mbf G_2\left(\bs \varphi, \bs \varphi\right) \cdot \bs \psi \sum_{p = \abs{q_1 + q_2}}^{2 \ell} d_{p, q_1, q_2}^{[2]}  \int_0^\pi P_p^{q_1 + q_2}(\cos(\theta)) P_\ell^m(\cos(\theta)) \sin(\theta) \dd \theta.
        \end{align*}
		Therefore, when we set this expression to zero, we see that $h_{q_1, q_2, m}^{[2]} = C_{q_1, q_2, m}^{[2]} \, A_{q_1} A_{q_2}$, where
		\begin{align}
			C_{q_1, q_2, m}^{[2]} &= \frac{R^2}{I_\ell} \, \delta_{m, q_1 + q_2} \, d_{\ell, q_1, q_2}^{[2]} \, \mbf G_2\left(\bs \varphi, \bs \varphi\right) \cdot \bs \psi, \qquad \text{with}\label{secondordercoef}
			\\
			I_\ell &= \int_0^R i_\ell(\bs \omega r) \, S_\ell \, \bs \varphi \cdot i_\ell(\bs \omega r) \, S_\ell \, \bs \psi \, r^2 \, \dd r + R^2 \, \bs \varphi \cdot \bs \psi. \notag
		\end{align}
		These coefficients are sufficient to classify generic transcritical bifurcations that occur when $\ell$ is even. However, for odd $\ell$ or if the \rtext{quadratic} coefficient vanishes for even $\ell$, we need to go further to compute cubic terms.
        To simplify the computation of cubic terms, we assume that $h_{q_1, q_2, m}^{[2]} = 0$. Under this assumption, the solution to \eqref{secondordersimplified} that is orthogonal to $\ker\left(\jac \mbf f(\mbf P) + \, \mathbb D \, \bs \nabla^2\right)$ is given by
		\begin{align*}
			\mbf W_{q_1, q_2}^{[2]} &= \begin{pmatrix}
				\mbf W_{q_1, q_2, 1}^{[2]}
				\\[1.5ex]
				\mbf W_{q_1, q_2, 2}^{[2]}
			\end{pmatrix}
            \\
            &= A_{q_1} A_{q_2} \left(\sum_{p = \abs{q_1 + q_2}}^{2 \ell} \begin{pmatrix}
					i_p\left(\bs \omega r\right) \, \mbf U_{p, q_1, q_2}^{[2]}
					\s 
					\mbf u_{p, q_1, q_2}^{[2]}
				\end{pmatrix}P_p^{q_1 + q_2}(\cos(\theta)) \, e^{i\left(q_1 + q_2\right)\phi}\right),
		\end{align*}
		where $\mathbb B_p \, \mbf u_{p, q_1, q_2}^{[2]} = - d_{p, q_1, q_2}^{[2]} \, \mbf G_2\left(\bs \varphi,\bs \varphi\right)$, and $\mbf U_{p, q_1, q_2}^{[2]} = S_p \, \mbf u_{p, q_1, q_2}^{[2]}$, for each integer $\abs{q_1 + q_2} \leq p \leq 2 \ell$.
		
		Note that these systems of equations have a solution for every $\abs{q_1 + q_2} \leq p \leq 2 \ell$ since $d_{p, q_1, q_2}^{[2]} = 0$ for every odd $p$ and, if $\ell$ is odd, we stated in Theorem \ref{th:maintheorem} that $\mathbb B_p$ is invertible when $p$ is even. On the other hand, if $\ell$ is even, then although $\mathbb B_\ell$ is non-invertible, the right-hand side belongs to its image when $p = \ell$.
  
        With this, we are ready to proceed to the next order.
  
	\subsection{Order 3} \label{sub:order3}
        At this order, equation \eqref{generalsystem} becomes
		\begin{multline*}
			\sum_{b = - \ell}^\ell \partial_{A_b} \mbf W_b^{[1]} \, \sum_{- \ell \leq m_1, m_2, m_3 \leq \ell} h_{m_1, m_2, m_3, b}^{[3]}
			\s 
			= \sum_{- \ell \leq q_1, q_2, q_3 \leq \ell}\left(\jac \mbf f(\mbf P) + \mathbb D \, \bs \nabla^2\right)\mbf W_{q_1, q_2, q_3}^{[3]} + 2 \, \sum_{- \ell \leq q_1, q_2, q_3 \leq \ell} \begin{pmatrix}
				\mbf 0_n
				\s 
				\mbf G_2\left(\mbf W_{q_1, q_2, 2}^{[2]}, \mbf W_{q_3, 2}^{[1]}\right)
			\end{pmatrix}
			\s
			+ \sum_{- \ell \leq q_1, q_2, q_3 \leq \ell} \begin{pmatrix}
				\mbf 0_n
				\s 
				\mbf G_3\left(\mbf W_{q_1, 2}^{[1]}, \mbf W_{q_2, 2}^{[1]}, \mbf W_{q_3, 2}^{[1]}\right)
			\end{pmatrix}.
		\end{multline*}
		This entails
        \begin{align*}
            & \sum_{-\ell\leq q_1, q_2, q_3\leq \ell}\left(\jac \mbf f(\mbf P) + \mathbb D \, \bs \nabla^2\right)\mbf W_{q_1,q_2,q_3}^{[3]} = \sum_{b = -\ell}^\ell \partial_{A_b} \mbf W_b^{[1]} \, \sum_{-\ell \leq m_1, m_2, m_3 \leq \ell} h_{m_1, m_2, m_3, b}^{[3]}
			\\ 
			&\hspace{0.25cm} - 2  \sum_{-\ell\leq q_1, q_2, q_3 \leq \ell} A_{q_1} A_{q_2} A_{q_3}  P_\ell^{q_3}(\cos(\theta)) e^{i q_3 \phi}\sum_{p = \abs{q_1 + q_2}}^{2 \ell} P_p^{q_1 + q_2}(\cos(\theta)) \, e^{i \left(q_1 + q_2\right) \phi} \times
            \\
            &\hspace{1cm} \begin{pmatrix}
				\mbf 0_n
				\s 
			    \mbf G_2\left(\mbf u_{p, q_1, q_2}^{[2]}, \bs \varphi\right)
			\end{pmatrix}
				- \sum_{- \ell \leq q_1, q_2, q_3 \leq \ell} A_{q_1} A_{q_2} A_{q_3} \, P_\ell^{q_1}(\cos(\theta)) \times
                \\ 
                &\hspace{1.5cm} P_\ell^{q_2}(\cos(\theta)) \, P_\ell^{q_3}(\cos(\theta)) \, e^{i\left(q_1 + q_2 + q_3\right)\phi} \begin{pmatrix}
				\mbf 0_n
				\s
				\mbf G_3\left(\bs \varphi, \bs \varphi, \bs \varphi\right)
			\end{pmatrix}.
	  \end{align*}
		Therefore, by splitting this equation according to the amplitudes, we see that, for every $- \ell \leq q_1, q_2, q_3 \leq \ell$,
		\begin{multline}
			\left(\jac \mbf f(\mbf P) + \mathbb D \, \bs \nabla^2\right)\mbf W_{q_1, q_2, q_3}^{[3]} = \sum_{b = -\ell}^\ell \begin{pmatrix}
				i_\ell\left(\bs \omega \, r\right) \, \bs{\hat \varphi}
				\s 
				\bs \varphi
			\end{pmatrix} P_\ell^b(\cos(\theta)) \, e^{ib\phi} \, h_{q_1, q_2, q_3, b}^{[3]}
			\s 
			- 2 \, P_\ell^{q_3}(\cos(\theta)) \sum_{p = \abs{q_1 + q_2}}^{2 \ell} P_p^{q_1 + q_2}(\cos(\theta)) \, e^{i\left(q_1 + q_2 + q_3\right) \phi}\begin{pmatrix}
				\mbf 0_n
				\s 
				\mbf G_2\left(\mbf u_{p, q_1, q_2}^{[2]}, \bs \varphi\right)
			\end{pmatrix}
			\s 
			- P_\ell^{q_1}(\cos(\theta)) \, P_\ell^{q_2}(\cos(\theta)) \, P_\ell^{q_3}(\cos(\theta)) \, e^{i\left(q_1 + q_2 + q_3\right)\phi}\begin{pmatrix}
				\mbf 0_n
				\s 
				\mbf G_3\left(\bs \varphi, \bs \varphi, \bs \varphi\right)
			\end{pmatrix}. \label{thirdorderequation}
		\end{multline}
		Thus, when we apply the inner product with $\bs \psi_m^*$ in \eqref{thirdorderequation}, we \rtext{obtain} the following expression
        \begin{align*}
            & \innerp{\left(\jac \mbf f(\mbf P) + \mathbb D \bs \nabla^2\right)\mbf W_{q_1, q_2, q_3}^{[3]}, \bs \psi_m^*}
            \\
            & = \sum_{b = -\ell}^\ell \innerp{\begin{pmatrix}
				i_\ell\left(\bs \omega r\right) \, \bs{\hat \varphi}
				\s 
				\bs \varphi
			\end{pmatrix} P_\ell^b(\cos(\theta)) \, e^{ib\phi} \, h_{q_1, q_2, q_3, b}^{[3]}, \, \bs \psi_m^*}
			\\
			& - 2 \, \innerp{P_\ell^{q_3}(\cos(\theta)) \sum_{p = \abs{q_1 + q_2}}^{\rtext{2 \, \ell}} P_p^{q_1 + q_2}(\cos(\theta)) \, e^{i\left(q_1 + q_2 + q_3\right) \phi}\begin{pmatrix}
				\mbf 0_n
				\s 
				\mbf G_2\left(\mbf u_{p, q_1, q_2}^{[2]}, \bs \varphi\right)
			\end{pmatrix}, \, \bs \psi_m^*}
			\\
			& - \innerp{P_\ell^{q_1}(\cos(\theta)) \, P_\ell^{q_2}(\cos(\theta)) \, P_\ell^{q_3}(\cos(\theta)) \, e^{i\left(q_1 + q_2 + q_3\right)\phi}\begin{pmatrix}
				\mbf 0_n
				\s 
				\mbf G_3\left(\bs \varphi, \bs \varphi, \bs \varphi\right)
			\end{pmatrix}, \, \bs \psi_m^*},
        \end{align*}
		which can be expanded as
        \begin{align*}
            & \innerp{\left(\jac \mbf f(\mbf P) + \mathbb D \bs \nabla^2\right)\mbf W_{q_1, q_2, q_3}^{[3]}, \bs \psi_m^*}
            \\
            &= 
			\left(\int_0^R i_\ell\left(\bs \omega r\right) S_\ell \, \bs \varphi \cdot i_\ell\left(\bs \omega r\right) S_\ell \, \bs \psi \, r^2  \, \dd r + R^2 \, \bs \varphi \cdot \bs \psi\right)
            \, h_{q_1, q_2, q_3, m}^{[3]} - 2 \, \delta_{m, q_1 + q_2 + q_3} \, R^2 \times 
            \\ 
            & \sum_{p = \abs{q_1 + q_2}}^{2 \ell} \mbf G_2\left(\mbf u_{p, q_1, q_2}^{[2]}, \bs \varphi\right) \cdot \bs \psi \int_0^\pi P_\ell^{q_3}(\cos(\theta)) \, P_p^{q_1 + q_2}(\cos(\theta)) \, P_\ell^{m}(\cos(\theta)) \, \sin(\theta) \, \dd \theta
            \\
            & - \delta_{m, q_1 + q_2 + q_3}  R^2 \mbf G_3\left(\bs \varphi, \bs \varphi, \bs \varphi\right) \cdot \bs \psi \times
            \\ 
            &\int_0^\pi P_\ell^{q_1}(\cos(\theta))  P_\ell^{q_2}(\cos(\theta))  P_\ell^{q_3}(\cos(\theta))  P_\ell^{m}(\cos(\theta)) \sin(\theta) \dd \theta.   
        \end{align*}
		Finally, the solvability condition comes from equating this expression to zero, which leads us to conclude that $h_{q_1, q_2, q_3, m}^{[3]}\rtext{\left(\mbf A_\ell\right)} = C_{q_1, q_2, q_3, m}^{[3]} \, A_{q_1} A_{q_2} A_{q_3}$, where
		\begin{align}
            C_{q_1, q_2, q_3, m}^{[3]} &= \frac{R^2}{I_\ell} \, \delta_{m, q_1 + q_2 + q_3} \left(2 \,  \sum_{p = \abs{q_1 + q_2}}^{2 \ell} \mbf G_2\left(\mbf u_{p, q_1, q_2}^{[2]}, \bs \varphi\right) \cdot \bs \psi \int_{- 1}^1 P_\ell^{q_3}(x) \, P_p^{q_1 + q_2}(x) \times \right. \notag
            \\
    		& \left. P_\ell^m(x) \, \dd x \vphantom{\sum_{p = \abs{q_1 + q_2}}^{2 \ell}} + \mbf G_3\left(\bs \varphi, \bs \varphi, \bs \varphi\right) \cdot \bs \psi \int_{- 1}^1 P_\ell^{q_1}(x) \, P_\ell^{q_2}(x) \, P_\ell^{q_3}(x) \, P_\ell^m(x) \, \dd x\right), \label{thirdordercoef}
        \end{align}


    \subsection{Unfolding} \label{sec:unfolding}
        The calculation so far has been performed under the assumption that we are precisely at the bifurcation point $\mu = \mu^*$. Now, we recall that, for the first part of the proof, we set \rtext{$\mu = \mu^*$ t}o study what happens for small $\mu - \mu^*$; we now want $\rtext{\mbf A_\ell}$ to solve a differential equation of the following form
		\begin{align*}
			\partial_t \rtext{\mbf A_\ell} = \sum_{m_1 = - \ell}^\ell \mbf h_{m_1}^{[1, 1]}\left(A_{m_1}, \mu\right) &+ \sum_{m_1 = -\ell}^\ell \mbf h_{m_1}^{[1, 0]}\left(A_{m_1}, \mu\right)
            \\ 
            & \hspace{1cm} + \sum_{- \ell \leq m_1 \leq m_2 \leq \ell} \mbf h_{m_1, m_2}^{[2, 0]}\left(A_{m_1}, A_{m_2}, \mu\right) + \ldots,
		\end{align*}
        where the superscripts stand for the degree of each polynomial in $\rtext{\mbf A_\ell}$ and $\mu - \mu^*$, respectively. In particular, note that for each $p = 1, 2, 3$, $\mbf h_{m_1, m_2, \ldots}^{[p, 0]}\left(A_{m_1}, A_{m_2}, \ldots, \mu\right) = \mbf h_{m_1, m_2, \ldots}^{[p]}\left(A_{m_1}, A_{m_2}, \ldots\right)$, which corresponds to one of the terms we calculated in the preceding subsections. We extend this by considering a change of variables with the following form
		\begin{align*}
			\begin{pmatrix}
				\mbf U
				\s 
				\mbf u
			\end{pmatrix}  = \mbf P + \sum_{q_1 = -\ell}^\ell \mbf W_{q_1}^{[1, 1]}\left(A_{q_1}, \mu\right) & + \sum_{q_1 = -\ell}^\ell \mbf W_{q_1}^{[1, 0]}\left(A_{q_1}, \mu\right) \notag
            \\
            & \hspace{1cm} + \sum_{-\ell\leq q_1, q_2\leq \ell}\mbf W_{q_1, q_2}^{[2, 0]}\left(A_{q_1}, A_{q_2}, \mu\right) + \ldots, 
		\end{align*}
        where we only consider the leading-order unfolding term (see e.g.~\cite{Murdock}\rtext{)}, and the terms $\mbf W_{q_1, q_2, \ldots}^{[p, 0]}\left(A_{q_1}, A_{q_2}, \ldots, \mu\right)$ for $p = 1, 2, 3$, were already found. Thus, we only need to find $\mbf h_{m_1}^{[1, 1]}\left(A_{m_1}, \mu\right)$, which we do by considering only the terms in the Taylor expansion of \eqref{generalsystem} that have a product between $\mu$ and one amplitude $A_p$, for $- \ell\leq p\leq \ell$ and $\mu$. Carefully summing such terms leads to the following expression
		\begin{multline}
			\sum_{b = -\ell}^\ell \partial_{A_b} \mbf W_b^{[1, 0]} \sum_{m_1 = -\ell}^\ell h_{m_1, b}^{[1, 1]} = \sum_{q_1 = -\ell}^\ell\left(\jac \mbf f(\mbf P) + \mathbb D \, \bs \nabla^2\right)\mbf W_{q_1}^{[1,1]}
			\s 
			+ \rtext{\left(\mu - \mu^*\right) \,} \sum_{q_1 = -\ell}^\ell \begin{pmatrix}
				- \rtext{\dfrac{\partial B}{\partial \mu}} \mbf W_{q_1, 1}^{[1, 0]}
				\\[1.5ex]
				\mbf G_{1, 1}\left(\mbf W_{q_1, 2}^{[1, 0]}\right) - \rtext{\dfrac{\partial K}{\partial \mu}} \left(\mbf W_{q_1, 2}^{[1, 0]} - \left. \mbf W_{q_1, 1}^{[1, 0]}\right|_{\rtext{r = R}}\right)
			\end{pmatrix}
            \s
            - \frac{\ell (\ell + 1)}{R^2} \rtext{\, \left(\mu - \mu^*\right) \,} \frac{\partial \mathbb D}{\partial \mu} \sum_{q_1 = -\ell}^\ell \mbf W_{q_1}^{[1, 0]}. \label{cross-ordereq}
		\end{multline}
        Next, note that \eqref{cross-ordereq} can be split for each amplitude, which means that, for each $- \ell \leq q_1 \leq \ell$
	  \begin{multline*}
			\left(\jac \mbf f(\mbf P) + \mathbb D \, \bs \nabla^2\right)\mbf W_{q_1}^{[1, 1]} = \sum_{b = -\ell}^\ell \begin{pmatrix}
				i_\ell\left(\bs \omega r\right) \, \bs{\hat \varphi}
				\s 
				\bs \varphi
			\end{pmatrix} P_\ell^b(\cos(\theta)) \, e^{ib\phi} \, h_{q_1, b}^{[1, 1]}
			\s 
			- A_{q_1} \, \rtext{\left(\mu - \mu^*\right)} \, \begin{pmatrix}
				- \rtext{\dfrac{\partial B}{\partial \mu}} \, i_\ell\left(\bs \omega \, r\right) \, \bs{\hat \varphi}
				\\[1.5ex]
				\mbf G_{1, 1}\left(\bs \varphi\right) - \rtext{\dfrac{\partial K}{\partial \mu}} \, \left(\bs \varphi - i_\ell(\bs \omega \, R) 
                \, \bs{\hat \varphi}\right)
			\end{pmatrix} P_\ell^{q_1}(\cos(\theta)) \, e^{iq_1 \phi}
			\s 
			+ \frac{\ell (\ell + 1)}{R^2} \, A_{q_1} \, \rtext{\left(\mu - \mu^*\right)} \, \frac{\partial \mathbb D}{\partial \mu} \, \begin{pmatrix}
				i_\ell(\bs \omega \, r) \, \bs{\hat \varphi} 
				\s 
				\bs \varphi
			\end{pmatrix} P_\ell^{q_1} (\cos(\theta)) \, e^{i q_1 \phi}.
		\end{multline*}
        Again, in this case, we have a solvability condition that can be used to find $h_{q_1, b}^{[1, 1]}$ for every pair of integers $- \ell \leq b, q_1 \leq \ell$. In particular, we use the Fredholm alternative with the vector functions $\bs \psi_m^*$. Note that
		\begin{multline*}
			\innerp{\left(\jac \mbf f(\mbf P) + \mathbb D \, \bs \nabla^2\right)\mbf W_{q_1}^{[1,1]}, \, \bs \psi_m^*} = \innerp{\begin{pmatrix}
				i_\ell\left(\bs \omega \, r\right) \, \bs{\hat \varphi}
				\s 
				\bs \varphi
			\end{pmatrix} P_\ell^m(\cos(\theta)) \, e^{i m \phi} \, h_{q_1, m}^{[1, 1]}, \, \bs \psi_m^*}
			\s 
			- A_{q_1} \, \rtext{\left(\mu - \mu^*\right)} \, \innerp{\begin{pmatrix}
				- \rtext{\dfrac{\partial B}{\partial \mu}} \, i_\ell\left(\bs \omega \, r\right) \, \bs{\hat \varphi}
				\\[1.5ex]
				\mbf G_{1, 1}\left(\bs \varphi\right) - \rtext{\dfrac{\partial K}{\partial \mu}} \, \left(\bs \varphi - i_\ell(\bs \omega \, R) \, \bs{\hat \varphi}\right)
			\end{pmatrix} P_\ell^{q_1}(\cos(\theta)) \, e^{iq_1 \phi}, \, \bs \psi_m^*}
			\s 
			+ \frac{\ell (\ell + 1)}{R^2} \, A_{q_1} \, \rtext{\left(\mu - \mu^*\right) \,} \innerp{\frac{\partial \mathbb D}{\partial \mu} \, \begin{pmatrix}
				i_\ell\left(\bs \omega r\right) \, \bs{\hat \varphi}
				\s 
				\bs \varphi
			\end{pmatrix} P_\ell^{q_1}(\cos(\theta)) \, e^{i q_1 \phi}, \, \bs \psi_m^*}.
		\end{multline*}
		Equating this expression to zero, we \rtext{obtain} $h_{q_1, m}^{[1, 1]}\left(A_m, \mu\right) = \rtext{\left(\mu - \mu^*\right)} \, C_{1, 1} \, \delta_{q_1, m} \,  A_m$, where
        \begin{align}
            \label{crossordercoef}
			C_{1, 1} &= \frac{1}{I_\ell} \left(\int_0^R \left(- \frac{\partial B}{\partial \mu} \, i_\ell(\bs \omega \, r) \, S_\ell \, \bs \varphi\right) \cdot \left(i_\ell(\bs \omega \, r) \, S_\ell \, \bs \psi \right) r^2 \, \dd r\right.
            \\ 
            & + R^2 \, \left(\mbf G_{1, 1}(\bs \varphi) - \frac{\partial K}{\partial \mu} \, \left(\bs \varphi - i_\ell(\bs \omega \, R) \, \bs{\hat \varphi}\right)\right) \cdot \bs \psi - \frac{\ell (\ell + 1)}{R^2} \times \notag
            \\
			& \hspace{1cm} \left. \left(\int_0^R \left(\frac{\partial \mathbb D_U}{\partial \mu} \, i_\ell(\bs \omega \, r) \, S_\ell \, \bs \varphi\right) \cdot \left(i_\ell(\bs \omega \, r) \, S_\ell \, \bs \psi\right) r^2 \, \dd r + R^2 \, \left(\frac{\partial \mathbb D_u}{\partial \mu} \, \bs \varphi\right) \cdot \bs \psi\right)\right). \notag 
        \end{align}
        Hence, we \rtext{obtain} the desired result: the amplitude equation is given by
		\begin{align*}
			\partial_t \rtext{\mbf A_\ell} &= C_{1, 1} \, \left(\mu - \mu^*\right) \,  \rtext{\mbf A_\ell} + \sum_{- \ell \leq q_1, q_2\leq \ell} \mbf C_{q_1, q_2}^{[2]} \, A_{q_1} A_{q_2} + h.o.t.,
		\end{align*}
		if $\ell$ is even and the quadratic coefficients are non-zero, and
		\begin{align*}
			\partial_t \rtext{\mbf A_\ell} &= C_{1,1} \, \left(\mu - \mu^*\right) \,  \rtext{\mbf A_\ell} + \sum_{-\ell\leq q_1, q_2, q_3\leq \ell} \mbf C_{q_1, q_2, q_3}^{[3]} \, A_{q_1} A_{q_2} A_{q_3} + h.o.t.,
		\end{align*}
		if $\ell$ is odd or $\ell$ is even and the quadratic coefficients are zero.

\section{Numerical examples}\label{sec:examples}
    In this section, we present the numerical solutions of two different systems of type \eqref{generalsystem} and calculate the corresponding amplitude equations given in Theorem \ref{th:maintheorem}. As described in detail in Supplementary Materials \ref{sec:bsfem}, we use the Implicit-Explicit (IMEX) bulk-surface finite element method (see \cite{cusseddu2019, elliott2013finite, madzvamuse2015, macdonald2016computational,  madzvamuse2016bulk}, and references therein), which we implemented using the open-source finite element package, FEniCS \cite{alnaes2015fenics,logg2012automated}. All calculations of the coefficients of the amplitude equations were carried out with a Python code, which is made available \cite{amplitude-eq}. Furthermore, in the same repository, a Mathematica script is available to plot bifurcation diagrams and evaluate the coefficients of the amplitude equations at different bifurcation points.
  
    For each equation, we compute the amplitude equations and provide numerical justification of the predictions of the theory through \rtext{the} computation of the simplest patterned states. Exploring the complete dynamics of the amplitude equations and comparing it to the dynamics of full BS-RDE\rtext{s} is beyond the scope of the current study.
    
    \subsection{Brusselator} \label{sub:Brusselator} 
        As a first example, we consider a bulk-surface system with so-called Brusselator kinetics \cite{Brusselator}. See \cite{Fahad} and references therein for some of the consequences on localised pattern formation in the case of sub-critical Turing bifurcations on the real line. One motivation of the present work is to look for analogous structures in the bulk-surface framework. We assume that all nonlinear reactions take place on the surface, while we have linear decaying kinetics in the bulk. 
   
        Specifically, we consider the following bulk-surface system
        \begin{align*}
            \partial_t U &= - \sigma_1 \, U + D_1 \, \nabla^2_\Omega \, U, \quad \text{in } \Omega,
            \\
            \partial_t V &= - \sigma_2 \, V + D_2 \, \nabla^2_\Omega \, V, \quad \text{in } \Omega,
            \\
            \partial_t u &= \frac{\alpha}{\delta} - \gamma \, u + u^2 \, v - K_1\left(u - \left. U \right|_{\rtext{r = R}}\right) + \delta^2 \, \nabla_\Gamma^2 u, \quad \text{on } \Gamma,
            \s 
            \partial_t v &= (\gamma - 1) \, u - u^2 \, v - K_2 \left(v - \left. V\right|_{\rtext{r = R}}\right) + \nabla_\Gamma^2 v, \quad \text{on } \Gamma,
        \end{align*}
        subject to the Robin boundary conditions
        \begin{align*}
            D_1 \left. \frac{\partial U}{\partial \mbf{\hat n}}\right|_{\rtext{r = R}} = K_1 \left(u - \left. U\right|_{\rtext{r = R}}\right), \quad \text{and } \quad 
            D_2 \left. \frac{\partial V}{\partial \mbf{\hat n}}\right|_{\rtext{r = R}} = K_2 \left(v - \left. V\right|_{\rtext{r = R}}\right), \quad \text{on } \Gamma.
        \end{align*}		
        We start by looking for radially symmetric steady states. In order to use the same notation we have used throughout this article, let us define $\mbf U = (U, V)^\intercal$, $\mbf u = (u, v)^\intercal$, $\mbf g(\mbf u) = \begin{pmatrix}
            \dfrac{\alpha}{\delta} - \gamma \, u + u^2 \, v
            \\[1ex]
            (\gamma - 1) \, u - u^2 \, v
        \end{pmatrix}$, $K = \begin{pmatrix}
            K_1 & 0
            \\
            0 & K_2
        \end{pmatrix}$, and $\mathbb D_u = \begin{pmatrix}
            \delta^2 & 0
            \\
            0 & 1
        \end{pmatrix}$.
        Then, the radially-symmetric steady state is given by $\mbf P = \left(\mbf U^*(r),\mbf u^*\right)$, where $\mbf U^*(r) = i_0(\bs \omega r) \, S_0 \, \mbf u^*$, and $\mbf g(\mbf u) - K \left(I - i_0(\bs \omega r) \, S_0 \right) \mbf u^* = \mbf 0$. We fix the non-dimensional parameters given in Table \ref{tab:fixedparBrus}, and leave $\gamma$ and $\delta$ as free parameters.
        \begin{table}
            \centering
            \begin{tabular}{|c|c|c|c|c|c|c|c|c|}
                \hline
                Parameter & R & $\sigma_1$ & $\sigma_2$ & $\alpha$ & $K_1$ & $K_2$ & $D_1$ & $D_2$
                \\
                \hline
                Value & 1 & 1 & 0.1 & 1 & 0.005 & 1 & 1 & 1
                \\
                \hline
            \end{tabular}
            \caption{Parameter values fixed for the analysis of the Brusselator.}
            \label{tab:fixedparBrus}
        \end{table}
        Under this choice of parameter values, the system has only one radially symmetric steady state $(\mbf U^*(r), \mbf u^*)$. For this steady state, we get the bifurcation curves shown in Figure \ref{fig:simplebrusbif} in the $\gamma - \delta$ plane. These are the curves determined by the condition $\rtext{\det \left(\mathbb B_\ell\right)} = 0$ for different values of $\ell$, as shown in the figure. At each point, these curves give rise to the amplitude equations stated in Theorem \ref{th:maintheorem}. Note that the radially symmetric steady state is stable only in the region below and to the left of all the bifurcation curves shown in the figure. To pose the amplitude equations, we use $\varepsilon = C_{1, 1} \left(\mu - \mu^*\right)$ to simplify notation. We have computed the amplitude equations for the bifurcations occurring at each of the bifurcation points identified by solid dots in the figure. In each case, the coefficients were computed using the code in \cite{amplitude-eq} for which each nonlinear term can take an arbitrary value and they turned out to obey the same form as shown by the theory (see Sec. 3.1 and \cite{callahan}) without making that assumption {\it a priori}.
        
        \begin{figure}
            \begin{subfigure}[c]{0.48\textwidth}
                \centering
                \begin{tikzpicture}
                    \node[anchor=south west,inner sep=0] (image) at (0,0) {\includegraphics[width = \textwidth]{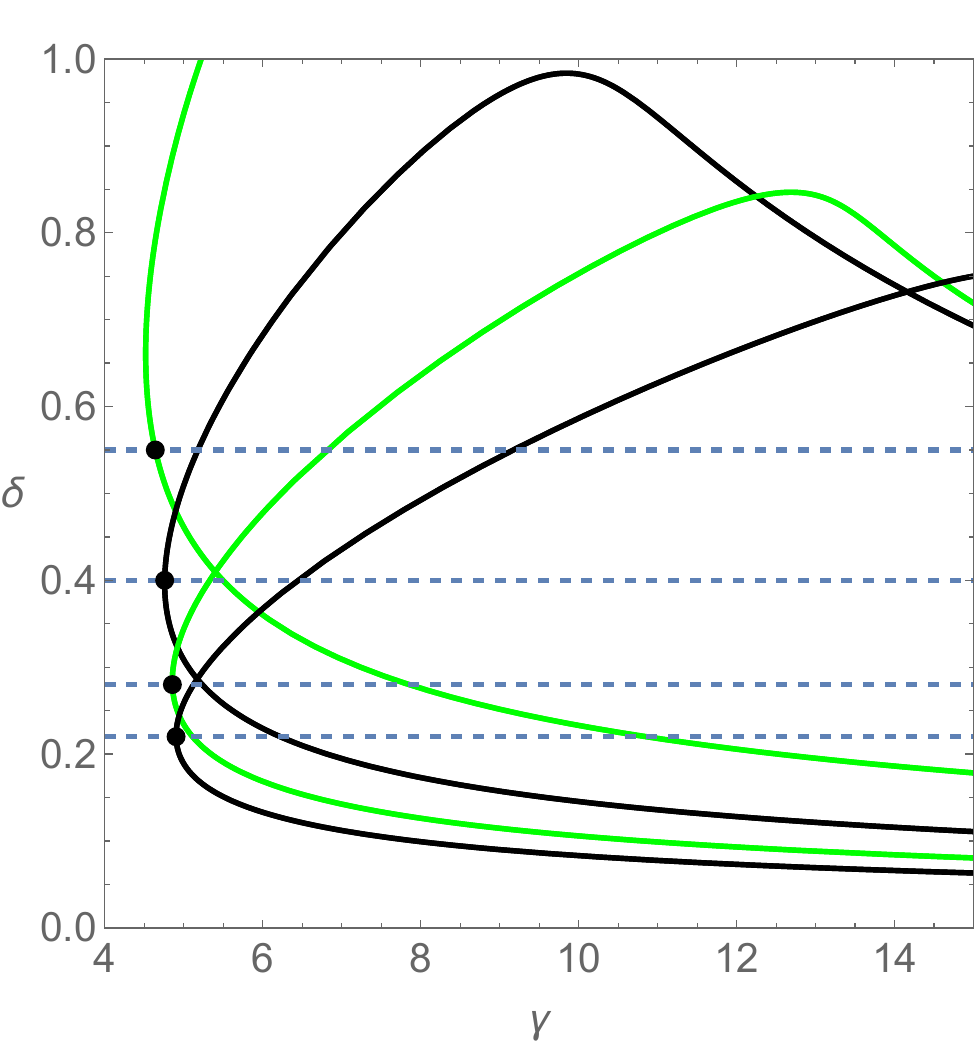}};
                    \node at (1.6, 8.3) {$\ell$ = 1};
                    \node at (4.6, 8.3) {$\ell$ = 2};
                    \node at (6.6, 7.2) {\small{$\ell$ = 3}};
                    \node at (6.5, 5.5) {\small{$\ell$ = 4}};
                \end{tikzpicture}
                \caption{}
                \label{fig:simplebrusbif}
            \end{subfigure}
            \hfill
            \begin{subfigure}[c]{0.48\textwidth}
                \centering
                \begin{tikzpicture}
                    \node[anchor=south west,inner sep=0] (image) at (0,0) {\includegraphics[width=\textwidth]{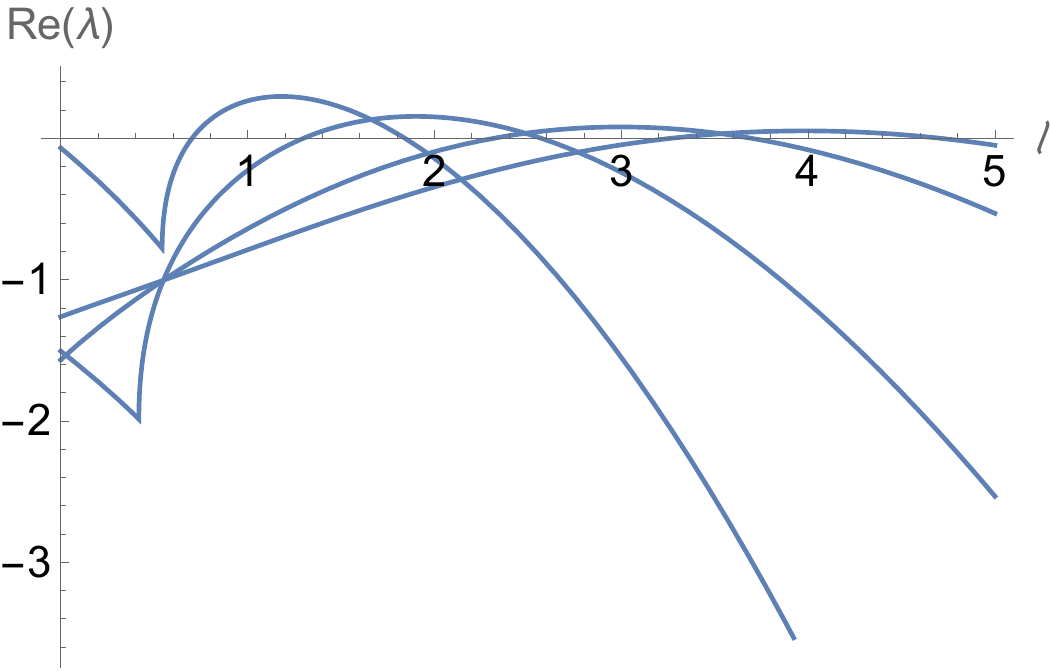}};
                \end{tikzpicture}
                \caption{}
                \label{fig:disprell=1}
            \end{subfigure}
            \caption{(a) Bifurcation curves at which $\mbf P$ undergoes a Turing bifurcation for different values of $\gamma$ and $\delta$ in the Brusselator, with the parameters in Table \ref{tab:fixedparBrus} fixed. Green (resp. black) curves are (Turing) pitchfork (resp. (Turing) transcritical) bifurcations in the corresponding amplitude equations. (b) Dispersion relation for the Brusselator when the parameters in Table \ref{tab:fixedparBrus} are fixed, $\gamma = 5$, and $\delta \in \{0.55, 0.4, 0.28, 0.22\}$. \rtext{Here, $\Re(\lambda)$ represents the real part of the eigenvalue of $\mathbb B_\ell$ with \rtext{the largest} real part, and} the unstable value of $\ell$ increases as $\delta$ decreases, according to panel (a).}
        \end{figure}
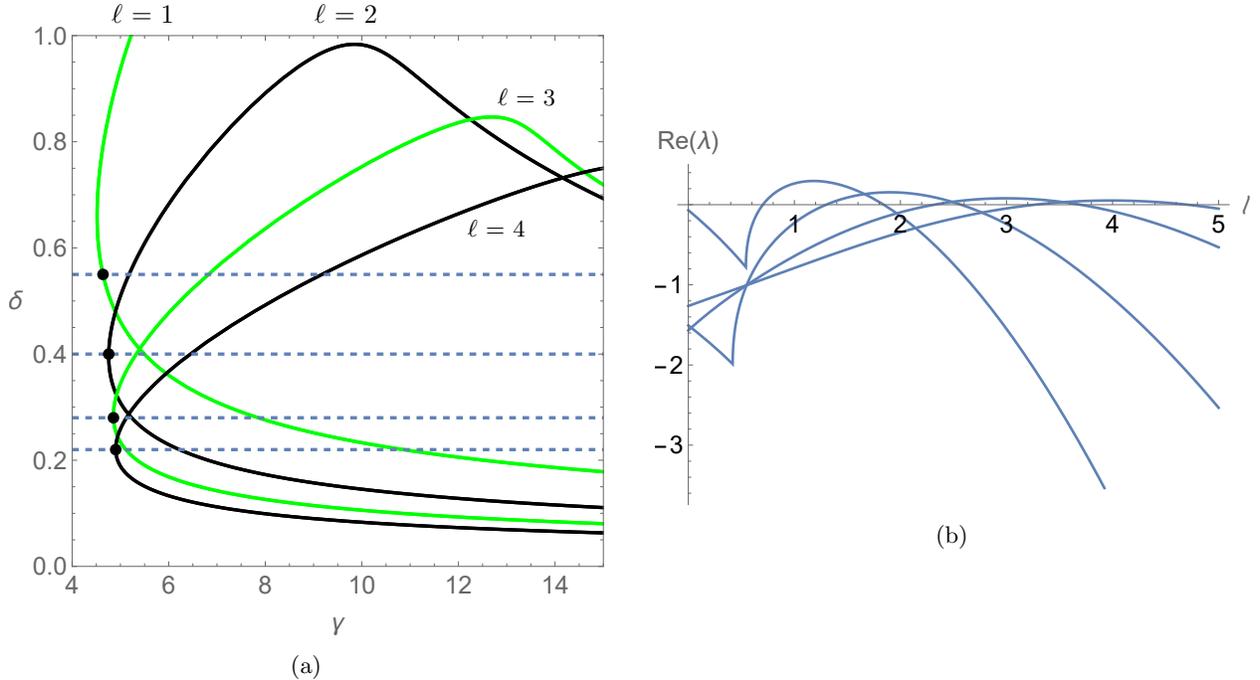

        \paragraph{The case $\ell = 1$} When $(\gamma, \delta) = (4.639433, 0.55)$, $\mbf P$ goes through a (Turing) pitchfork bifurcation for $\ell = 1$. The amplitude equations were computed to take the form given in \eqref{examplel1}, with  $a = - 0.380329$ and $C_{1, 1} = 0.730658$ for $\mu = \gamma$. The manifold of steady states analysed in Sec.~\ref{sec:explanation} is given by
        \begin{align}
            \varepsilon = 0.380329 \left(A_0^2 + 2 \, \abs{A_1}^2\right). \label{invman}
        \end{align}
        The fact that $a < 0$ implies that this family is stable. Therefore, because of the study in \cite{callahan} and the equivariant branching lemma \cite{golubitsky}, we know there must exist a stable small-amplitude pattern right after the bifurcation point. 
      
        Now, if we set $(\gamma, \delta) = (5, 0.55)$, then we get the dispersion relation shown in Figure \ref{fig:disprell=1} that has $\ell = 1$ as an unstable mode. Furthermore, after setting a random initial condition close to the steady state and integrating the system {up to $t = 200$}, we find the pattern shown in the top-left panel of Figure \ref{fig:Brusselator_l=1,2,3,4}, which corresponds to solution $(a)$ in Figure \ref{fig:Callahan_spherical_harmonics}.
        \begin{figure}
            \centering
            \begin{subfigure}[b]{\textwidth}
                \centering
                \includegraphics[width=0.8\textwidth]{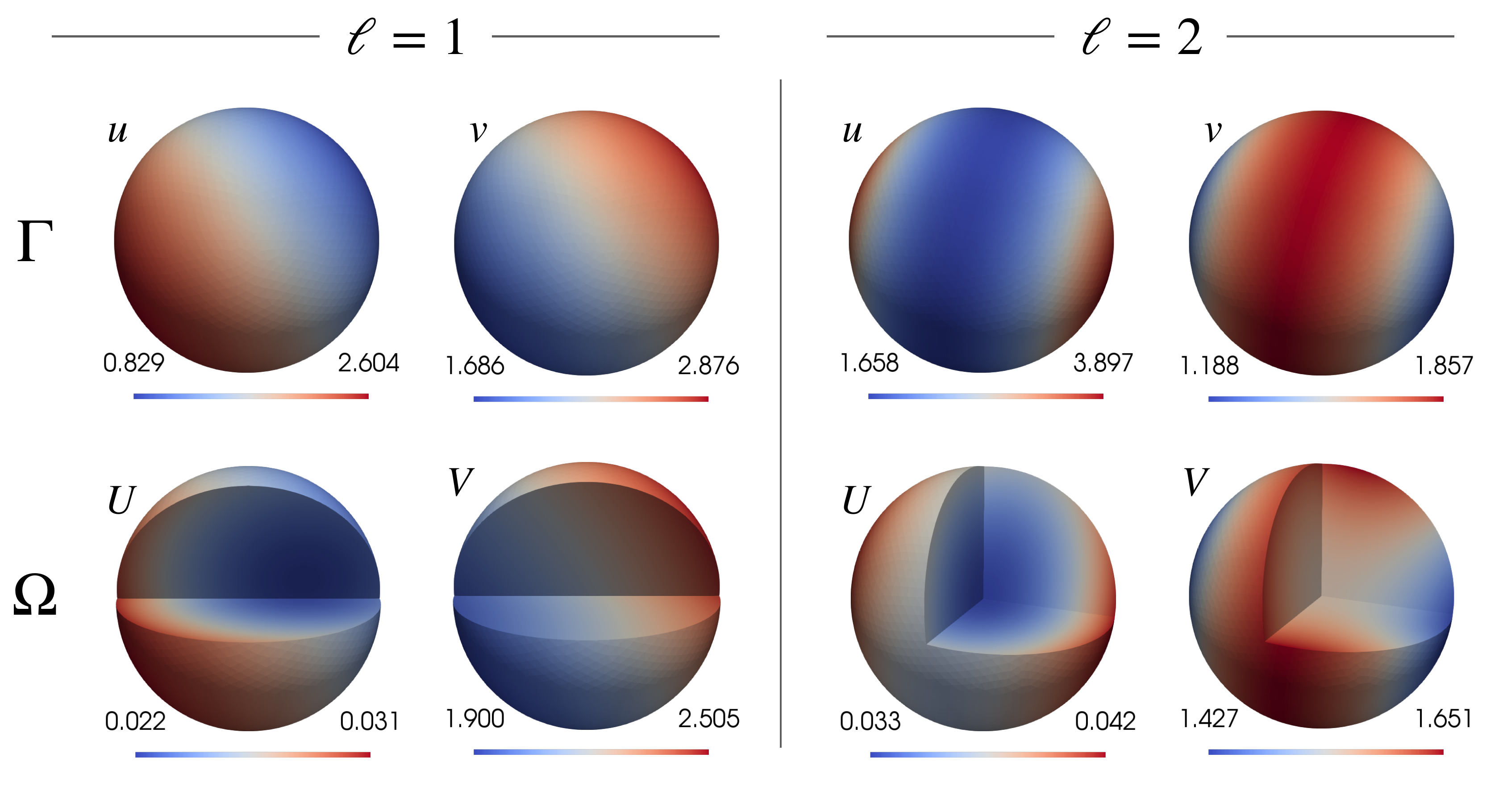}
            \end{subfigure}
            \hfill
            \begin{subfigure}[b]{\textwidth}
                \centering
                \includegraphics[width=0.8\textwidth]{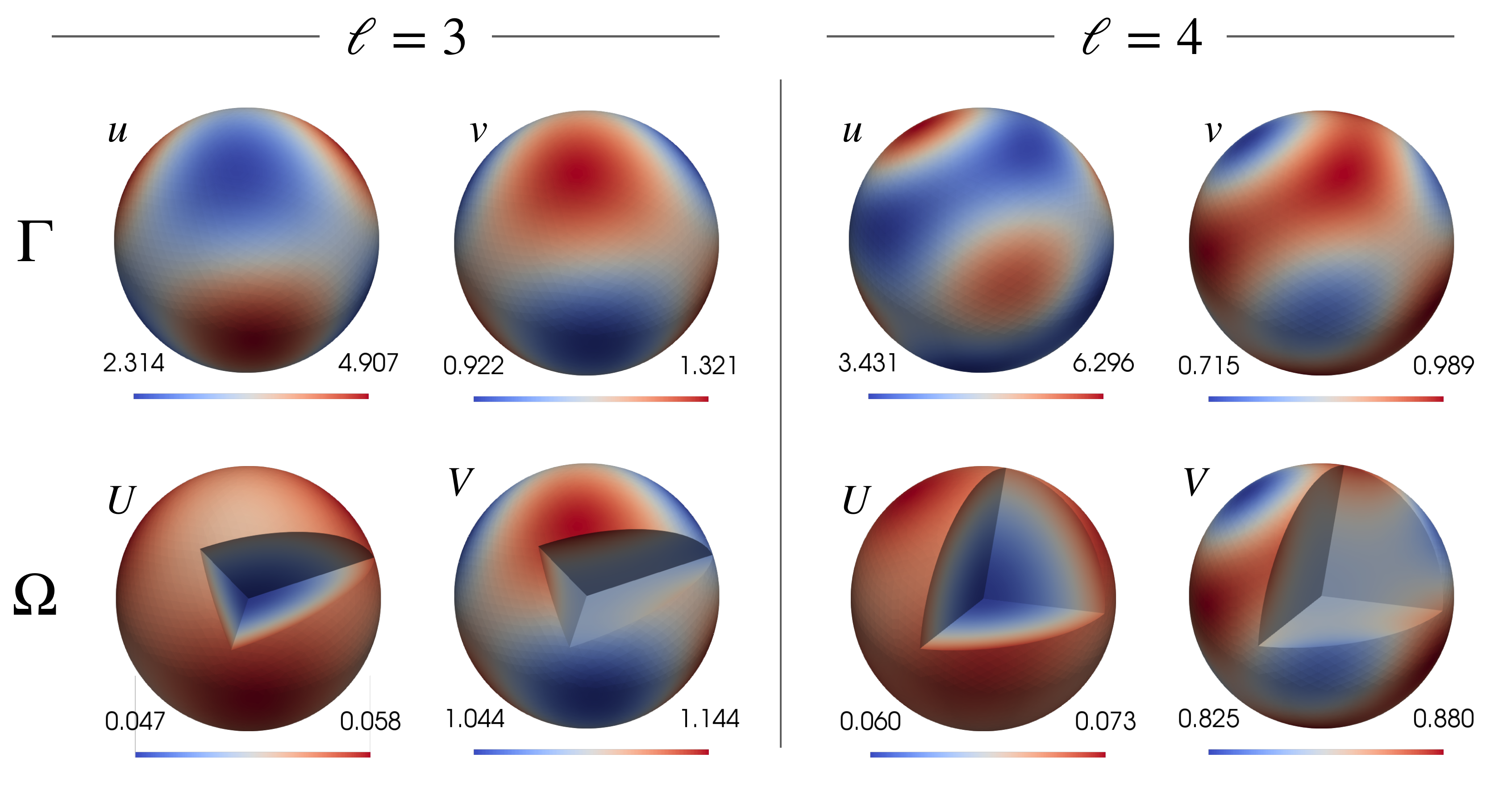}
            \end{subfigure}
            \caption{The numerical solutions $(u, v, U, V)$ of the Brusselator system, for four different parameter sets, corresponding to the choices $\ell = 1, 2, 3, 4$, at the corresponding final times, respectively, $t = 200$; $t = 200$; $t = 400$ and $t = 600$.}
            \label{fig:Brusselator_l=1,2,3,4}
        \end{figure}

        \paragraph{The case $\ell = 2$} When $(\gamma, \delta) = (4.758051, 0.4)$, $\mbf P$ goes through a Turing transcritical bifurcation for $\ell = 2$. Here, the amplitude equations are given by \eqref{examplel2} with $a = - 0.0132424$ and $C_{1, 1} = 0.620246$ for $\mu = \gamma$. As we can see, $a$ has a small magnitude. In fact, there is a nearby codimension-two point for which $a = 0$ at
        \\
        $(\gamma, \delta) = (4.76748, 0.376823)$. At such point, we find $b = 0.173063 > 0$. Therefore, the point marked in Figure \ref{fig:simplebrusbif} for $\ell = 2$ belongs to Region 1 in the diagram shown in Figure \ref{fig:bifdiagl2}. Thus, we expect to see a stable finite-amplitude steady state just after the bifurcation. Once again, we note that all coefficients were obtained using the code in \cite{amplitude-eq}. They follow the same structure as in \cite{callahan} for the parameter values $a$ and $b$ provided.
  
        After setting $(\gamma, \delta) = (5, 0.4)$, we get the dispersion relation shown in the top-right panel of Figure \ref{fig:disprell=1} with $\ell = 2$ as the unstable mode. After integrating the system up to $t = 200$, with a random initial condition close to the steady state, we obtained the pattern shown in the top-right panel of Figure \ref{fig:Brusselator_l=1,2,3,4}. Indeed, the transient numerical simulation leading to this stable pattern was found to show evidence of an initial jump to an intermediate weakly unstable pattern of smaller amplitude. See Supplementary Materials for graphs of transient simulations. Note that the stable pattern comprises two spots on the sides of the sphere together with a stripe in the middle. This corresponds to spherical harmonic $Y_2^0$ (panel $(b)$ in  Fig.~\ref{fig:Callahan_spherical_harmonics}) as predicted by the theory.

        \paragraph{The case $\ell = 3$} When $(\gamma, \delta) = (4.855527, 0.28)$, $\mbf P$ goes through a (Turing) pitchfork bifurcation. Here, the amplitude equations are found to be given by \eqref{examplel3}
        with $a = 0.00128584$, $b = - 0.078502$, and $C_{1, 1} = 0.551339$ when $\mu = \gamma$. Therefore, we are in Region 1 of the diagram shown in Figure \ref{fig:bifdiagl3}, and we expect to see a stable low-amplitude steady state after the bifurcation that resembles spherical harmonic $Y_3^2$ (panel $(c)$ in Fig.~\ref{fig:Callahan_spherical_harmonics}). We note that $a$ is close to zero. Thus, if we were to move the parameters a little bit into the region where $a$ is negative (Region 7), then we \rtext{would find} a different stable pattern (solution not shown, but it resembles $Y_3^3$, the pattern found for $\ell = 3$ in Fig.~\ref{fig:Fahad_l=1,2,3_t=170}).

        Indeed, setting $(\gamma, \delta) = (5, 0.28)$, which is just beyond the bifurcation point, we find the dispersion relation shown in Figure \ref{fig:disprell=1} with the mode $\ell = 3$ being unstable. Furthermore, after integrating the system with a random initial condition close to the steady state up to $t = 400$, we found the pattern shown in the bottom-left panel of Figure \ref{fig:Brusselator_l=1,2,3,4}, which indeed corresponds to $Y_3^2$ as predicted. 
        \paragraph{The case $\ell = 4$} When $(\gamma, \delta) = (4.902346, 0.22)$, $\mbf P$ goes through a (Turing) transcritical bifurcation. In this case, the amplitude equations up to second order are given by \eqref{quadratic_coef} with $a = - 0.000140359$ and $C_{1, 1} = 0.521154$ for $\mu = \gamma$. Again, in this case, we find that $a$ is small, and we are close to a codimension-two point for which $a = 0$, which we find to occur at $(\gamma, \delta) = (4.90393, 0.216563)$.
            
        To determine the kind of pattern we expect to see requires the evaluation of cubic coefficients, which depend on two independent parameters, $b$ and $c$ (see \cite{callahan,chossat} and Supplementary Materials). We found at the codimension-two point that $b = 0.0000758108$ and $c = - 0.0480261$. This can be shown to imply (based on \cite[Tables 4-5]{callahan}) that pattern $(f)$ in Fig.~\ref{fig:Callahan_spherical_harmonics} will be the stable branch that emerges from the bifurcation.

        To confirm this weakly nonlinear prediction, we set the parameters $(\gamma, \delta) = (5, 0.22)$ for which we get the dispersion relation shown in Figure \ref{fig:disprell=1} with the mode $\ell = 4$ being unstable. Moreover, after setting a random initial condition close to the steady state and integrating the system up to $t = 400$, we got the pattern shown in the bottom-right panel of Figure \ref{fig:Brusselator_l=1,2,3,4}. This pattern also comprises spots and stripes as the one in the top-right panel of Figure \ref{fig:Brusselator_l=1,2,3,4}. It corresponds to pattern $(f)$ in Figure \ref{fig:Callahan_spherical_harmonics}, as predicted by the theory.

    \subsection{Cell polarity model} \label{sub:cellpolarity}
        Since we have developed the methodology in a general way, it is useful to consider an example with more than two components. To that end, we present a model analyzed in Chapter 6 of \cite{Fahad2022PhD}. This is a small modification, through the presence of source and loss terms, to the model developed by Abley {\it et al}  \cite{abley-intracellular} for cellular pattern formation. Here, there are four fields, each of which represents the concentration of a small G-protein, known as an ROP in plants. These ROPs occur in two families with concentrations $(u, v)$ and $(w, z)$ respectively, each of which has an active ($u$ and $w$)  and an inactive ($v$ and $z$) form. The inactive form is mostly present in the cell body, and the active form is on the cell membrane, where the reaction kinetics occurs. Different patterns of the active form are thought to provide the precursor to differential growth; see e.g.~\cite{beta,MeronIssue}. \rtext{We consider the case where there are source and loss terms, in contrast to other cell-polarity BS-RDE models that show species conservation \cite{ratz2015turing,ratz2014symmetry,cusseddu2019}.}
 
        \rtext{Most} mathematical studies of related systems typically take a homogenization approach within a single lower-dimensional domain, which represents both the bulk and the surface. Here, the system will be posed as
        \begin{align*}
            \partial_t \mbf U &= - B \, \mbf U + \mathbb D_U \, \nabla_\Omega^2 \mbf U, \quad \text{in } \Omega,
            \\
            \mathbb D_U \, \rtext{\left. \frac{\partial \mbf U}{\partial \mbf{\hat n}}\right|_{r = R}} &= K \left(\mbf u - \left.\mbf U\right|_{\rtext{r = R}}\right), 
            \\
            \partial_t \mbf u &= \mbf g(\mbf u) - K \left(\mbf u - \left. \mbf U \right|_{\rtext{r = R}}\right) + \mathbb D_u \, \nabla_\Gamma^2 \, \mbf u, \quad \text{on } \Gamma, \quad \text{ where}
        \end{align*}
        $\mbf U = (U, V, W, Z)^\intercal$, $\mbf u = (u, v, w, z)^\intercal$, $B = \diag\left(\sigma_1, \sigma_2, \sigma_3, \sigma_4\right)$,
        \\
        $\mathbb D_U = \diag\left(D_1, D_2, D_3, D_4\right)$,
        $K = \diag\left(K_1, K_2, K_3, K_4\right)$,
        \\
        $\mbf g(\mbf u) = \begin{pmatrix}
            F(u, v, w, z) - \kappa \, \xi \, u
            \\
            - F(u, v, w, z) + \kappa \, \theta
            \\
            G(u, v, w, z) - \kappa \, \xi \, w
            \\
            -G(u, v, w, z) + \kappa \, \theta
        \end{pmatrix}$, $\mathbb D_u = \diag\left(\delta_1^2, 1, \delta_2^2, 1\right)$, and
        \\
        $F(u ,v, w, z) = \left(\eta  u^2 + \rho\right) v - \left(\alpha w + \mu\right) u$, and $G(u, v, w, z) = \left(\eta w^2 + \rho\right) z - \left(\alpha u + \mu\right) w$.

        We find that although \rtext{the} calculation of the coefficients of the amplitude equations takes more time, the normal forms have the same form as the ones in the previous example, as predicted by the equivariant bifurcation theory. For illustration, we consider the parameter values shown in Table \ref{tab:fixedparFahad}, which are indicative rather than precise values for any particular biological system.
        \begin{table}
            \centering
            \begin{tabular}{|c|c|c|c|c|c|c|c|c|c|c|}
                \hline
                Parameter & $R$ & $\sigma_1$ & $\sigma_2$ & $\sigma_3$ & $\sigma_4$ & $\alpha$ & $\eta$ & $\theta$ & $\mu$ & $\xi$
                \\
                \hline
                Value & 1 & 100 & 1 & 100 & 1 & 0.05 & 3.6 & 5.5 & 0.5 & 2.7
                \\
                \hline
                \hline
                Parameter & $\rho$ & $K_1$ & $K_2$ & $K_3$ & $K_4$ & $D_1$ & $D_2$ & $D_3$ & $D_4$ & $\delta_2$
                \\
                \hline
                Value & 0.06 & 1 & 100 & 1 & 100 & 1 & 1 & 1 & 1 & 1
                \\
                \hline
            \end{tabular}
            \caption{Parameter values fixed for the analysis of the cell-polarity model.}
            \label{tab:fixedparFahad}
        \end{table}
        Here, we leave $\delta_1$ and $\kappa$ as free parameters. As usual, we start by computing bifurcation curves and dispersion relations for different values of $\ell$. Figure \ref{fig:bifdiagfahad} shows the corresponding neutral stability curves in the $\left(\kappa, \delta_1\right)$ plane.
        
        \begin{figure}
            \centering
            \begin{subfigure}[c]{0.48\textwidth}
                \centering
                \begin{tikzpicture}
                    \node[anchor=south west,inner sep=0] (image) at (0,0) {\includegraphics[width = \textwidth]{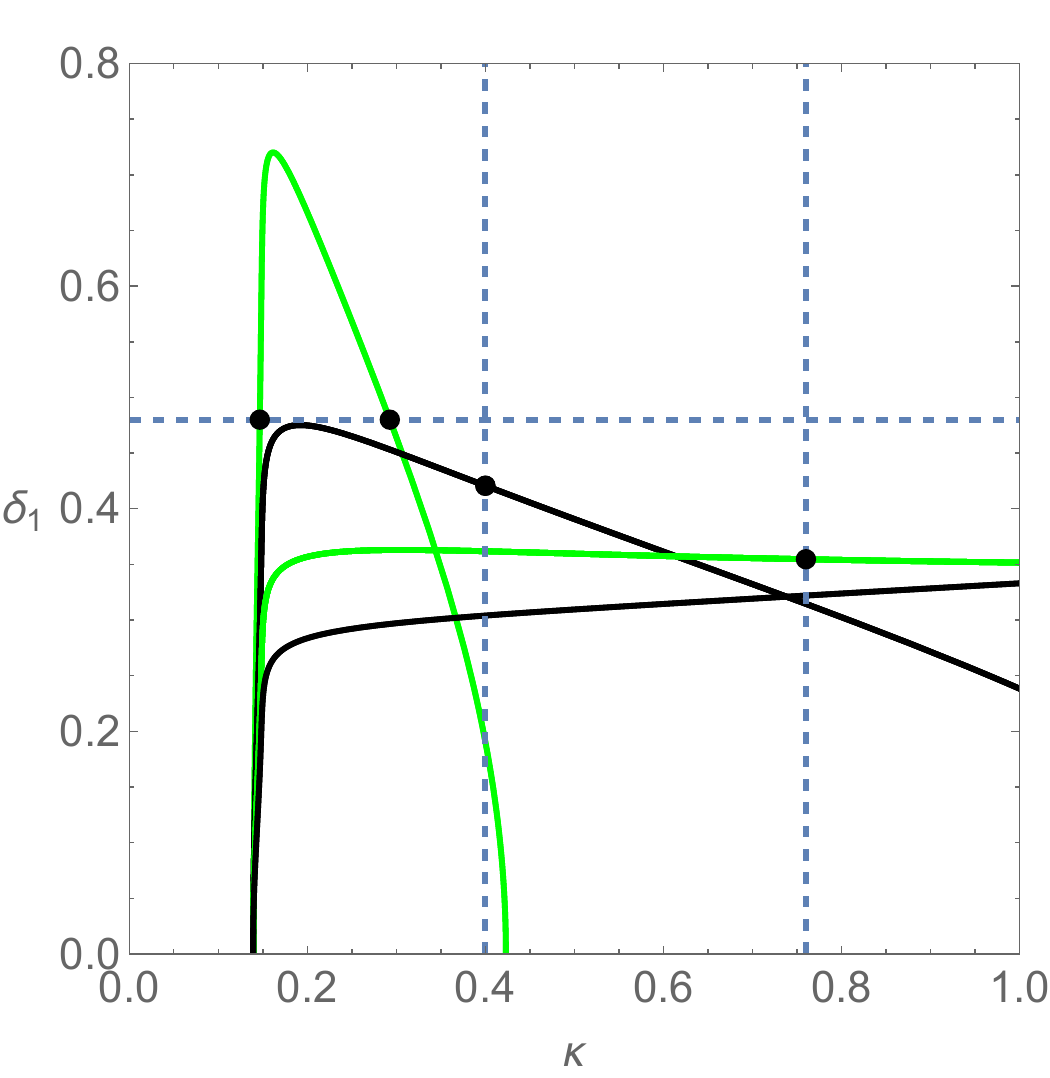}};
                    \node at (2.0, 7.3) {\small $\ell$ = 1};
                    \node at (3.8, 4.8) {\small $\ell$ = 2};
                    \node at (6.2, 4.2) {\small $\ell$ = 3};
                    \node at (4.4, 3.25) {\small $\ell$ = 4};
                \end{tikzpicture}
                \caption{}
                \label{fig:bifdiagfahad}
            \end{subfigure}
            \hfill
            \begin{subfigure}[c]{0.48\textwidth}
                \centering
                \includegraphics[width = \textwidth]{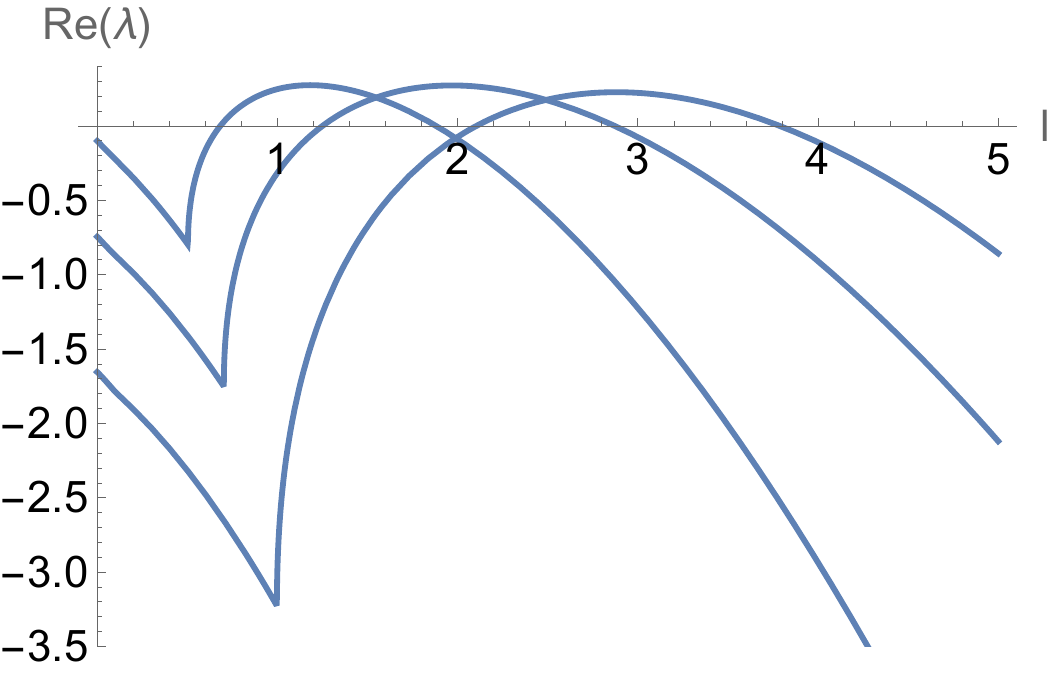}
                \caption{}
                \label{fig:disprelfahad}
            \end{subfigure}
            \caption{(a) Bifurcation curves for the cell polarity model at one steady state when the parameters in Table \ref{tab:fixedparFahad} are fixed. Green (resp. black) curves are pitchfork (resp. transcritical) bifurcations in the corresponding amplitude equations. (b) Dispersion relations for $\left(\kappa, \delta_1\right) \in \{(0.25, 0.48), (0.4, 0.37), (0.76, 0.33)\}$. \rtext{Here, $\Re(\lambda)$ represents the real part of the eigenvalue of $\mathbb B_\ell$ with \rtext{ largest} real part.}}
        \end{figure}
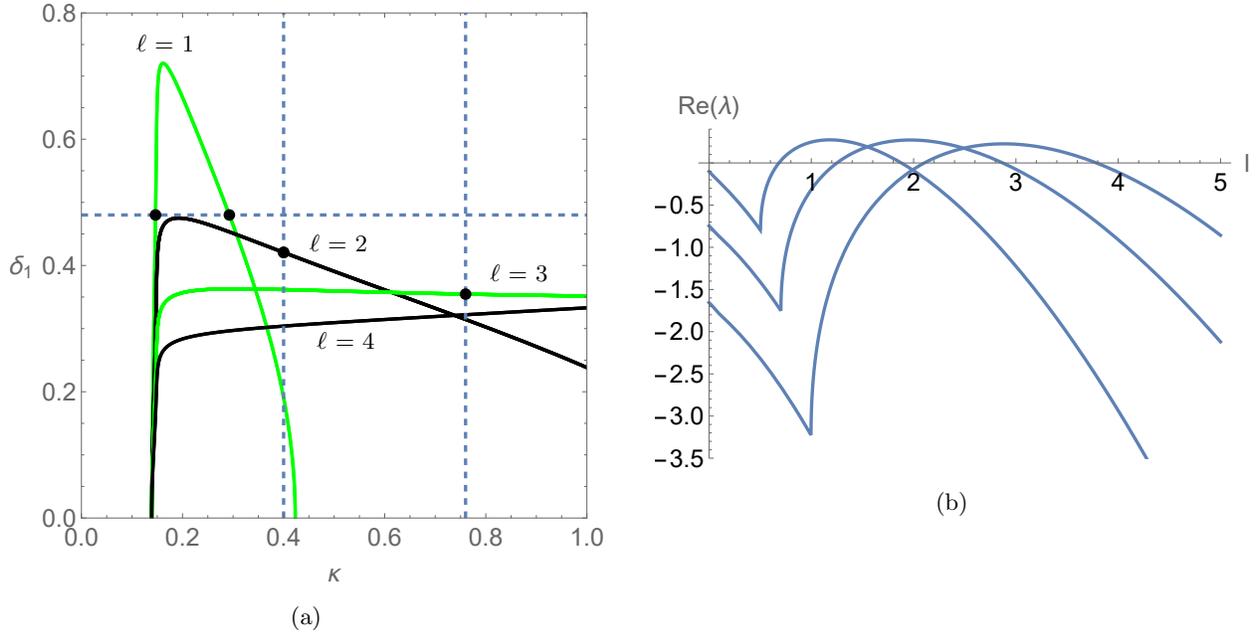


        \paragraph{The case $\ell = 1$} When $\left(\kappa, \delta_1\right) = (0.292672, 0.48)$, there is a bifurcation for $\ell = 1$, with the amplitude equations given by \eqref{examplel1} with $a = 1.6352 > 0$ and $C_{1, 1} = - 6.30308$ for $\mu = \kappa$. On the other hand, when $\left(\kappa, \delta_1\right) = (0.146644, 0.48)$, the same happens with $a = 160.42 > 0$ and $C_{1, 1} = 112.596$ for $\mu = \kappa$. This implies that we do not expect to see a stable pattern when integrating the solution when $\delta_1 = 0.48$ and $0.146644 < \kappa < 0.48$ is sufficiently close to these bifurcation points.
        Figure \ref{fig:disprelfahad} shows the dispersion relation obtained for $\left(\kappa, \delta_1\right) = (0.25, 0.48)$, which is in the interior of the region in which the $\ell = 1$ mode is the only unstable one. Furthermore, when integrating the system from a random perturbation to the ground state, we do not see a stable steady state. Instead, the solution seems to converge to a periodic orbit. One state of the solution is shown in the first and second rows of Figure \ref{fig:Fahad_l=1,2,3_t=170} for a fixed time instant. It is clear from the temporal evolution that the simulation quickly settles to a limit cycle (see the $L^2$-norm of Figure \ref{fig:fahadl=1}, shown within the supplementary material and the associated videos there).

        \paragraph{The case $\ell = 2$} When $\left(\kappa, \delta_1\right) = (0.4, 0.420644)$, there is a bifurcation for $\ell = 2$ that gives rise to amplitude equations given by \eqref{examplel2} with $a = 0.0414126$ and $C_{1, 1} = - 5.81929$ for $\mu = \delta_1$. Here, we are indeed close to a codimension-two bifurcation point. In fact, there is a codimension-two point for $\left(\kappa, \delta_1\right) = (0.362182, 0.432158)$, at which $b = 1.15796 > 0$. Therefore, we expect to see a stable moderate-amplitude steady state close to the bifurcation. The dispersion relation obtained for $\left(\kappa, \delta_1\right) = (0.4, 0.37)$ is shown in Figure \ref{fig:disprelfahad}. Furthermore, the pattern obtained after integrating the system at these parameter values from a random perturbation to the steady state $\mbf P$ is shown in the third and fourth rows of Figure \ref{fig:Fahad_l=1,2,3_t=170}.

        \paragraph{The case $\ell = 3$} When $\left(\kappa, \delta_1\right) = (0.76, 0.354603)$, there is a bifurcation associated with $\ell = 3$ that gives rise to amplitude equations that look like \eqref{examplel3} with $a = - 0.0252879$, $b = - 1.03659$ and $C_{1, 1} = - 9.44413$ for $\mu = \delta_1$. Therefore, we are in Region 7 on the diagram shown in Figure \ref{fig:bifdiagl3}, which makes us expect pattern $(d)$ in Figure \ref{fig:Callahan_spherical_harmonics} to become stable right after the bifurcation. The dispersion relation obtained for $\left(\kappa, \delta_1\right) = (0.76, 0.33)$ is shown in Figure \ref{fig:disprelfahad}. The pattern obtained after integrating the system at these parameter values from a random small perturbation to the steady state $\mbf P$ is shown in the last two rows of Figure \ref{fig:Fahad_l=1,2,3_t=170}.
        
        \begin{figure}
            \centering
            \begin{subfigure}[b]{\textwidth}
                \centering
                \includegraphics[width=0.8\textwidth]{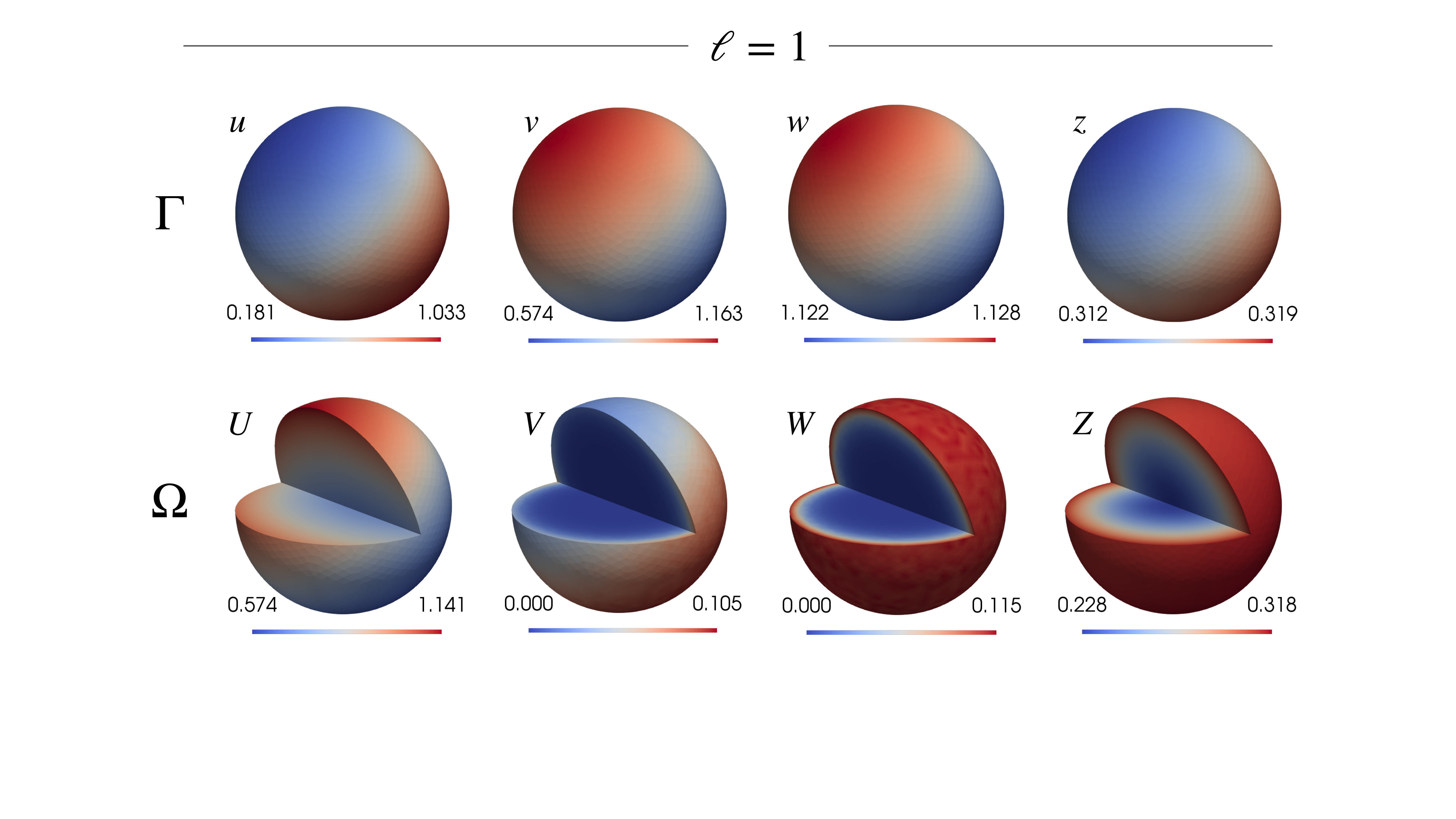}
            \end{subfigure}
            \hfill
            \begin{subfigure}[b]{\textwidth}
                \centering
                \includegraphics[width=0.8\textwidth]{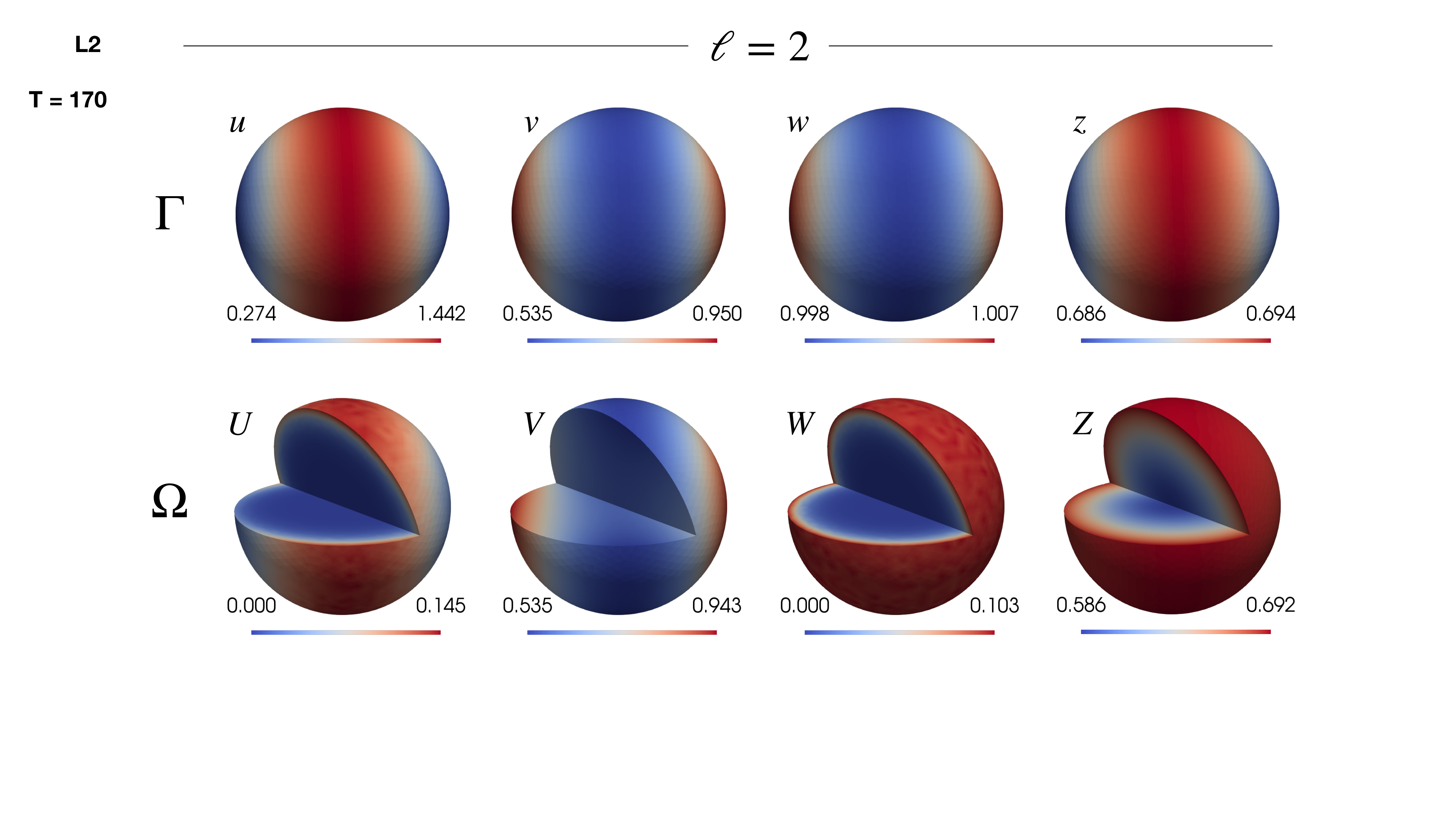}
            \end{subfigure}
            \begin{subfigure}[b]{\textwidth}
                \centering
                \includegraphics[width=0.8\textwidth]{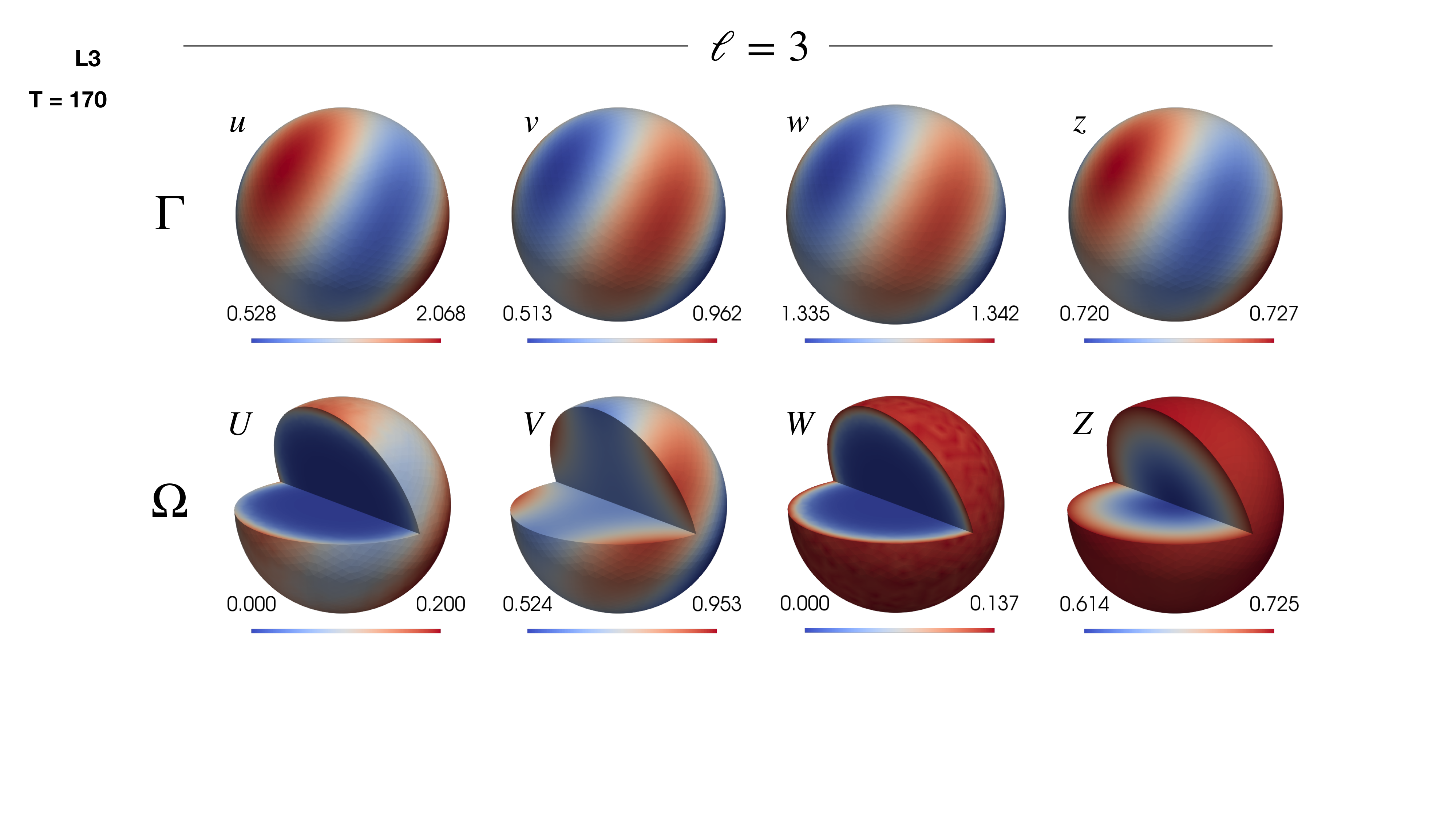}
            \end{subfigure}
            \caption{The numerical solutions of the cell polarity system at time $t = 170$, for three different parameter sets, corresponding to the choices $\ell = 1, 2, 3$. For each case, the surface solutions $u, v, w, z$ are plotted in the first row, while the bulk solutions $U, V, W, Z$ are in the second one. Note that the solution for $\ell = 1$ is not at steady state, whereas the other two cases are (see text for details).}
            \label{fig:Fahad_l=1,2,3_t=170}
        \end{figure}

\section{Conclusion}\label{sec:conclusions}
    In this paper, we have provided a method for computing the linear and nonlinear parameters necessary for analysing (Turing) transcritical and pitchfork bifurcations for a broad class of bulk-surface reaction-diffusion systems posed on a sphere. In particular, we consider the physically realistic case where the nonlinear interaction kinetics occurs on the surface, with linear diffusion and decay in the bulk, coupled by a boundary condition that has a natural interpretation in terms of concentration-driven flux. 

    The case of the sphere with linear kinetics in the bulk was tractable because of the compatibility of the eigenfunctions of the Laplace and Laplace-Beltrami operators in the bulk and on the surface. This enables nonlinear terms arising from products to be expressed in terms of the same eigenbases, leading to closed-form expressions for the required amplitude-equation coefficients. \rtext{This, in turn, enables us to find closed-form solvability conditions in Sec.~4 to calculate the normal-form coefficients. The method then proceeds as it would for a problem without the complexity of having separate bulk and surface PDEs, although some of the expressions become rather cumbersome. To aid the reader, the first author \cite{amplitude-eq} has provided both symbolic and numerical code for evaluating the necessary coefficients for any example system of the form \eqref{generalsystem}.}

    \rtext{We note that other assumptions on the form of coupling between bulk and surface concentrations are possible; for example, in principle, we could include nonlinearty in the coupling boundary conditions. Other possibilities might be to allow the matrix $K$ in \eqref{generalsystem} to be singular. Such a decoupling would effectively allow a different number of species to evolve in the bulk and on the surface. Another possible simplification could be to take the limit that some of the bulk diffusion constants tend to infinity (as in \cite{ratz2014symmetry}). In principle, nonlinearity could be included in both the bulk PDEs as well, although care would need to be taken in the form of coupling boundary conditions. All such considerations are left for future work}.

    Compared \rtext{to} bifurcations in planar geometries, the amplitude equations are a little harder to analyze in the case studied because of the $O(3)$-symmetry of the sphere \cite{golubitsky}. Without making an {\it a priori} assumption, we find that the coefficients follow the normal forms predicted by the theory of such symmetry-breaking \cite{callahan,chossat}. Our main contribution then is to show how to calculate such normal forms for our class of bulk-surface systems. Moreover, we have provided software implementation of this general method \cite{amplitude-eq}. We have also backed up our theory with computational results using a state-of-the-art finite element method. 

    Note, however, that we have not fully explored the stability of steady solutions that can exist, especially for higher than $\ell = 1, 2$, as the number of components of the amplitude equations becomes prohibitive. Nevertheless, we have been able to calculate whether the primary bifurcation is super\rtext{-} or sub\rtext{-}critical. There has been recent interest in sub-critical bifurcations and how they can lead to localised structures (see e.g.~\rtext{\cite{Fahad,edgardodegenerate,chapman,burke2007,Villar-Sepulveda-beyond}}) and references therein. The analogy of such states in bulk-surface systems remains unexplored. 


    Perhaps some of the spatio-temporal behaviour, including the temporally periodic orbit we observed for $\ell = 1$ in the cell polarity model, may be explainable by studying the complete dynamics of the relevant normal form. Also, we established that the normal forms arising at each kind of bifurcation have several symmetries, with consequent invariant subspaces. \rtext{A more complete investigation of the dynamics of these normal forms is left for future work.}
    
    \rtext{For $\ell > 1$, the full dynamics of each normal form is likely to be yet more complex, not least due to the many additional invariant subspaces. In such situations, one might expect to see parameters where there is repetitive metastable dynamics due to robust heteroclinic cycles, see e.g.~\cite{Krupa,LloydChampWilson}. In the present context of evolution equations with $O(3)$-symmetry, Chossat and co-workers \cite{Chossat1,Chossat2,Chossat3} have established conditions under which such robust cycles exist close to parameter values where there is simultaneous instability for two different values of $\ell$. Investigation of the dynamics close to such codimension-two bifurcation points for the systems considered here is left for future work.}

    Finally, we remark on how the method developed in this paper could be adapted to provide amplitude equations for bulk-surface reaction-diffusion systems in other geometries. \rtext{The presented} approach is general, and extensions to other geometries may well be possible, provided a suitable alternative to Lemma \ref{lemma:linearexpansion} can be established for that geometry. \rtext{Other possibilities include abandoning the method of \cite{tirapegui} for normal form computation and resorting to fixed-point methods. Alternatively, as in \cite{ratz2014symmetry}, one could explore certain limits in which the bulk dynamics become trivial, leading to nonlocal surface-only equations, which can be posed on more complex geometries (see \cite{bobrowski2025existence} for example). These considerations are left for future work.}

\section*{Acknowledgments}
    The authors acknowledge helpful conversations with Alastair Rucklidge (Leeds, UK),
    Michael Ward (UBC, Canada) and Fr\'{e}d\'{e}ric Paquin-Lefebvre
    (Ecole Normal Sup\'{e}rieure, Paris). The work of EV-S was supported by ANID, Beca Chile Doctorado en el extranjero, number 72210071. DC was supported by the ``Fundação para a Ciência e a Tecnologia (FCT)'' under the projects with reference UIDB/00208/2020, \rtext{UIDB/00013/2020} and UIDP/00013/2020. AM was supported by a Canada Research Chair  (CRC-2022-00147), the Natural Sciences and Engineering Research Council of Canada (NSERC), Discovery Grants Program (RGPIN-2023-05231), the British Columbia Knowledge Development Fund (BCKDF), Canada Foundation for Innovation – John R. Evans Leaders Fund – Partnerships (CFI-JELF), the UKRI Engineering and  Physical Sciences Research Council (EPSRC: EP/J016780/1) and the  British Columbia Foundation for Non-Animal Research.

\bibliographystyle{plain}
\bibliography{refs}

\newpage

\appendix

\section*{\Large Supplementary materials}

\section{Proof of Lemma \ref{Lemma1}} \label{appendix_proof_of_Lemma}
    For \eqref{firstinvman}, note that
    \begin{align*}
        \frac{\dd}{\dd t}\left(\abs{A_2}^2 - \frac{3}{2} \, A_0^2\right) &= \dot A_{-2} \, A_2 + A_{-2} \, \dot A_2 - 3 \, A_0 \, \dot A_0
        \\
        &= A_{-2} \left(\varepsilon + 4 \, a \, A_0\right) \, A_2
        \\
        & + A_{ -2} \, A_2 \left(\varepsilon + 4 \, a \, A_0\right)
        \\
        & - 3 \, A_0 \, \left(A_0 \left(\varepsilon - 2 \, a \, A_0\right) + 4 \, a \, A_{-2} \, A_2\right)
        \\
        &= 2 \, \abs{A_2}^2 \left(\varepsilon + 4 \, a \, A_0\right)
        \\
        & - 3 \, \left(A_0^2 \left(\varepsilon - 2 \, a \, A_0\right) + 4 \, a \, A_0 \, \abs{A_2}^2\right)
        \\
        &= 0.
    \end{align*}
    On the other hand, to analyze \eqref{secondinvman}, note that $\dot A_0 = 0$ trivially, by the conditions that set the manifold. Furthermore,
    \begin{align*}
        \frac{\dd}{\dd t} \left(2 \, \abs{A_2}^2 - \abs{A_1}^2\right) &= 2 \left(\dot A_{-2} \, A_2 + A_{-2} \, \dot A_2\right) + \dot A_{-1} \, A_1 + A_{-1} \, \dot A_1
        \\
        &= 2 \left(\left(\varepsilon \, A_{-2} - \sqrt{6} \, a \, A_{-1}^2\right) \, A_2\right.
        \\
        & \left. + A_{-2} \, \left(\varepsilon \, A_2 - \sqrt{6} \, a \, A_1^2\right)\right)
        \\
        & + \left(\varepsilon \, A_{-1} + 2 \, \sqrt{6} \, a \, A_1 \, A_{-2}\right) \, A_1
        \\
        & + A_{-1} \, \left(\varepsilon \, A_1 + 2 \, \sqrt{6} \, a \, A_{-1} \, A_2\right)
        \\
        &= 2 \left(2 \, \varepsilon \, \abs{A_2}^2 - \sqrt{6} \, a \, A_{-1}^2 \, A_2 - \sqrt{6} \, a \, A_{-2} \, A_1^2\right)
        \\
        & - 2 \, \varepsilon \, \abs{A_1}^2 + 2 \, \sqrt{6} \, a \, A_1^2 \, A_{-2} + 2 \, \sqrt{6} \, a \, A_{-1}^2 \, A_2
        \\
        &= 0.
    \end{align*}
    Finally, to analyze \eqref{steadystateman}, note that if we subtract the first two equations, we get
    \begin{align*}
        \dot A_0 = A_0 \left(\varepsilon - 2 \, a \, A_0\right) - 2 \, a \, \abs{A_1}^2 + 4 \, a \, \abs{A_2}^2 = 0.
    \end{align*}
    Therefore, due to the last equation in \eqref{steadystateman}, we have that $A_0$ and $A_2$ are constants fulfilling
    \begin{align*}
        4 \, A_0 \left(- \, A_0 + \frac{\varepsilon}{a}\right) + 24 \, \abs{A_2}^2 = \frac{\varepsilon^2}{a^2}.
    \end{align*}
    Now, note that
    \begin{align*}
        \dot A_1 = A_1 \left(\varepsilon - 2 \, a \, A_0\right) + 2 \, \sqrt{6} \, a \, A_{-1} \, A_2 &= A_1 \left(\varepsilon - 2 \, a \, A_0\right) + 2 \, \sqrt{6} \, a \, A_{-1} \, \frac{\sqrt{6} \, a \, A_1^2}{\varepsilon + 4 \, a \, A_0}
        \\
        &= A_1 \, \frac{\left(\varepsilon - 2 \, a \, A_0\right) \left(\varepsilon + 4 \, a \, A_0\right) + 12 \, a^2 \, A_{-1} \, A_1}{\varepsilon + 4 \, a \, A_0}
        \\
        &= 0
    \end{align*}
    Therefore, the right-hand side of all the equations in \eqref{examplel2} \rtext{is} equal to zero. That implies that all the variables are constants in the manifold determined by \eqref{steadystateman}, which concludes the proof.

    \newpage

\section{The numerical scheme}\label{sec:bsfem}
    As a complement to our analytical results, we present and discuss the numerical solutions of some examples of bulk-surface systems of type \eqref{generalsystem}. We solve such systems by employing the bulk-surface element method, which has been developed and used in the literature for studying different systems, often arising from modelling real-world phenomena, such as cell polarisation and migration (see e.g. \cite{madzvamuse2015, cusseddu2019, elliott2013finite, macdonald2016computational,  madzvamuse2016bulk}). Before outlining the procedure for the specific system \eqref{generalsystem}, we rewrite the system in a more explicit form. Specifically, we write 
    \begin{align}
        \frac{\partial U_p}{\partial t}  &= - \sigma_p U_p + D_p \nabla_\Omega^2 U_p, \quad \text{ in } \Omega, \label{eq:BS_system_explicit_bulk}
        \\
        D_p\frac{\partial U_p}{\partial \nu} & = K_p (u_p - U_p), \quad \text{ on } \Gamma,
        \label{eq:BS_system_explicit_bc}
        \\
        \frac{\partial u_p}{\partial t}  &= g_p(u_1, \ldots, u_n) 
        - K_p u_p + K_p U_p  
        + d_p \nabla_\Gamma^2 u_p, \quad \text{ on } \Gamma,
        \label{eq:BS_system_explicit_surface}
    \end{align}
    for $p = 1, \ldots, n$. In what follows we assume $g_p \in L^2(\mathbb R^n)$ for all $p = 1, \ldots, n$. Let $\Phi_p$ and $\phi_p$ be functions of the Sobolev spaces $H^1(\Omega)$ and $H^1(\Gamma)$, respectively. The following weak formulation of the above problem is derived by testing \eqref{eq:BS_system_explicit_bulk} with $\Phi_p$ and applying the boundary condition \eqref{eq:BS_system_explicit_bc}, and by testing  \eqref{eq:BS_system_explicit_surface} with $\phi_p$. It must be noted that the surface equation \eqref{eq:BS_system_explicit_surface} has no boundary conditions since $\Gamma$ is a closed surface. \rtext{ While our work is carried out on a ball of radius $R$, here we describe the numerical method on a general bounded domain $\Omega \subset \mathbb R^3$, with no explicit dependence on the radius $R$. We remark that, when $\Omega = \left\{\mbf x \in \mathbb R^3: \lVert \mbf x \rVert \leq R \right\}$, the system can be solved on the reference domain $B_1(\mbf 0)$, after rescaling the bulk and surface diffusion coefficients $D_p$, $d_p$.}

    Hence, the bulk-surface finite element method seeks to \rtext{f}ind $(U_1, \ldots, U_n)$ $\in L^2([0, T]; H^1(\Omega)) \cap L^\infty([0, T]\times\Omega)$, with  
    $$
        \left( \frac{\partial U_1}{\partial t}, \ldots, \frac{\partial U_n}{\partial t} \right) \in L^2([0,T]; H^{- 1}(\Omega))
    $$
    and $\left(u_1, \ldots, u_n\right) \in L^2([0, T]; H^1(\Gamma))\cap L^\infty([0, T]\times\Gamma)$  with 
    $$
        \left ( \frac{\partial u_1}{\partial t}, \ldots, \frac{\partial u_n}{\partial t} \right ) \in L^2([0,T]; H^{-1}(\Gamma)),
    $$
    such that, $\forall \, \Phi_p\in H^1(\Omega)$ and $\forall \, \phi_p \in H^1(\Gamma)$, these functions satisfy 
    \begin{align}
        \int_\Omega \frac{\partial U_p}{\partial t} \, \Phi_p \, \dd \mbf{x}  &= - \int_\Omega \sigma_p \, U_p \, \Phi_p \, \dd \mbf{x} - D_p  \int_\Omega  \nabla_\Omega \, U_p \cdot \nabla_\Omega \, \Phi_p \, \dd \mbf{x}  \label{eq:BS_system_explicit_weak_bulk}
        \\
        & \hspace{2cm} + K_p \int_\Gamma   (u_p - U_p) \, \Phi_p \, \dd \mbf{s}, \notag 
        \\
        \int_\Gamma \frac{\partial u_p}{\partial t} \, \phi_p \, \dd \mbf{s} &= \int_\Gamma g_p(u_1, \ldots, u_n) \, \phi_p \, \dd \mbf{s} - K_p \int_\Gamma (u_p - U_p) \, \phi_p \, \dd \mbf{s} \label{eq:BS_system_explicit_weak_surface}
        \\
        & \hspace{2cm}- d_p \int_\Gamma \nabla_\Gamma \, u_p \cdot \nabla_\Gamma \, \phi_p \dd \mbf{s}, \notag 
    \end{align}
    with $U_p(0, \mbf x) = U_{p, 0}(\mbf x)$ and $u_p(0, \mbf x) = u_{p,0}(\mbf x)$ as prescribed by the initial conditions to \eqref{eq:BS_system_explicit_bulk}-\eqref{eq:BS_system_explicit_surface}, for $p = 1, \ldots, n$. The spaces $L^2([0, T]; H^1(\Omega))$ and  $L^2([0, T]; H^1(\Gamma))$ are known as Bochner spaces and $H^{- 1}(\Omega)$ and $H^{- 1}(\Gamma)$ are the dual space of, respectively, $H^1(\Omega)$ and $H^1(\Gamma)$ \cite{evans2010partial, salsa2016partial}.
    
    \subsection{Spatial discretisation}
        The bulk-surface finite element method \\ (BS-FEM) is based on a restriction of the above weak formulation to discrete function spaces. First, let $\Omega_h$ be a polyhedral approximation of the domain $\Omega$ and $\Gamma_h = \partial\Omega_h$ its boundary.  Let $\mathcal T_h = \left\{T_1, \ldots, T_{n_{x(h)}}\right\}$ be a mesh on $\Omega_h$, i.e. a set of polyhedrons whose intersection is either empty or a polygonal face and such that $\bigcup_{i = 1}^{n_x(h)} T_i = \Omega_h$ (for details on proper choices of such elements see, e.g., \cite{quarteroni2008numerical}). By considering the set $\mathcal S_h = \left\{T_i\cap\Gamma_h, \text{ for } i = 1, \ldots, n_{x(h)}\right\}$ which is constituted by polygons, i.e. by the faces of the elements of $\mathcal T_h$ relying on the boundary $\Gamma_h$, we naturally define a mesh on $\Gamma_h$ entirely induced by $\mathcal{T}_h$. Therefore, by choosing tetrahedron elements $T_i$, the mesh $\Gamma_h$ will be constituted by triangles. Let $N_\Omega$ and $N_\Gamma$ denote the number of total vertices in the bulk and surface meshes, respectively. These meshes help us in defining appropriate basis functions for deriving the BS-FEM. Let us define the following function spaces
        \begin{align*}
            V_h(\Omega_h) &:= \Big\{ \Phi_h :\Omega_h \to \mathbb R \text{ such that } \Phi_h \in C^0(\Omega_h) \text{ and } \Phi_h|_T \in \mathbb P_1(T), \forall \, T \in \mathcal T_h
            \Big\},
            \\
            W_h(\Gamma_h) &:= \Big\{\phi_h :\Gamma_h \to \mathbb R \text{ such that } \phi_h \in C^0(\Gamma_h) \text{ and }  \phi_h|_S \in \mathbb P_1(S), \forall \, S \in \mathcal S_h\Big\},
        \end{align*}
        where $\mathbb P_1(D)$ indicates the space of all linear polynomials over a domain $D$. In particular, $V_h\left(\Omega_h\right)\subset H^1\left(\Omega_h\right)$ and $V_h\left(\Gamma_h\right)\subset H^1\left(\Gamma_h\right)$. Basis functions for these two spaces are, respectively, $\left(\Psi_i(\mbf x)\right)_{i = 1, \ldots, N_\Omega}\in [V_h(\Omega_h)]^{N_\Omega}$ and $\left(\psi_i(\mbf x)\right)_{i = 1, \ldots, N_\Gamma}\in [W_h(\Gamma_h)]^{N_\Gamma}$ with the property $\Psi_i(\mbf x_j) = \delta_{i, j}$  
        for all vertices $\mbf x_j$ of the mesh $\mathcal T_h(\Omega_h)$ ($j = 1, \ldots, N_{\Omega}$) 
        and $\psi_i(\mbf x_k) = \delta_{i, k}$ for all vertices $\mbf x_k$ of $\mathcal S_h(\Gamma_h)$ ($k = 1, \ldots, N_{\Gamma}$). Here, $\delta_{i, j}$ is the Kronecker delta.

        Thanks to the basis functions, our problem goes from looking for functions over the time-space domain to look for solely time-dependent functions $A^p_i(t)$ and $a^p_i(t)$ such that $U_{h, p}(t, \mbf x) = \sum_i A^p_i(t) \Psi_i(\mbf x)$ and $u_{h, p}(t, \mbf x) = \sum_i a^p_i(t) \psi_i(\mbf x)$ for $p = 1, \ldots, n$. Therefore, the semi-discrete weak formulation reads as follows.
        
        Find $A^p_i \in \mathbb R$ for $i = 1, \ldots, N_\Omega$ and $a^p_k \in \mathbb R$ for $i = 1, \ldots, N_\Gamma$ and $p = 1, \ldots, n$ such that
        \begin{align*}
            \sum_{i = 1}^{N_\Omega} \frac{d {A^p_i}}{dt} \int_{\Omega_h} \Psi_i \, \Psi_j \, \dd \mbf x &= - \sigma_p \sum_{i = 1}^{N_\Omega} {A^p_i}(t) \int_{\Omega_h} \Psi_i \, \Psi_j \, \dd \mbf x
            \\
            & - D_p \sum_{i = 1}^{N_\Omega} {A^p_i}(t) \int_{\Omega_h} \nabla \Psi_i \cdot \nabla \Psi_j \, \dd \mbf x \nonumber
            \\
            & + K_p \left(\sum_{k = 1}^{N_\Gamma} {a^p_k}(t) \int_{\Gamma_h} \psi_k \, \Psi_j \, \dd \mbf s - \sum_{i = 1}^{N_\Omega} {A^p_i}(t) \int_{\Gamma_h} \Psi_i \, \Psi_j \, \dd \mbf s\right),
            \\ 
            \sum_{k = 1}^{N_\Gamma} \frac{d a^p_k}{dt} \int_{\Gamma_h} \psi_k \, \psi_h \, \dd \mbf s &= \int_{\Gamma_h} g_p\rtext{\left(\sum_{i = 1}^{N_\Gamma} a^1_i(t) \psi_i(\mbf x), \ldots, \sum_{i=1}^{N_\Gamma} a^n_i(t) \psi_i(\mbf x)\right)} \, \psi_h \, \dd \mbf s
            \\
            & - d_p \sum_{k = 1}^{N_\Gamma} {a^p_k}(t) \int_{\Gamma_h} \rtext{\nabla_\Gamma \, \psi_k \cdot \nabla_\Gamma \, \psi_h \,} \dd \mbf s \nonumber
            \\
            & - K_p \left( \sum_{k = 1}^{N_\Gamma} {a^p_k}(t) \int_{\Gamma_h} \psi_k \, \psi_h \, \dd \mbf s - \sum_{i = 1}^{N_\Omega} {A^p_i}(t) \int_{\Gamma_h} \Psi_i \, \psi_h \, \dd \mbf s\right),
        \end{align*}
        for $j = 1, \ldots, N_\Omega$ and $h=1, \ldots, N_\Gamma$. The above equations can be written, in compact matrix-vector form, as
        \begin{align*}
            \mathbb M^\Omega \frac{d \mbf A^p}{dt} &= - \sigma_p \, \mathbb M^\Omega {\mbf A^p} - D_p \, \mathbb K^\Omega  {\mbf A^p} + K_p \, (\mathbb J \, {\mbf a^p} - \mathbb H^\Omega {\mbf A^p}),
            \\
            \mathbb M^\Gamma  \frac{d \mbf a^p}{dt} &= g_p(\mbf a^1, \ldots, \mbf a^n) \, I
            - K_p \, (\mathbb M^\Gamma {\mbf a^p} 
            - \mathbb J^\intercal \mbf A^p)
            - d_p \, \mathbb K^\Gamma {\mbf a^p},
        \end{align*}
        for $p = 1, \ldots, n$, where $I$ is the identity matrix
		
	\subsection{Temporal discretisation}
		The above spatial semi-discrete weak formulation requires the solution of $2 n$ coupled systems of ordinary differential equations. We set a time step $\tau > 0$, look for the solutions at the time points $t_k = t_{k - 1} + \tau$, with $t_0 = 0$, and apply a finite difference scheme. Let ${\mbf a^{p, k}}$ be an approximation of ${\mbf a^p}(t_k)$. Therefore, by replacing the time derivatives by their first order difference operator (for example), we aim at solving the following system on the surface
		\begin{align*}
            \left(\tau^{- 1} \mathbb M^\Gamma
		    + d_p \, \mathbb K^\Gamma\right)
		    {\mbf a^{p, k + 1}}
		    &= \tau^{- 1} \mathbb M^\Gamma {\mbf a^{p, k}}
		    + g_p\left(\mbf a^{1, k}, \ldots, \mbf a^{n, k}\right) \, I
            \\
            &\hspace{3cm} - K_p \left(\mathbb M^\Gamma {\mbf a^{p, k}} - \mathbb J^\intercal \mbf A^{p, k} \right),
		\end{align*}
		for $p = 1, \ldots, n$. The above temporal discretisation scheme is known as IMEX, because the diffusion terms are considered implicitly and the reaction terms explicitly \cite{ruuth1995implicit,madzvamuse2006time}. The explicit treatment of the reaction terms constitutes a way of obtaining a linear system of equations. By solving the above systems, we get the solution of the surface equations at time $t_{k + 1}$. To obtain the solution of the bulk equations at time $t_{k + 1}$, we solve the following systems
		\begin{align}
            \left(\tau^{- 1} \mathbb M^\Omega + D_p \, \mathbb K^\Omega + K_p \, \mathbb H^\Omega + \sigma_p \, \mathbb M^\Omega\right) {\mbf A^{p, k + 1}} = \tau^{- 1} \mathbb M^\Omega {\mbf A^{p, k}} + K_p \, \mathbb J \, {\mbf a^{p, k + 1}}, 
		\end{align}
        for $p = 1, \ldots, n$. By taking advantage of the linear reactions in the bulk, and having already calculated the surface solutions at $t_{k + 1}$, we can solve the bulk systems fully implicitly.
    
    \newpage

\section{Amplitude equations up to order 3 for $\ell = 4$} \label{ap:order3ell=4}
    \begin{align}
        \label{l4-appendixa0}
        \dot A_0 &= A_0 \left(\varepsilon - 18 \, a \, A_0\right) + 18 \, a \, A_{-1} \, A_1 + 22 \, a \, A_{-2} \, A_2 
        \\
        & - 42 \, a \, A_{-3} \, A_3 -28 \, a \, A_{-4} \, A_4 + A_0 \left(2 \, c \, A_{- 4} \, A_4 + (140 \, b - 2 \, c) \, A_{- 3} \, A_3 \right. \notag
        \\
        & \left. + (- 600 \, b + 2 \, c) \, A_{- 2} \, A_2 + (320 \, b - 2 \, c) \, A_{- 1} \, A_1 \right. \notag
        \\
        & \left. + (-160 \, b + c) \, A_0^2\right) - 14 \, \sqrt{10} \, b \, A_{- 4} \, A_1 \, A_3 \notag
        \\
        & - 12 \, \sqrt{70} \, b \, A_{- 4} \, A_2^2 - 14 \, \sqrt{10} \, b \, A_{- 3} \, A_{- 1} \, A_4 + 23 \, \sqrt{70} \, b \, A_{- 3} \, A_1 \, A_2 \notag
        \\
        & - 12 \, \sqrt{70} \, b \, A_{- 2}^2 \, A_4 + 23 \, \sqrt{70} \, b \, A_{- 2} \, A_{- 1} \, A_3 + 21 \, \sqrt{10} \, b \, A_{- 2} \, A_1^2 \notag
        \\
        & + 21 \, \sqrt{10} \, b \, A_{- 1}^2 \, A_2 + \mathcal O\left(\norm{\rtext{\mbf A_4}}^4\right), \notag
    \end{align}
    \begin{align}
        \label{l4-appendixa1}
        \dot A_1 &= A_1 \left(\varepsilon - 18 \, a \, A_0\right) + 12 \, \sqrt{10} \, a \, A_{-1} \, A_2 - 2 \, \sqrt{70} \, a \, A_{-2} \, A_3
        \\
        & - 14 \, \sqrt{10} \, a \, A_{- 3} \, A_4  + A_1 \, \left((- 28 \, b + 2 \, c) \, A_{- 4} \, A_4 + (273 \, b - 2 \, c) \, A_{- 3} \, A_3 \right. \notag
        \\
        & \left. + (- 313 \, b + 2 \, c) \, A_{- 2} \, A_2 + (383 \, b - 2 \, c) \, A_{- 1} \, A_1 + (- 160 \, b + c) \, A_0^2\right) \notag
        \\
        & - 210 \, b \, A_{- 4} \, A_2 \, A_3 + 14 \, \sqrt{10} \, b \, A_{- 3} \, A_0 \, A_4 + 45 \, \sqrt{7} \, b \, A_{- 3} \, A_2^2 \notag
        \\
        & - 40 \, \sqrt{7} \, b \, A_{- 2} \, A_{- 1} \, A_4  - 23 \, \sqrt{70} \, b \, A_{- 2} \, A_0 \, A_3 + 90 \, \sqrt{7} \, b \, A_{- 1}^2 \, A_3 \notag
        \\
        & - 42 \, \sqrt{10} \, b \, A_{- 1} \, A_0 \, A_2 + \mathcal O\left(\norm{\rtext{\mbf A_4}}^4\right), \notag
    \end{align}
    \begin{align}
        \dot A_2 &= A_2 \left(\varepsilon + 22 \, a \, A_0\right) - 6 \, \sqrt{10} \, a \, A_1^2 + 2 \, \sqrt{70} \, a \, A_{-1} \, A_3 - 6 \, \sqrt{70} \, a \, A_{- 2} \, A_4 
        \label{l4-appendixa2}
        \\
        & + A_2 \, \left((- 168 \, b + 2 \, c) \, A_{- 4} \, A_4 + (343 \, b - 2 \, c) \, A_{- 3} \, A_3 \right. \notag
        \\
        & \left. + (- 208 \, b + 2 \, c) \, A_{- 2} \, A_2+ (313 \, b - 2 \, c) \, A_{- 1} \, A_1 + (- 300 \, b + c) \, A_0^2\right) \notag
        \\
        & - 70 \, \sqrt{7} \, b \, A_{- 4} \, A_3^2  + 210 \, b \, A_{- 3} \, A_1 \, A_4 - 24 \, \sqrt{70} \, b \, A_{- 2} \, A_0 \, A_4 \notag
        \\
        & - 90 \, \sqrt{7} \, b \, A_{- 2} \, A_1 \, A_3  + 20 \, \sqrt{7} \, b \, A_{- 1}^2 \, A_4 + 23 \, \sqrt{70} \, b \, A_{- 1} \, A_0 \, A_3 \notag
        \\
        & + 21 \, \sqrt{10} \, b \, A_0 \, A_1^2 + \mathcal O\left(\norm{\rtext{\mbf A_4}}^4\right), \notag  
    \end{align}
    \begin{align}
        \label{l4-appendixa3}
        \dot A_3 &= A_3 \left(\varepsilon + 42 \, a \, A_0\right) - 14 \, \sqrt{10} \, a \, A_{-1} \, A_4 - 2 \, \sqrt{70} \, a \, A_1 \, A_2
        \\
        & + A_3 \, \left((- 448 \, b + 2 \, c) \, A_{- 4} \, A_4 + (203 \, b - 2 \, c) \, A_{- 3} \, A_3 \right. \notag
        \\
        & \left. + (- 343 \, b + 2 \, c) \, A_{- 2} \, A_2 + (273 \, b - 2 \, c) \, A_{- 1} \, A_1 + (- 70 \, b + c) \, A_0^2\right) \notag
        \\
        & + 140 \, \sqrt{7} \, b \, A_{- 3} \, A_2 \, A_4  - 210 \, b \, A_{- 2} \, A_1 \, A_4 + 14 \, \sqrt{10} \, b \, A_{- 1} \, A_0 \, A_4 \notag
        \\
        & + 45 \, \sqrt{7} \, b \, A_{- 1} \, A_2^2 - 23 \, \sqrt{70} \, b \, A_0 \, A_1 \, A_2 + 30 \, \sqrt{7} \, b \, A_1^3 + \mathcal O\left(\norm{\rtext{\mbf A_4}}^4\right), \notag
    \end{align}
    \begin{align}
        \label{l4-appendixa4}
        \dot A_4 &= A_4 \left(\varepsilon - 28 \, a \, A_0\right) + 14 \, \sqrt{10} \, a \, A_1 \, A_3 - 3 \, \sqrt{70} \, a \, A_2^2
        \\
        & + A_4 \, \left((- 448 \, b + 2 \, c) \, A_{- 4} \, A_4 + (448 \, b - 2 \, c) \, A_{- 3} \, A_3 \right. \notag
        \\
        & \left.+ (- 168 \, b + 2 \, c) \, A_{- 2} \, A_2 + (28 \, b - 2 \, c) \, A_{- 1} \, A_1 + c \, A_0^2\right) - 70 \, \sqrt{7} \, b \, A_{- 2} \, A_3^2 \notag
        \\
        & + 210 \, b \, A_{- 1} \, A_2 \, A_3 - 14 \, \sqrt{10} \, b \, A_0 \, A_1 \, A_3 - 12 \, \sqrt{70} \, b \, A_0 \, A_2^2 \notag
        \\
        & + 20 \, \sqrt{7} \, b \, A_1^2 \, A_2 + \mathcal O\left(\norm{\rtext{\mbf A_4}}^4\right). \notag
    \end{align}
        
    \newpage

    \section{Plots of the $L^2$-norm of transient solutions}
        
    \begin{figure}[hb]
        \centering
        \includegraphics[width = 0.49 \textwidth]{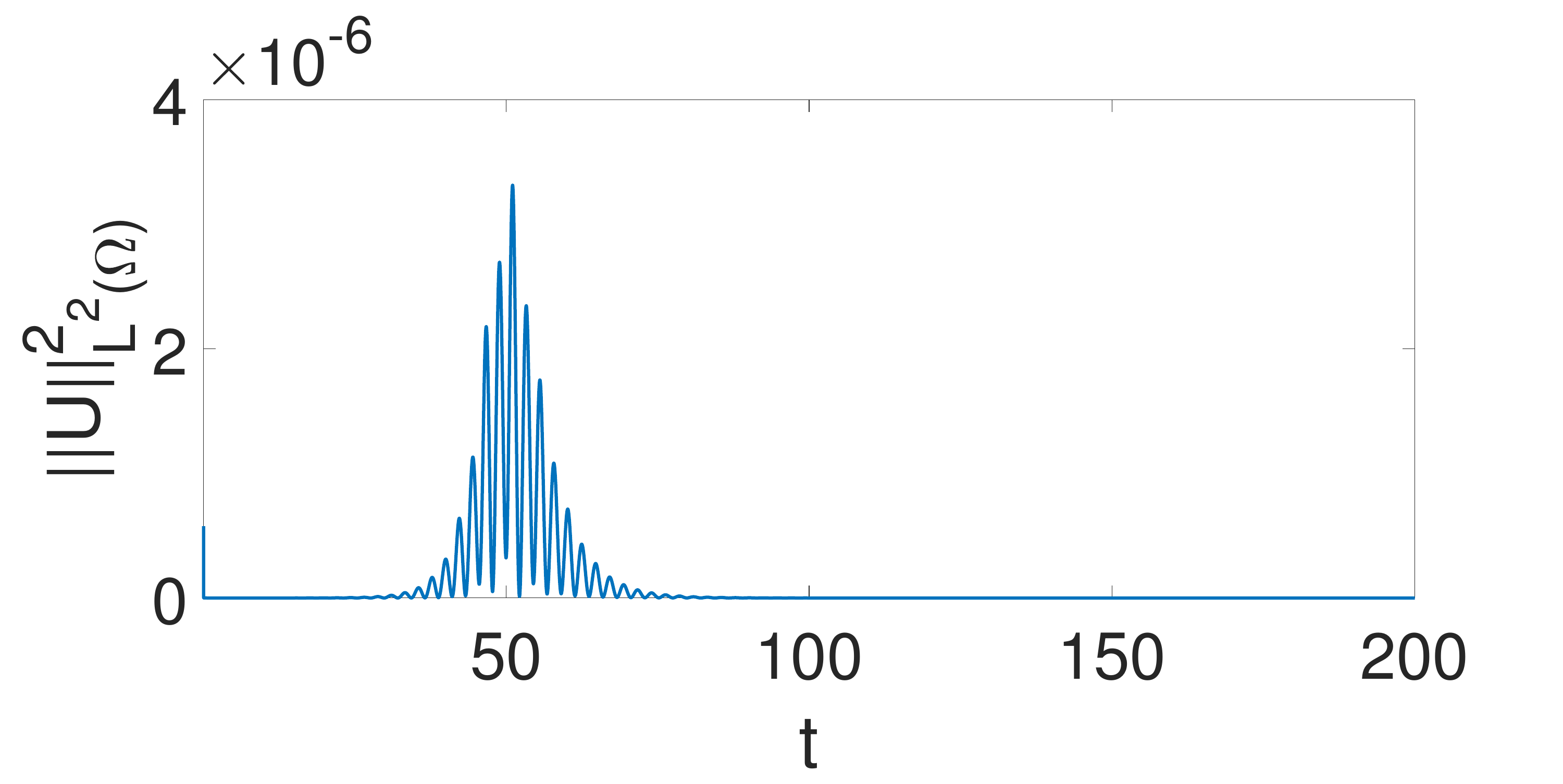} \hfill \includegraphics[width = 0.49 \textwidth]{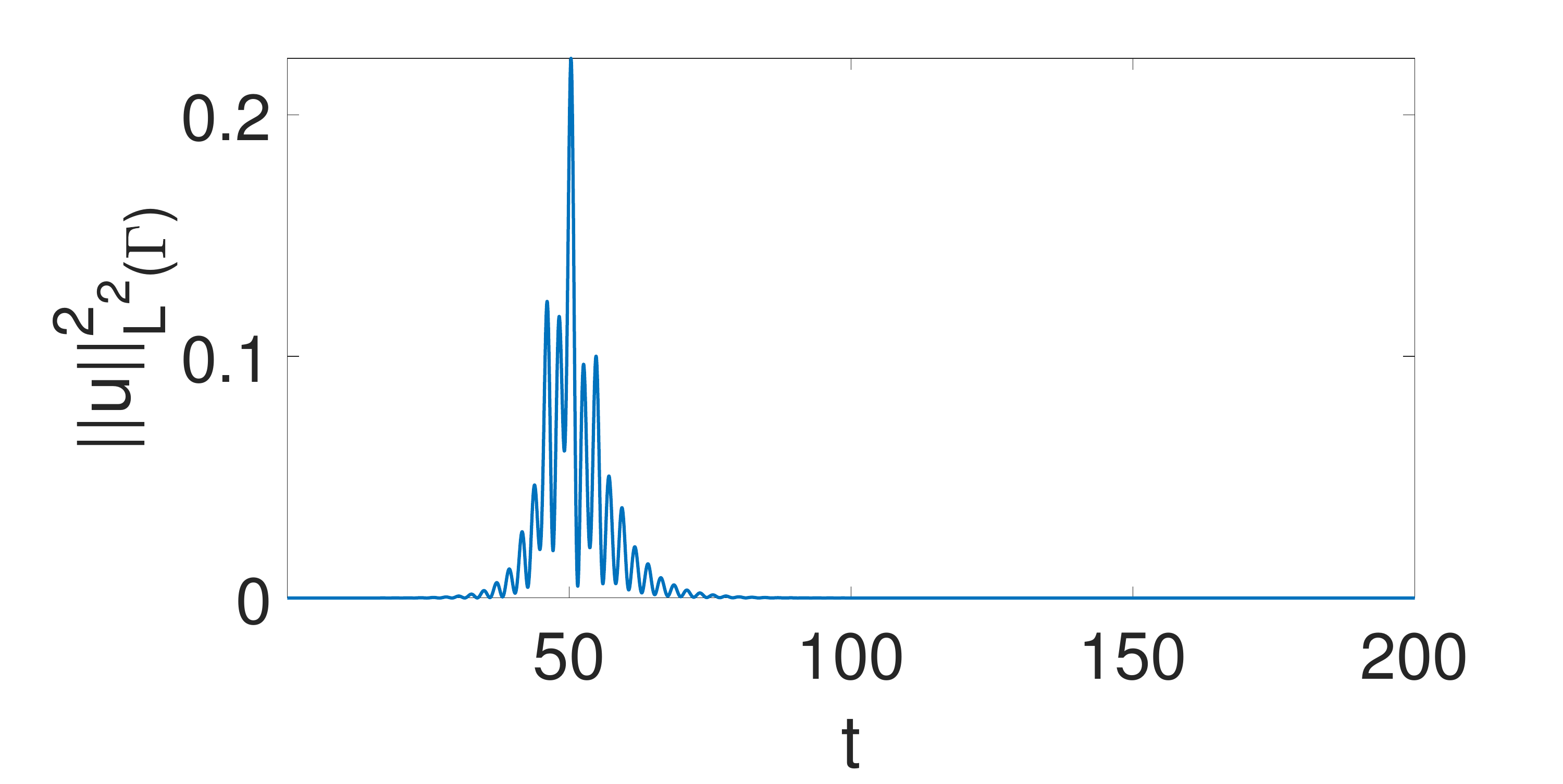}
        \\
        \includegraphics[width = 0.49 \textwidth]{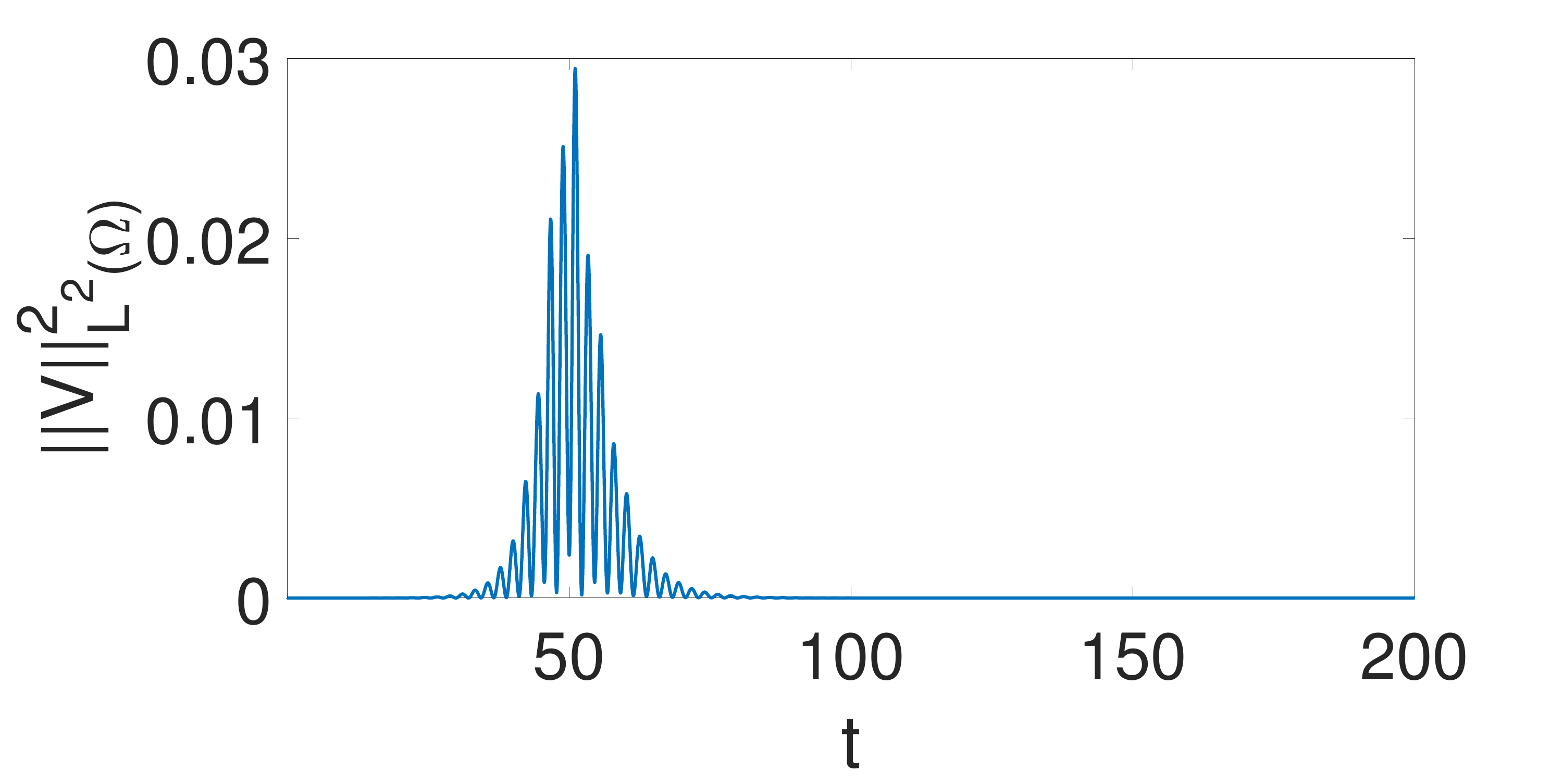} \hfill \includegraphics[width = 0.49 \textwidth]{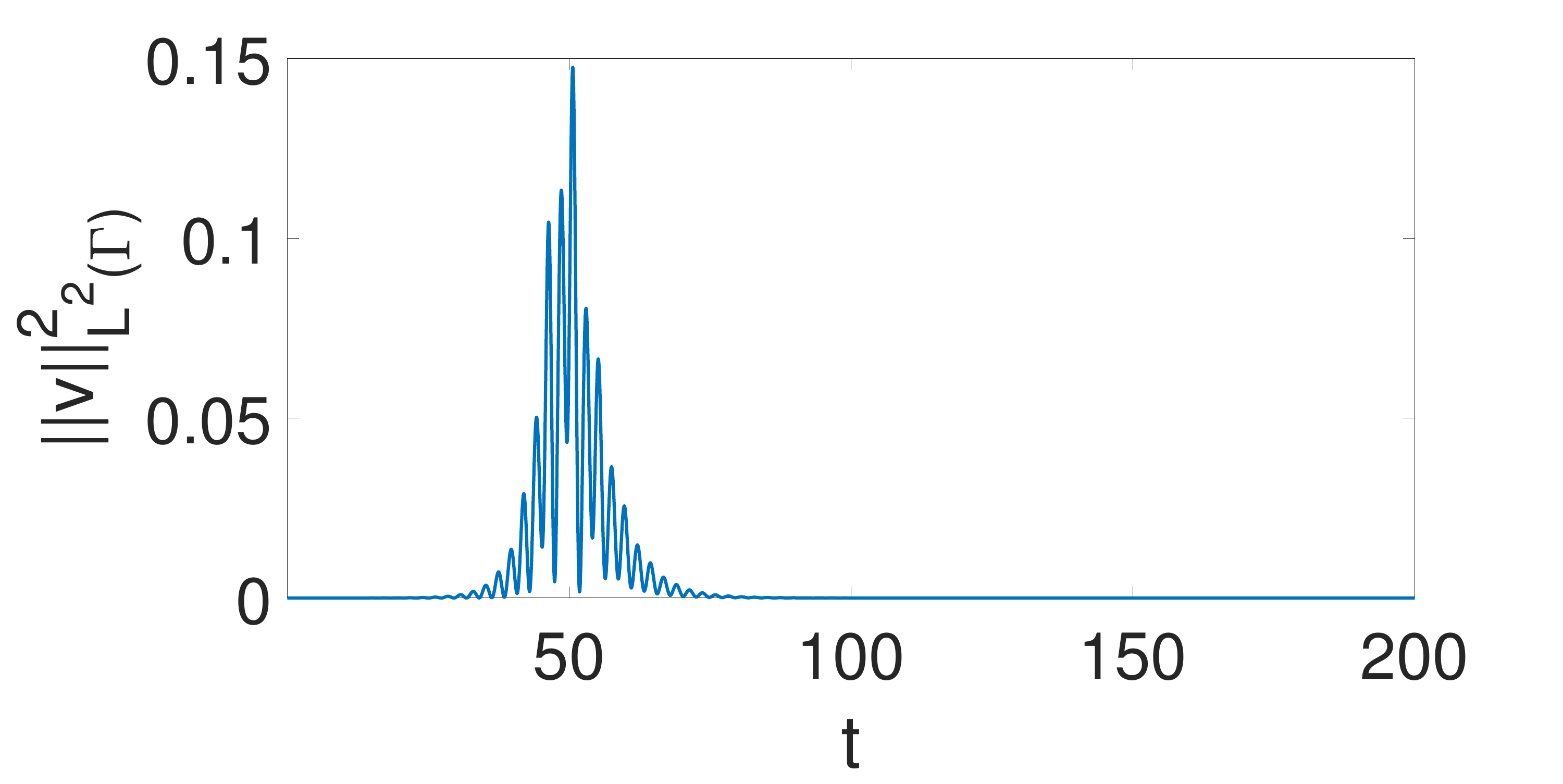}
        \caption{Brusselator for $\ell = 1$.}
    \end{figure}
    \begin{figure}[!b]
        \centering
        \includegraphics[width = 0.49 \textwidth]{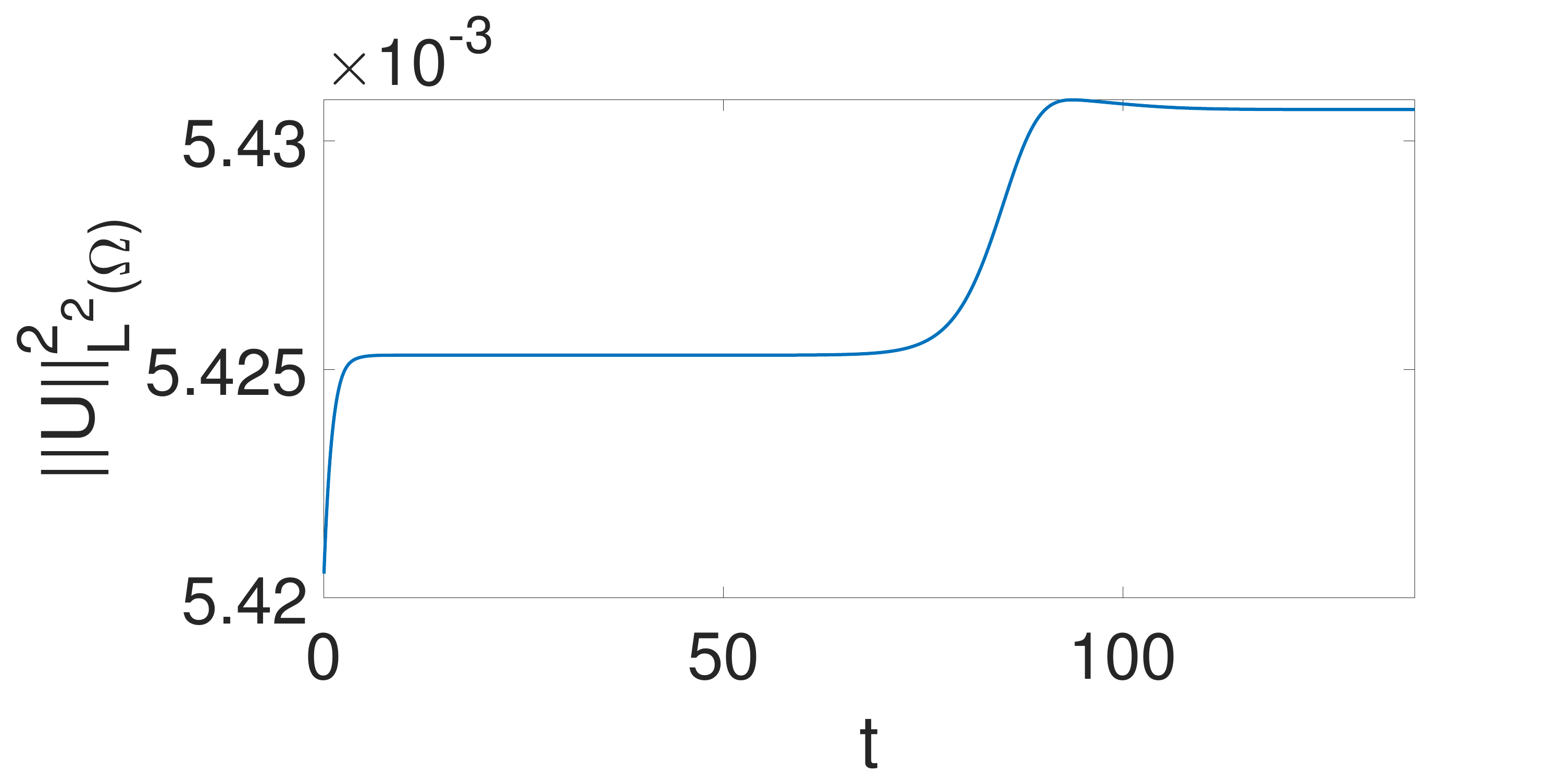} \hfill \includegraphics[width = 0.49 \textwidth]{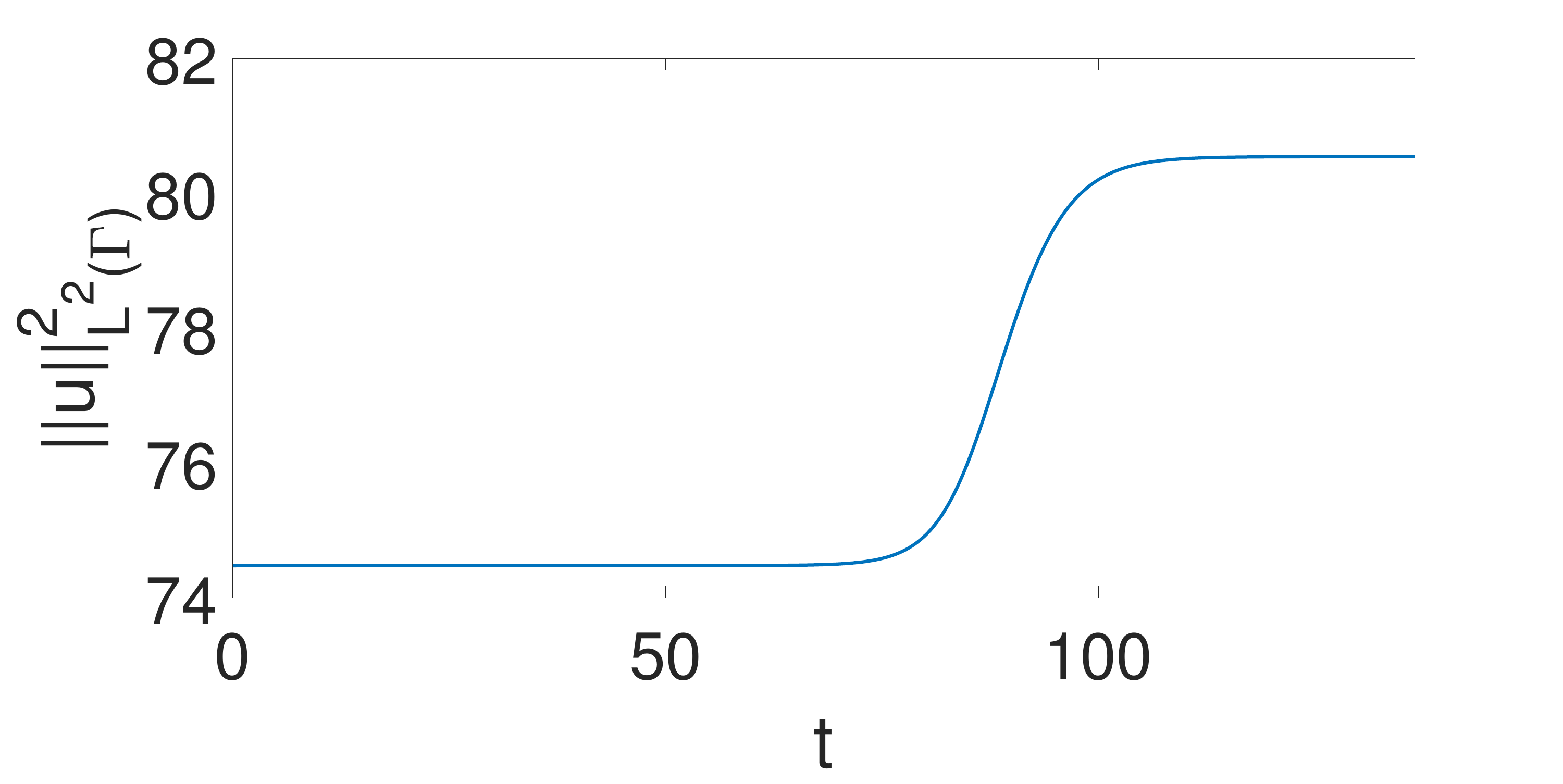}
        \\
        \includegraphics[width = 0.49 \textwidth]{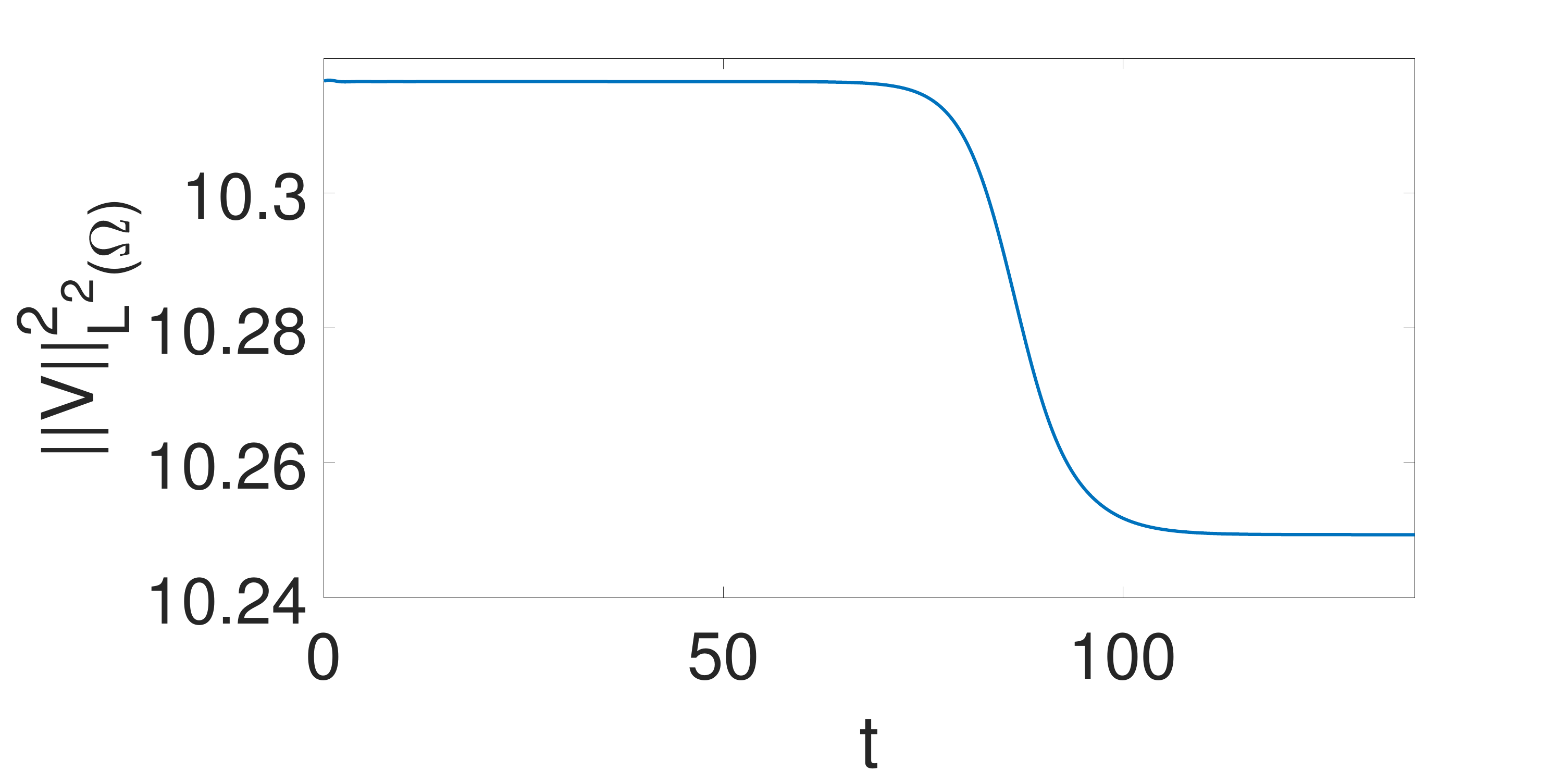} \hfill \includegraphics[width = 0.49 \textwidth]{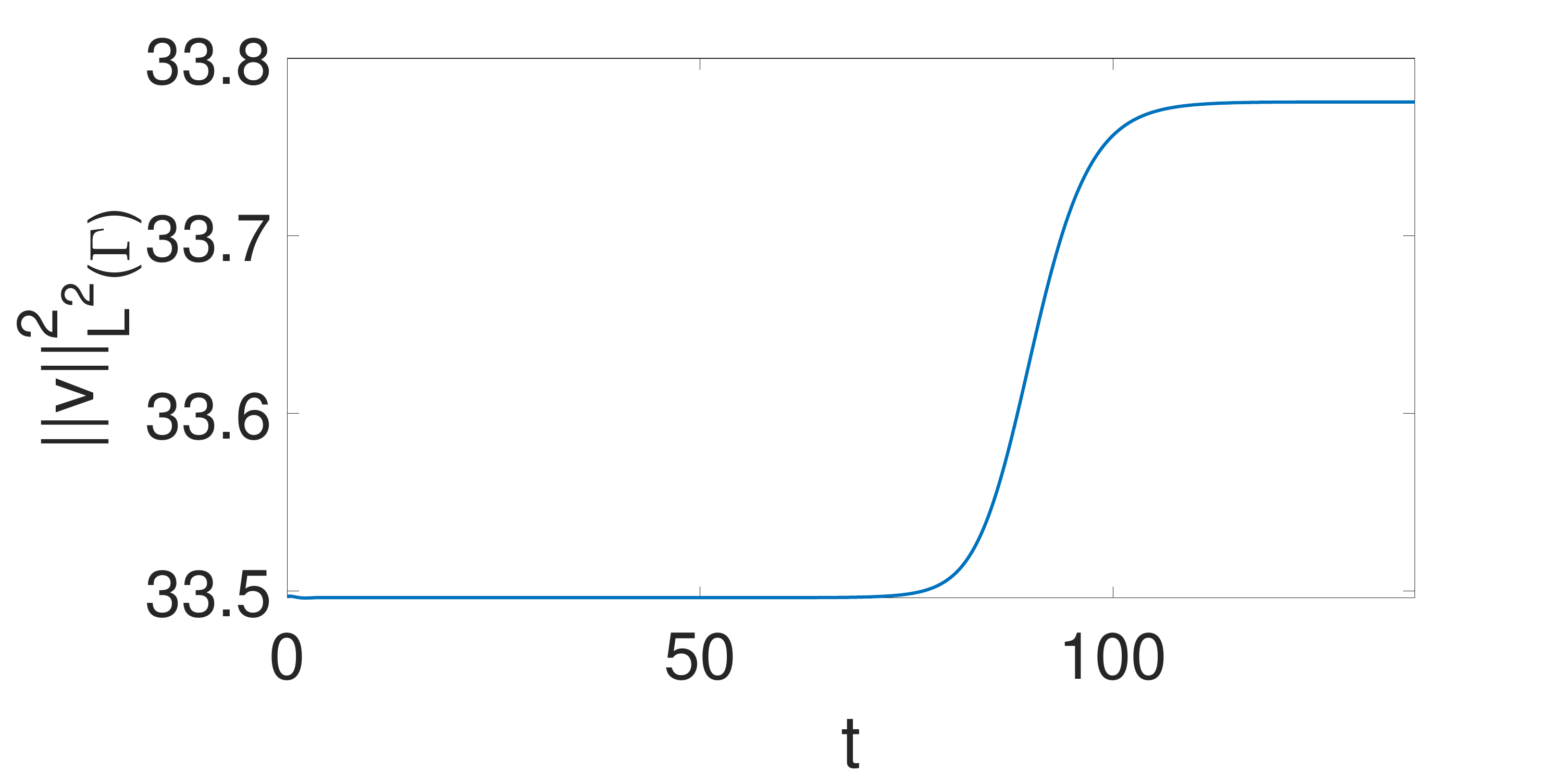}
        \caption{Brusselator for $\ell = 2$.}
        \label{fig:time-series-bru2}
    \end{figure}

    \newpage
    
    \ 
    
    \begin{figure}[hb]
        \centering
        \includegraphics[width = 0.49 \textwidth]{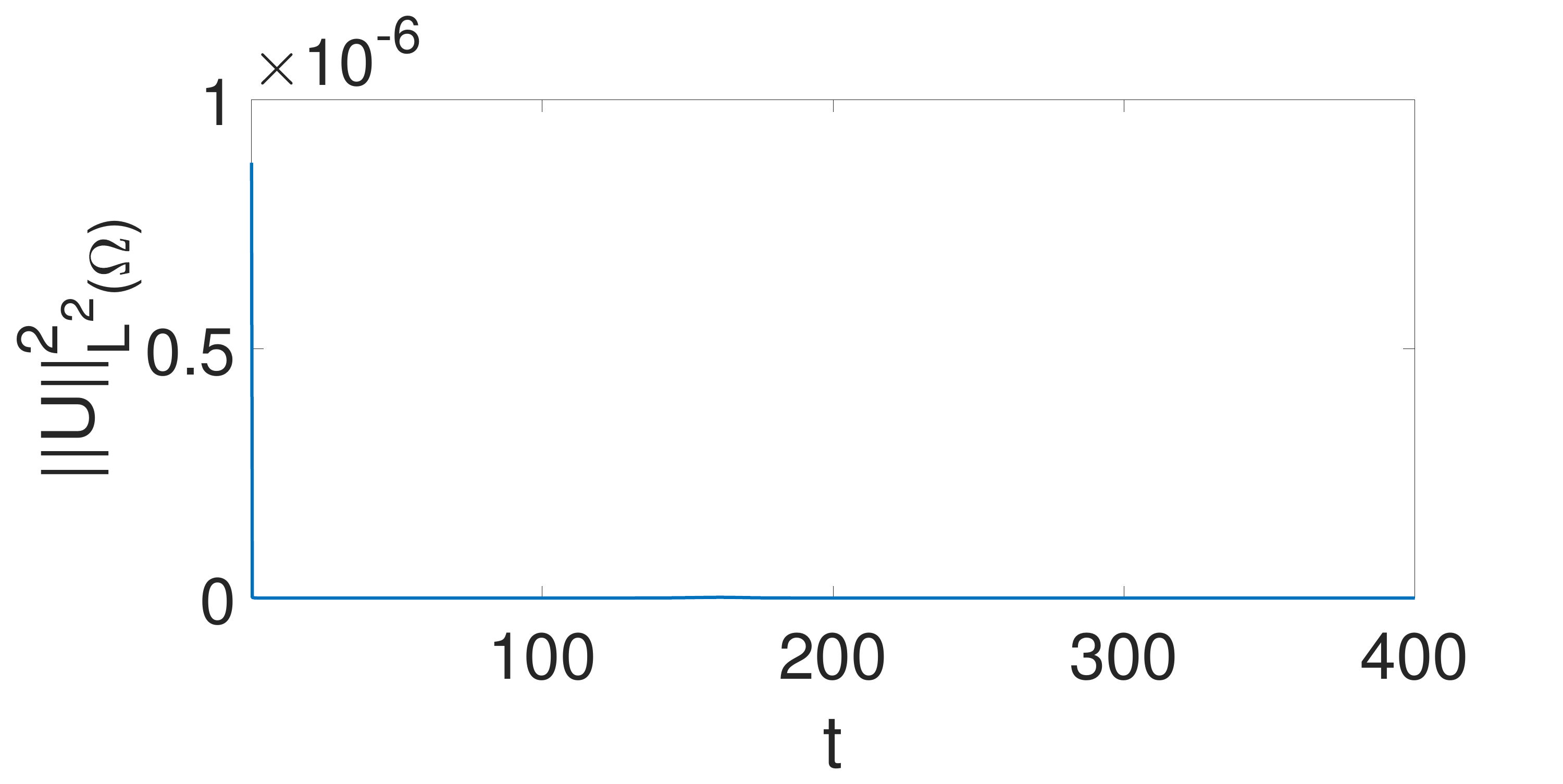} \hfill \includegraphics[width = 0.49 \textwidth]{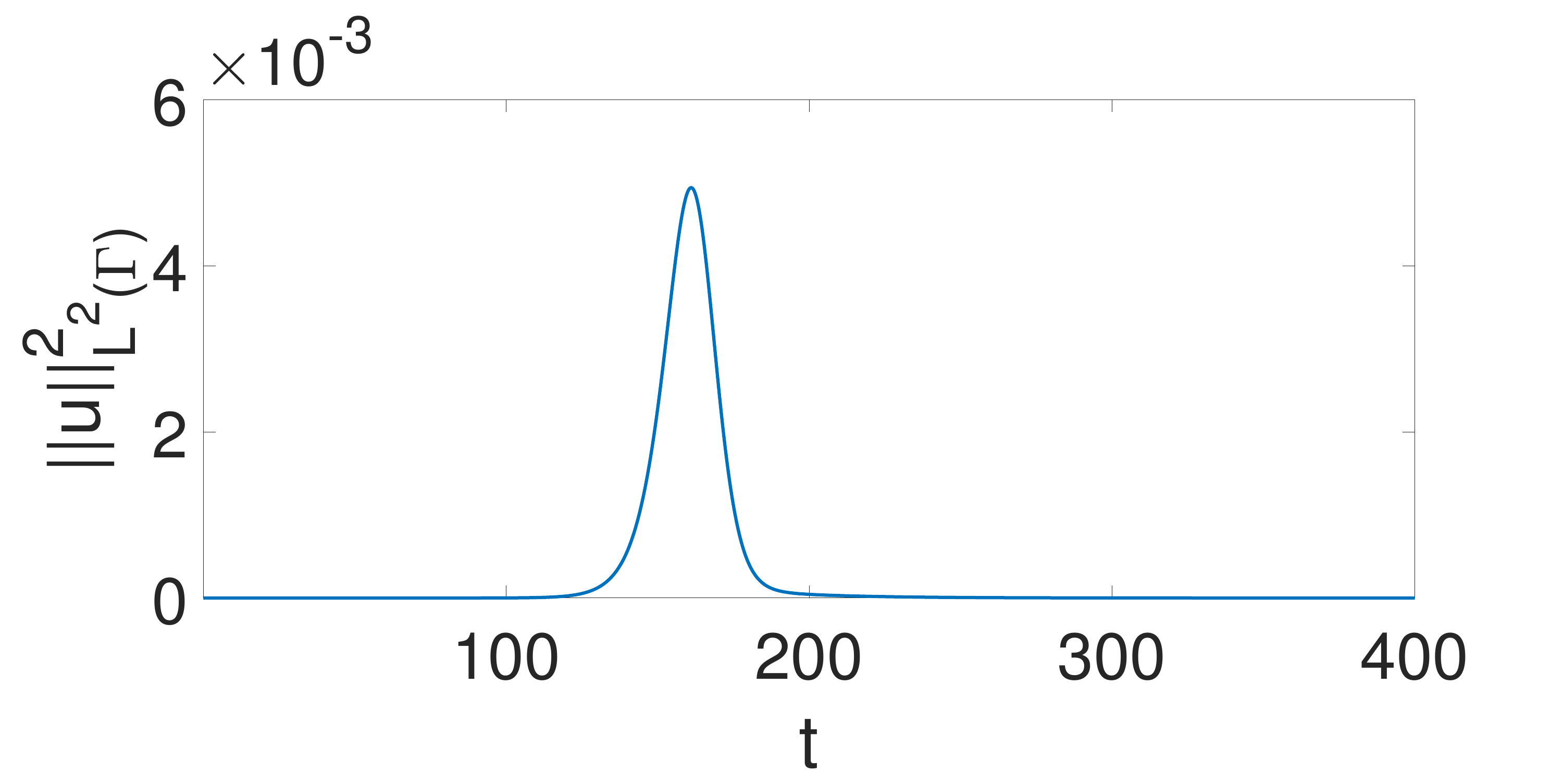}
        \\
        \includegraphics[width = 0.49 \textwidth]{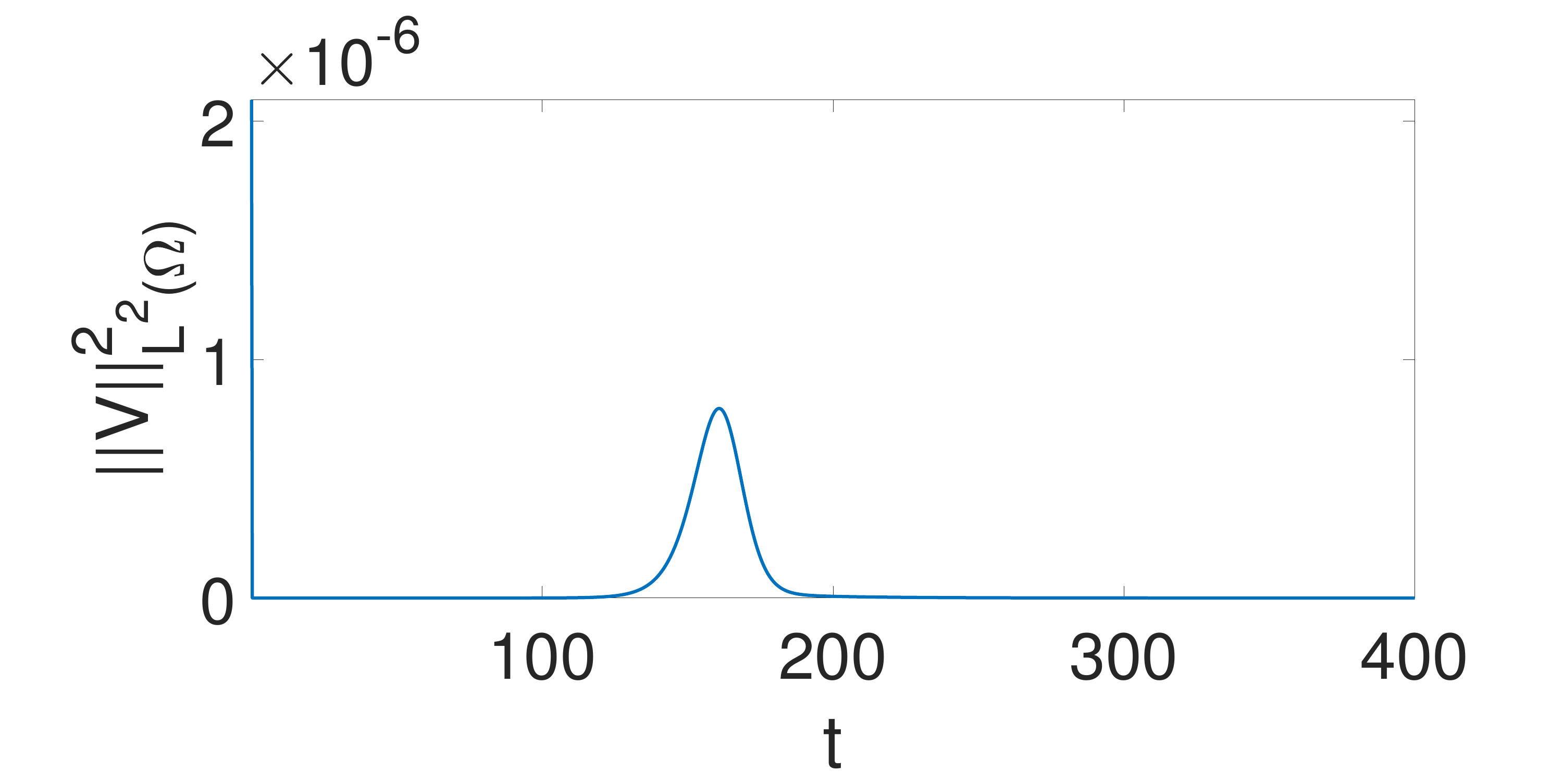} \hfill \includegraphics[width = 0.49 \textwidth]{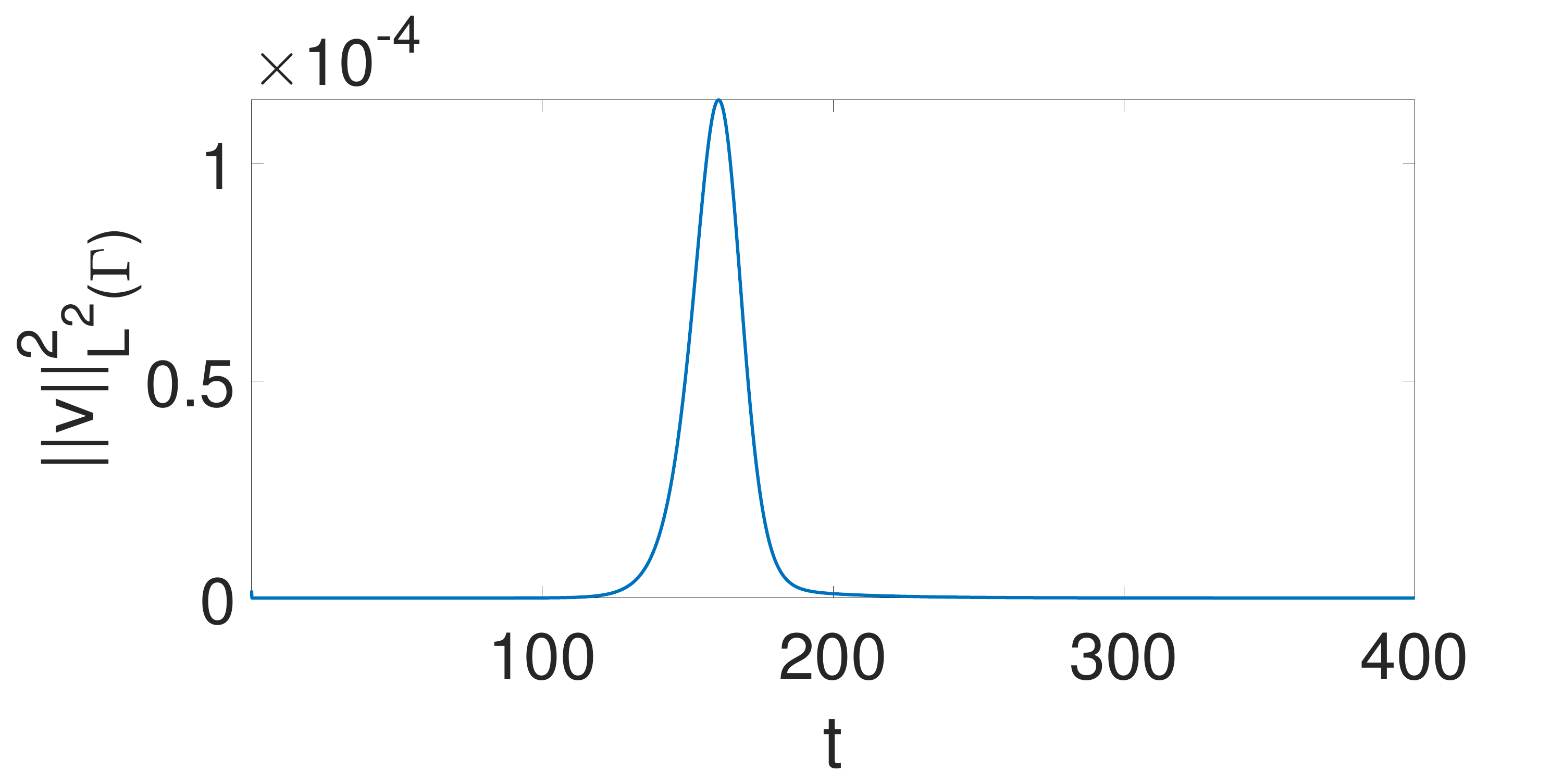}
        \caption{Brusselator for $\ell = 3$.}
    \end{figure}
    \begin{figure}[!b]
        \centering
        \includegraphics[width = 0.49 \textwidth]{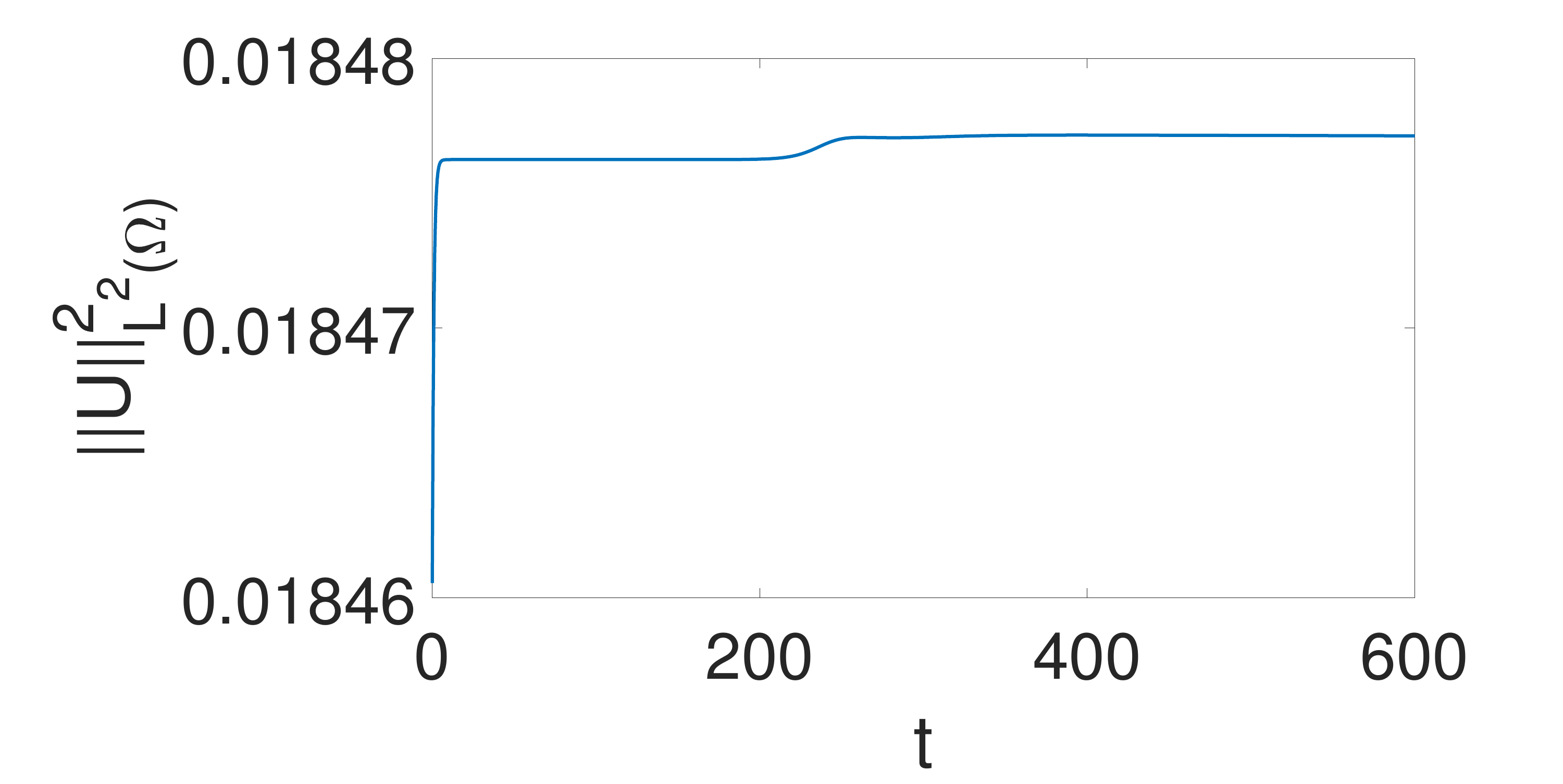} \hfill \includegraphics[width = 0.49 \textwidth]{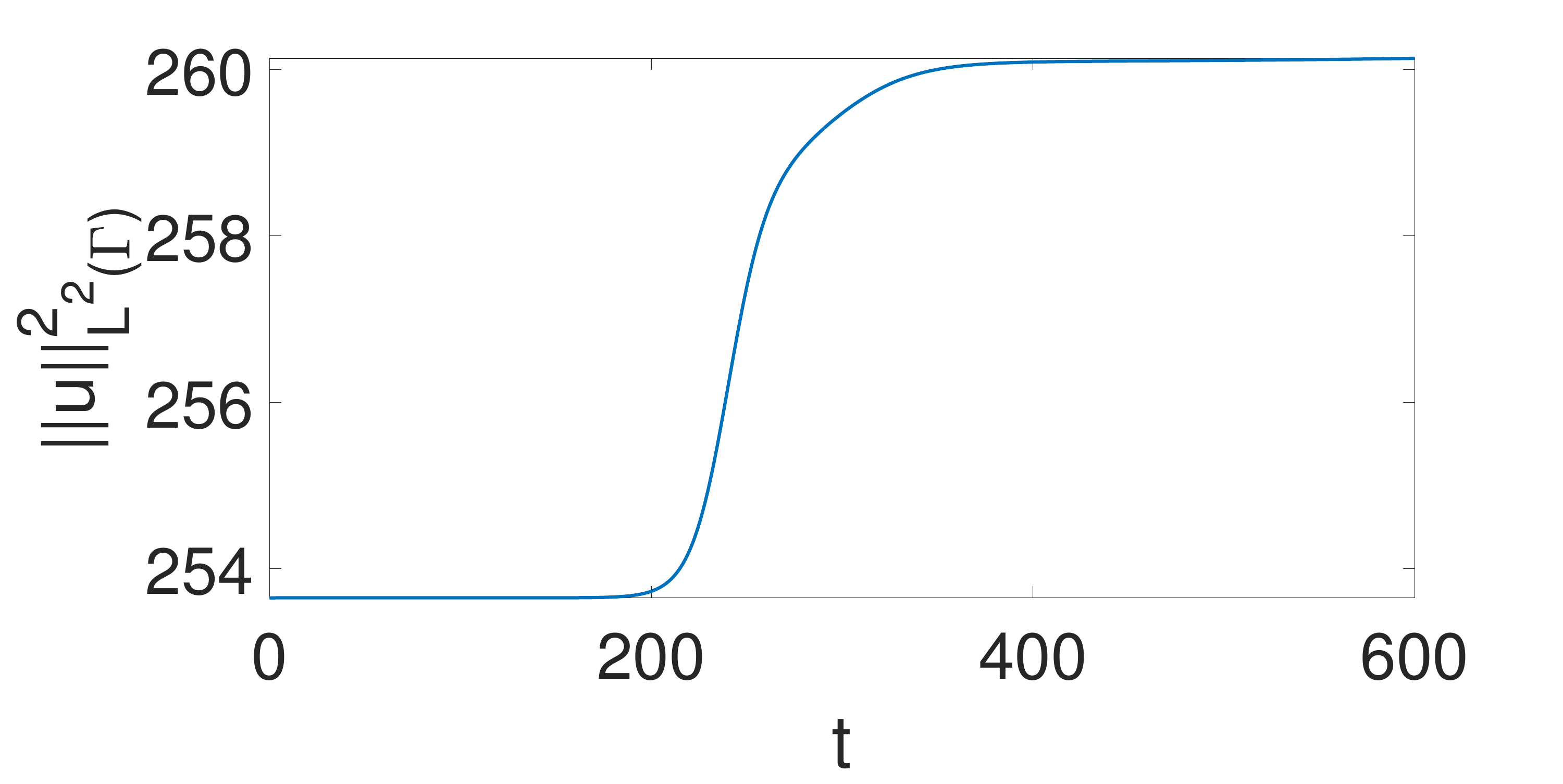}
        \\
        \includegraphics[width = 0.49 \textwidth]{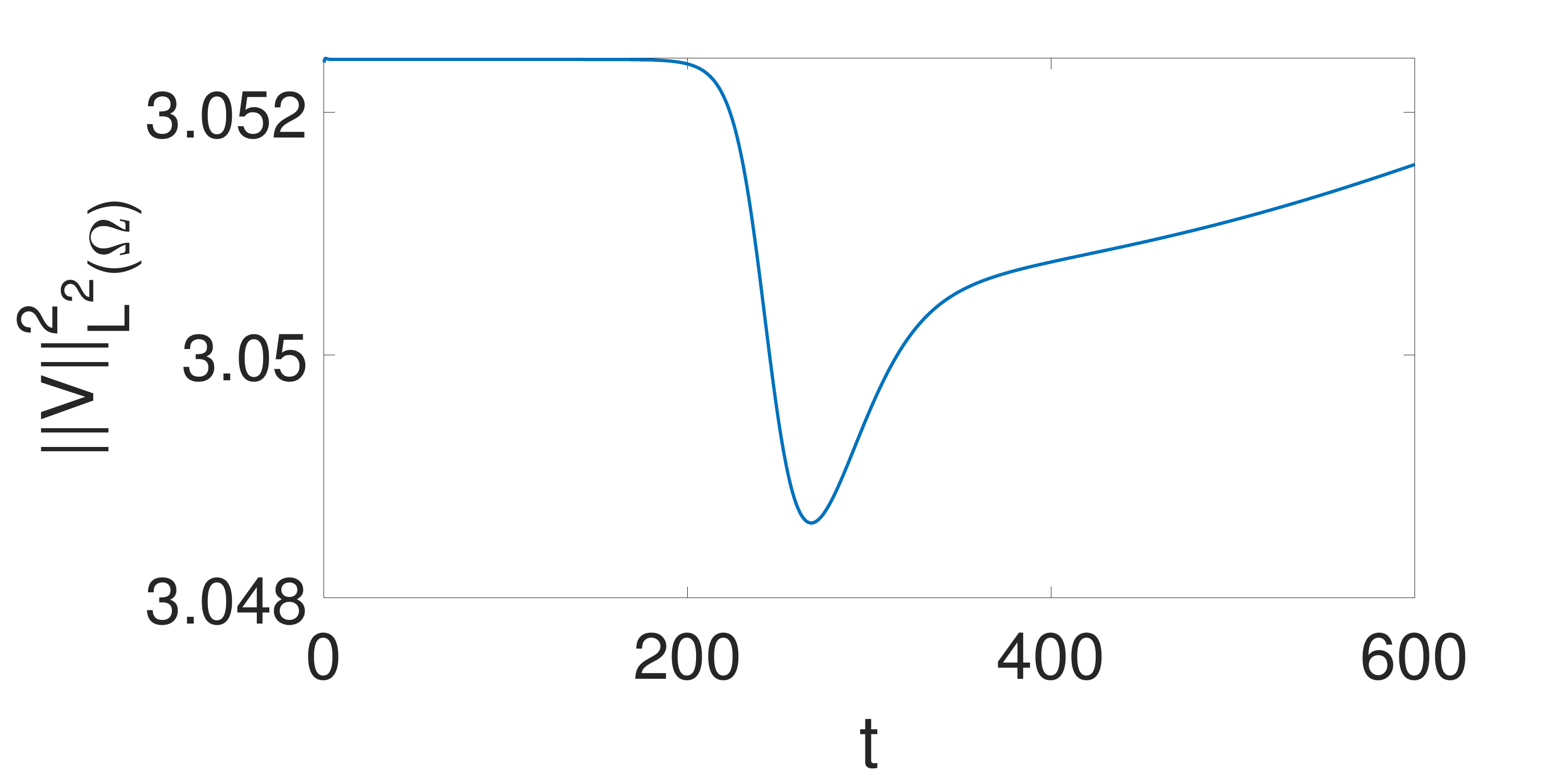} \hfill \includegraphics[width = 0.49 \textwidth]{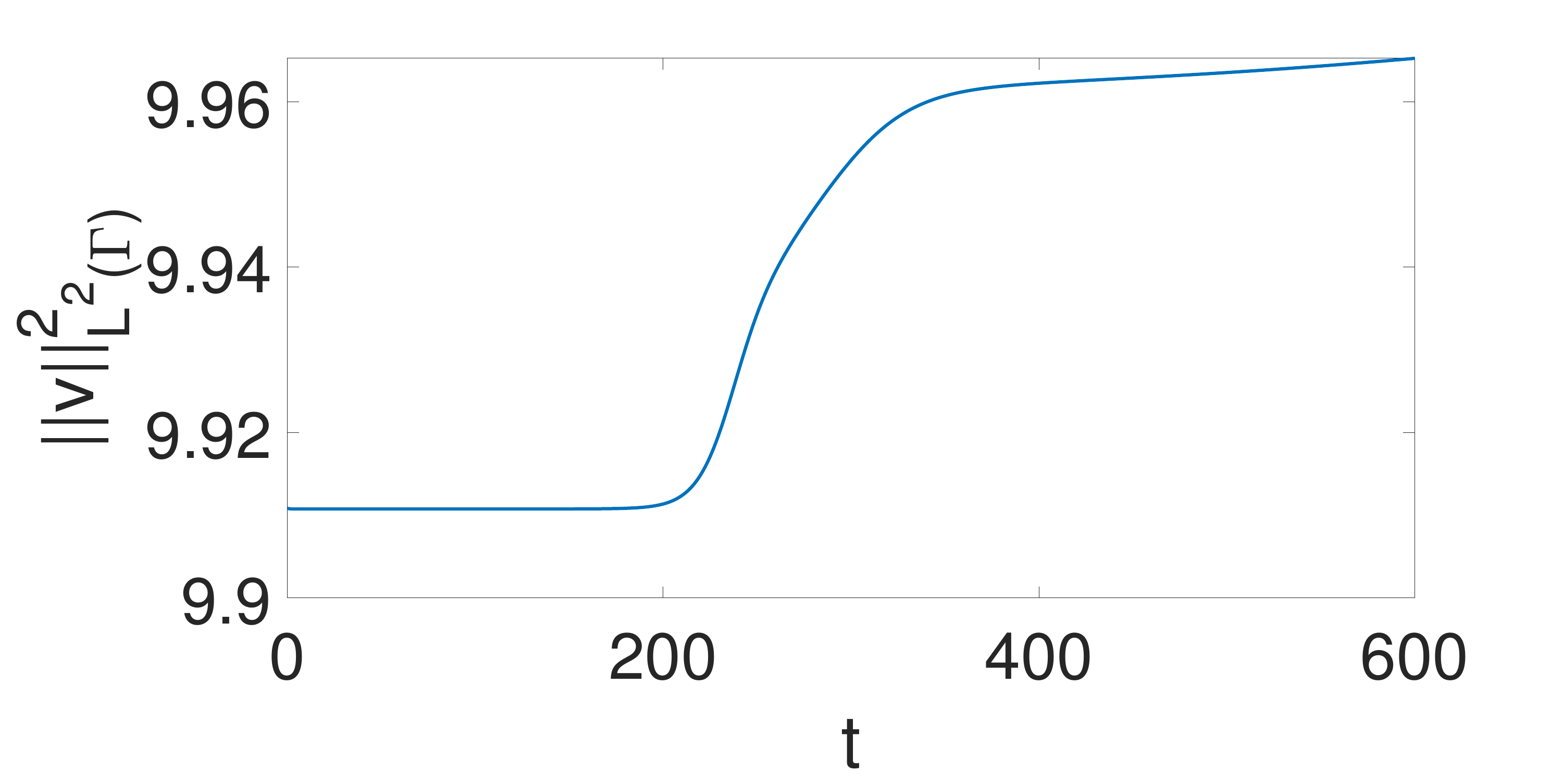}
        \caption{Brusselator for $\ell = 4$.}
    \end{figure}

    \newpage

    \
    
    \begin{figure}[hb]
        \centering
        \includegraphics[width = 0.49 \textwidth]{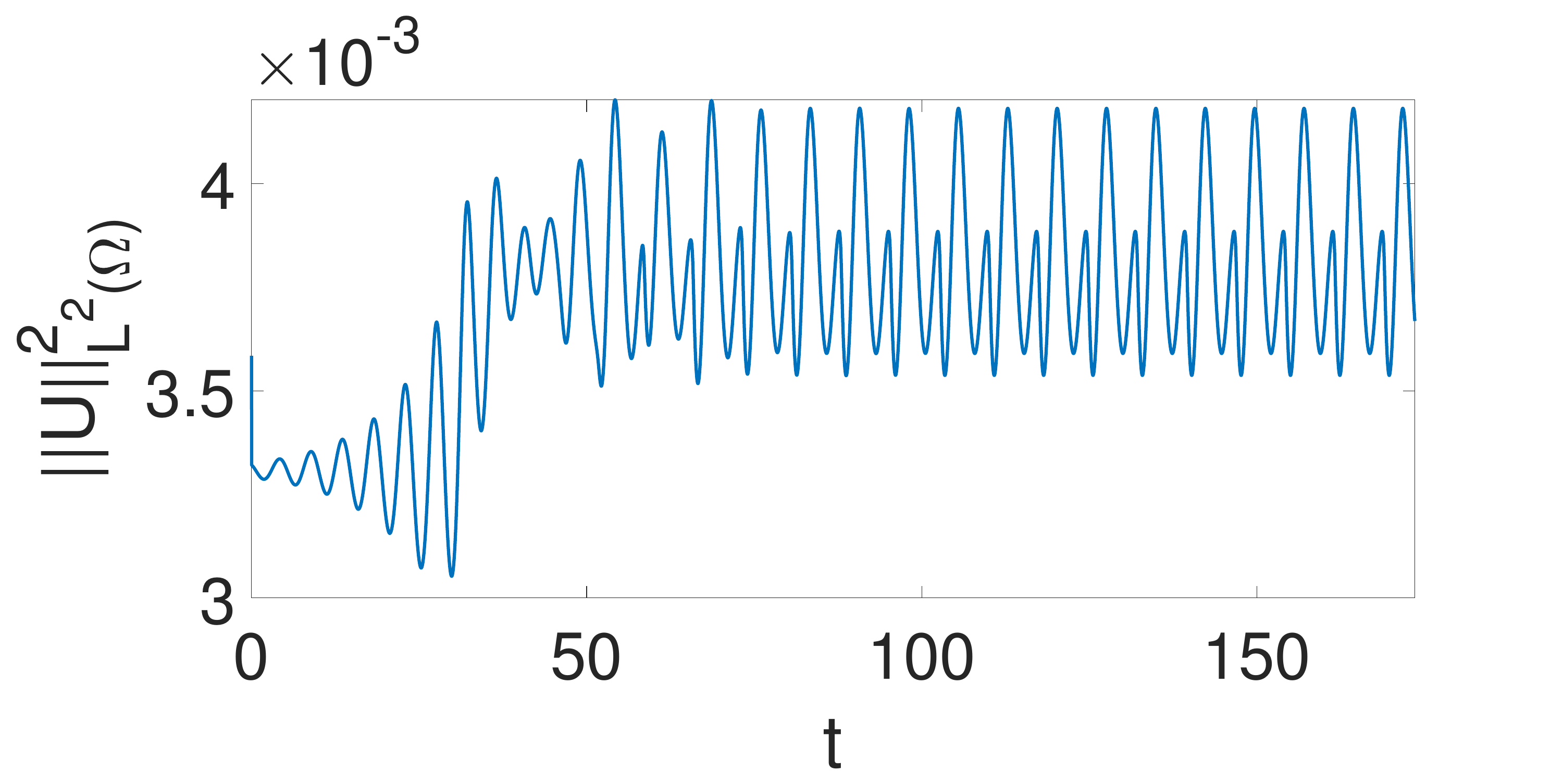} \hfill \includegraphics[width = 0.49 \textwidth]{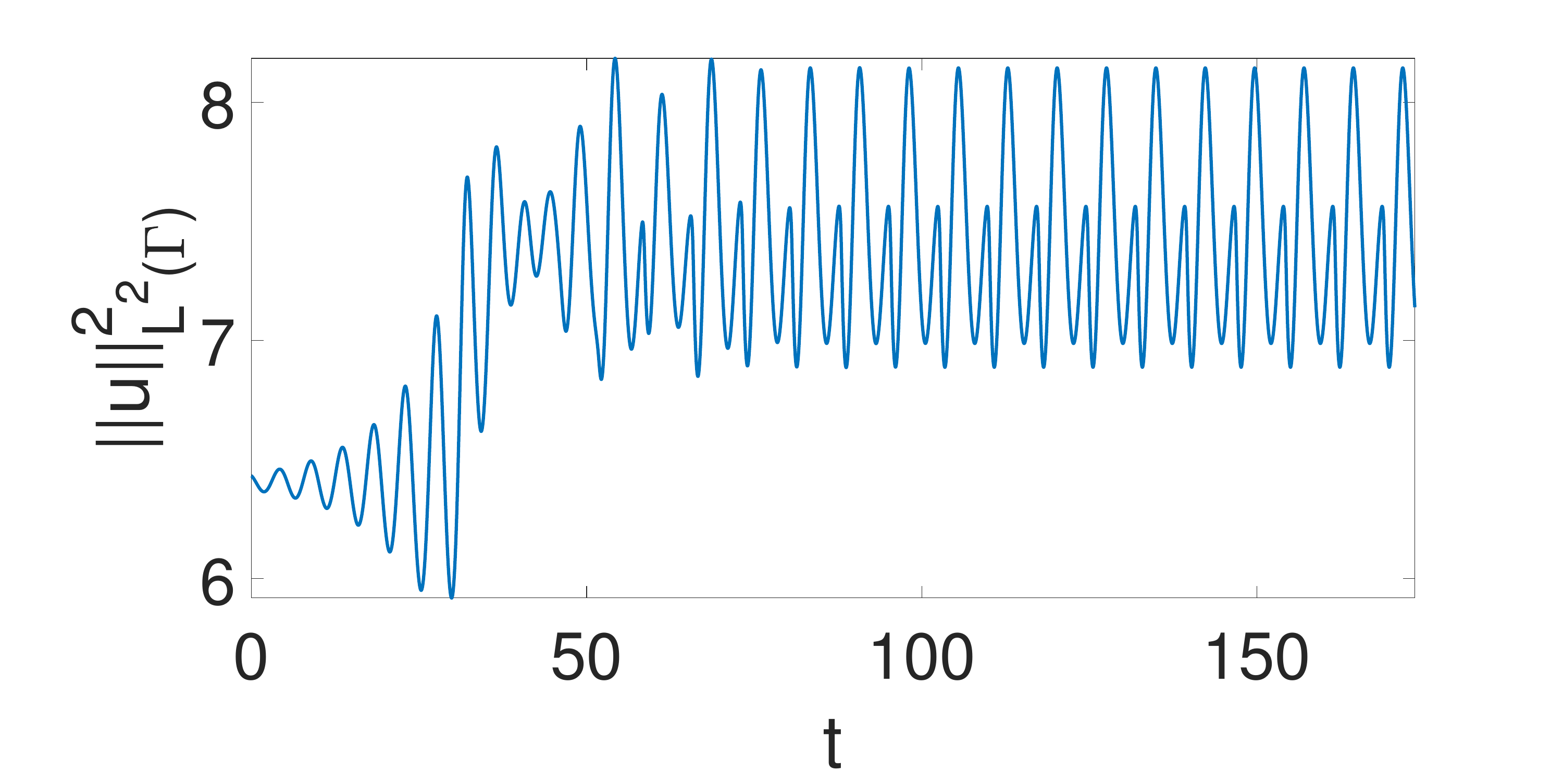}
        \\
        \includegraphics[width = 0.49 \textwidth]{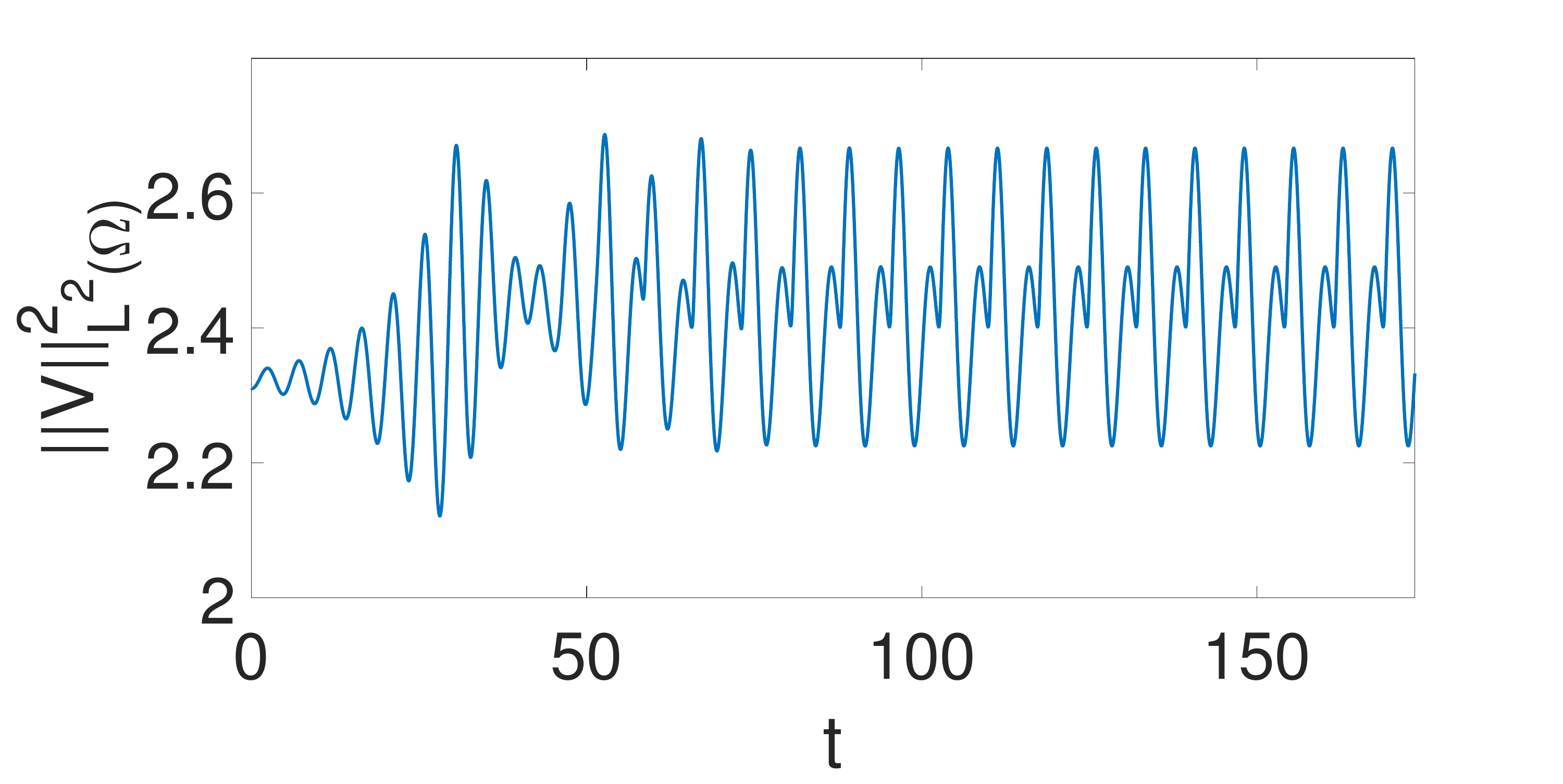} \hfill \includegraphics[width = 0.49 \textwidth]{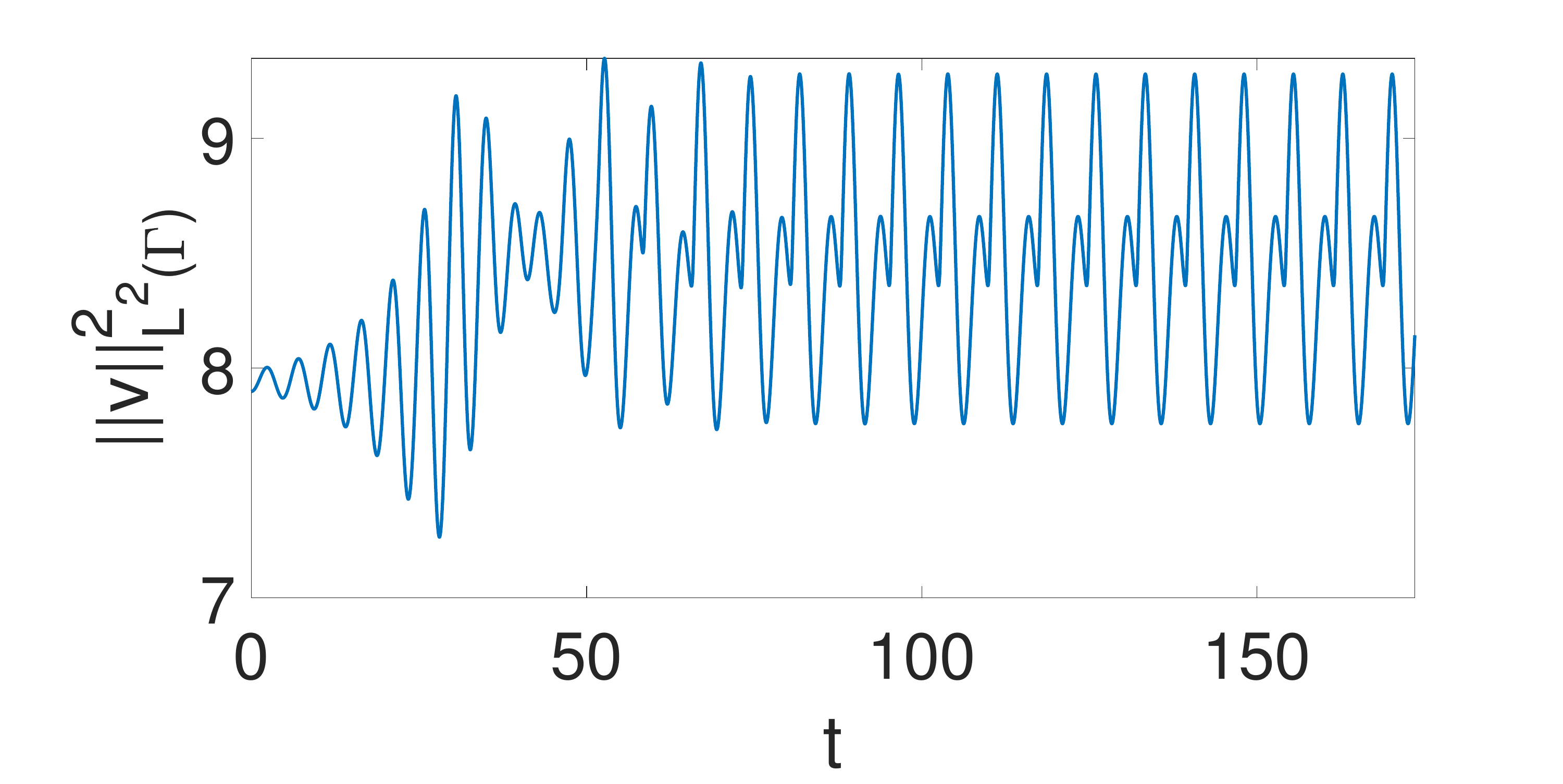}
        \\
        \includegraphics[width = 0.49 \textwidth]{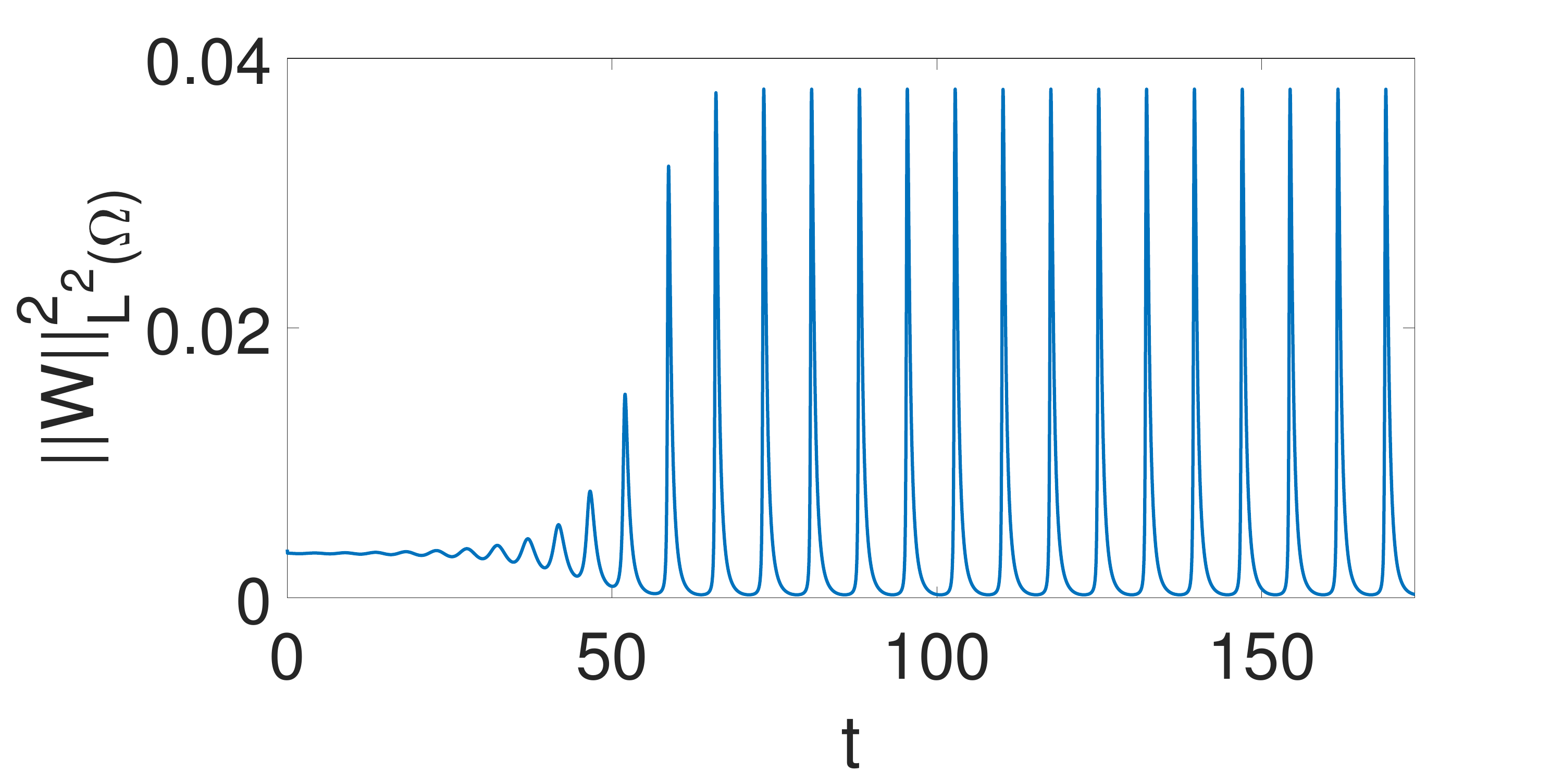} \hfill \includegraphics[width = 0.49 \textwidth]{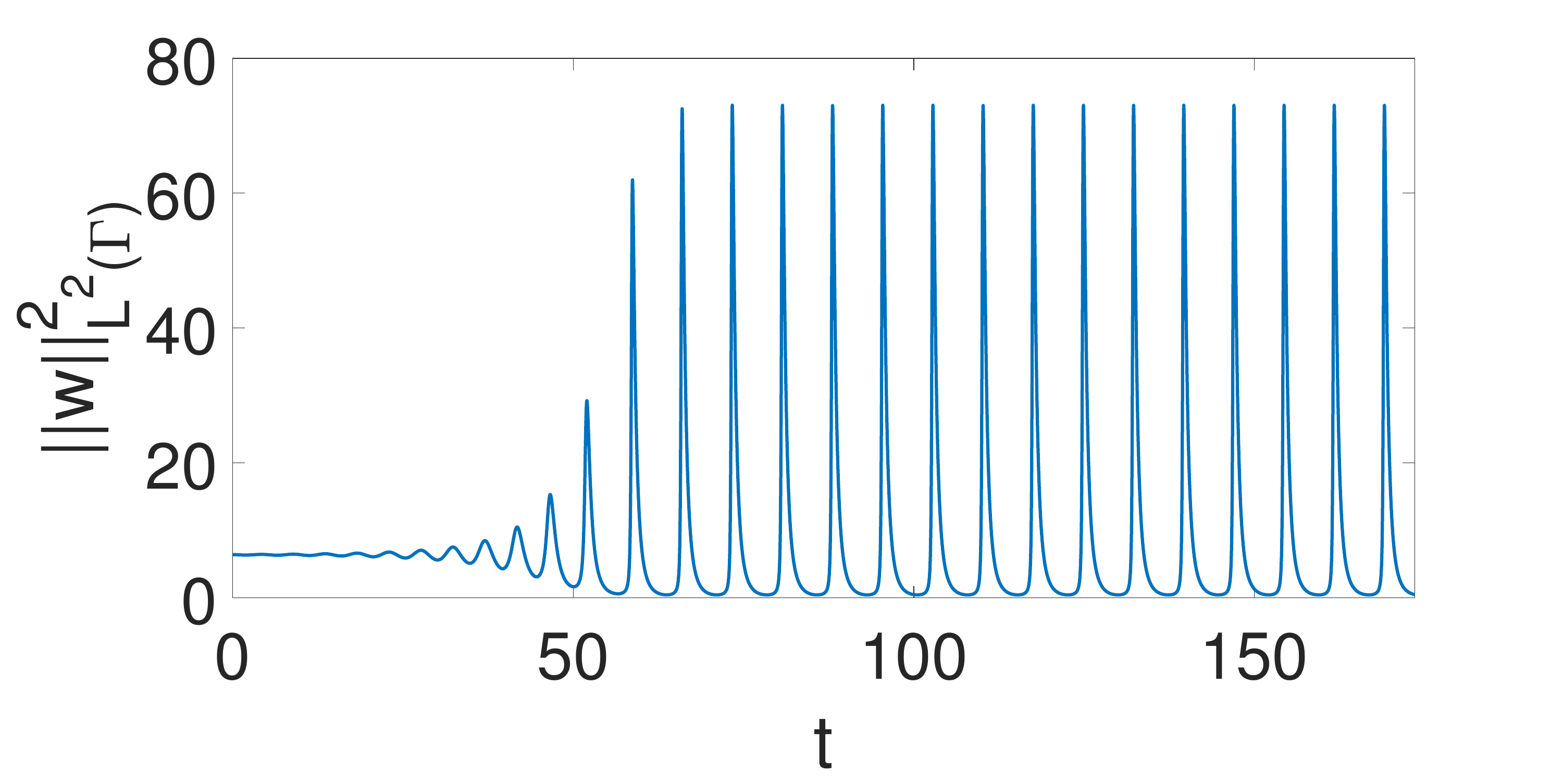}
        \\
        \includegraphics[width = 0.49 \textwidth]{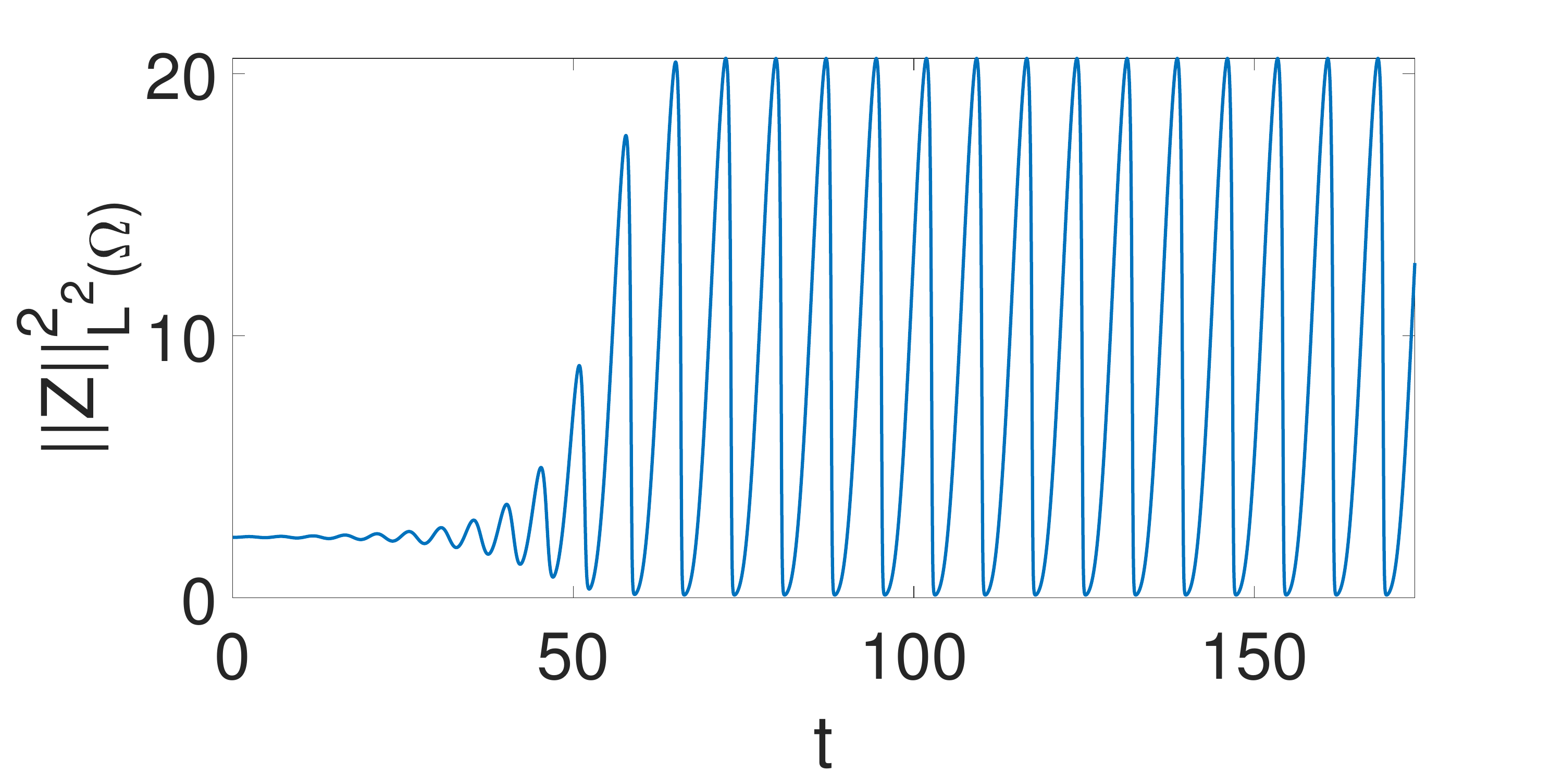} \hfill \includegraphics[width = 0.49 \textwidth]{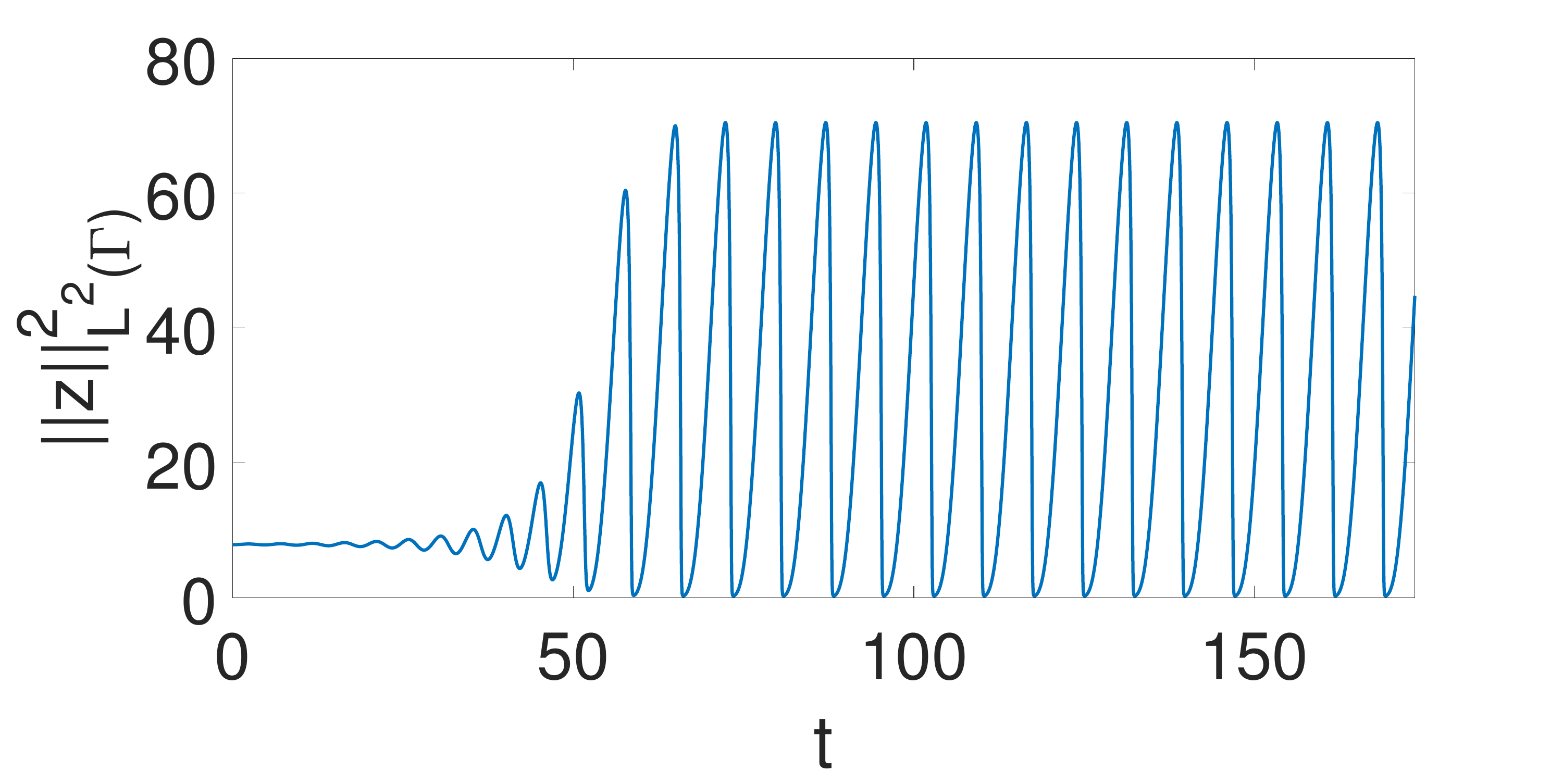}
        \caption{Cell-polarity model for $\ell = 1$.}
        \label{fig:fahadl=1}
    \end{figure}

    \newpage
    
    \ 
    
    \begin{figure}[hb]
        \centering
        \includegraphics[width = 0.49 \textwidth]{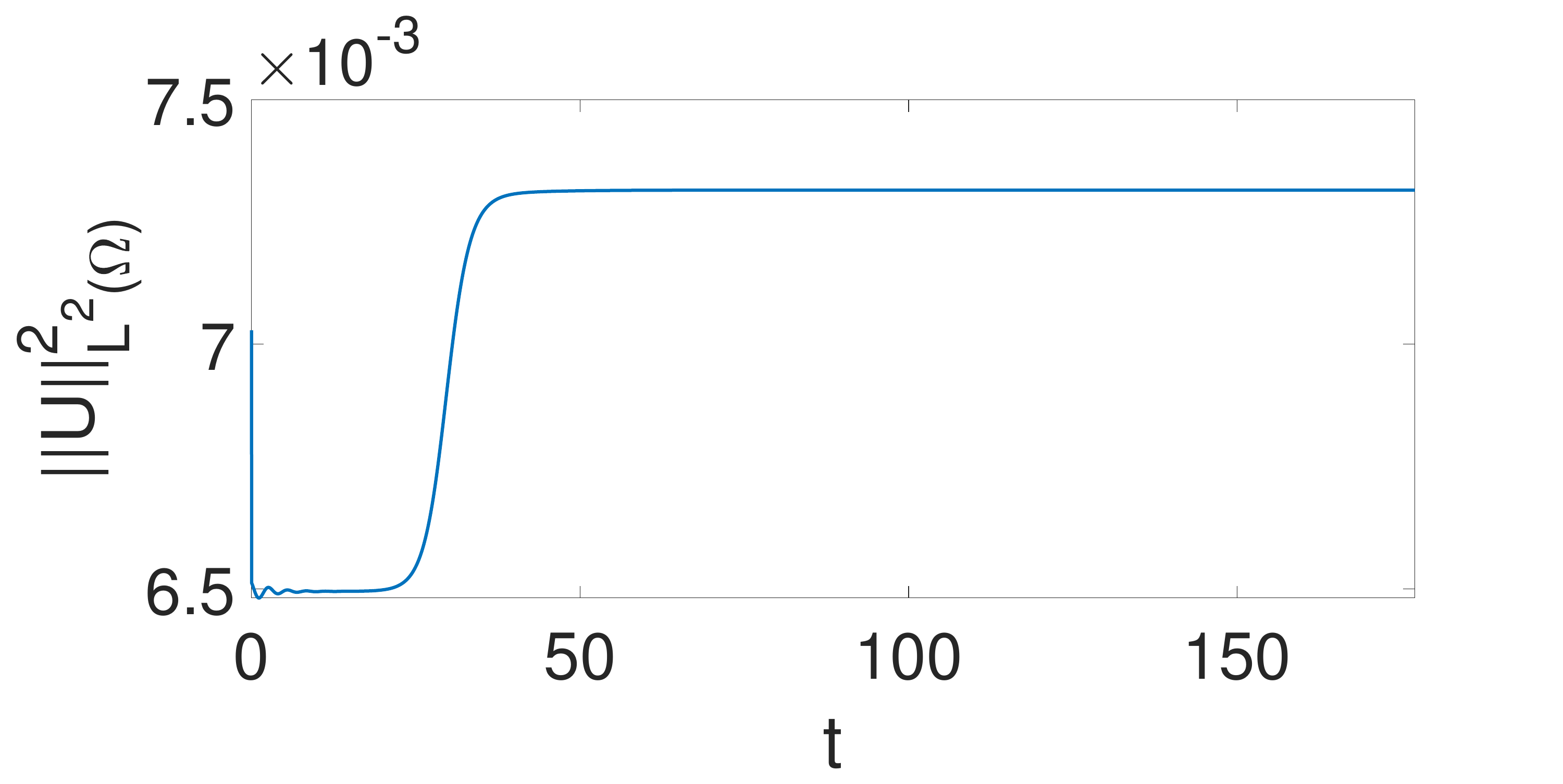} \hfill \includegraphics[width = 0.49 \textwidth]{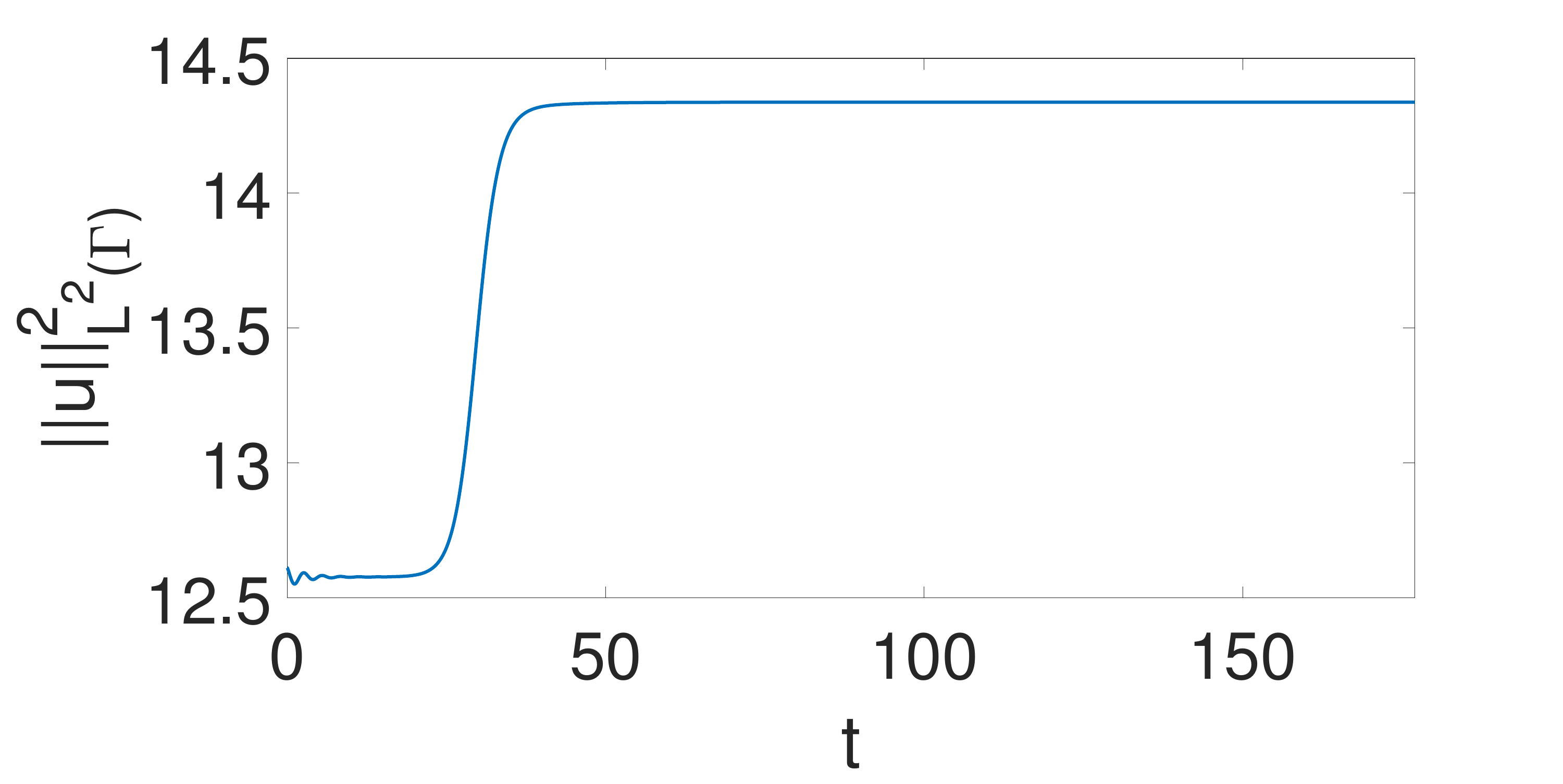}
        \\
        \includegraphics[width = 0.49 \textwidth]{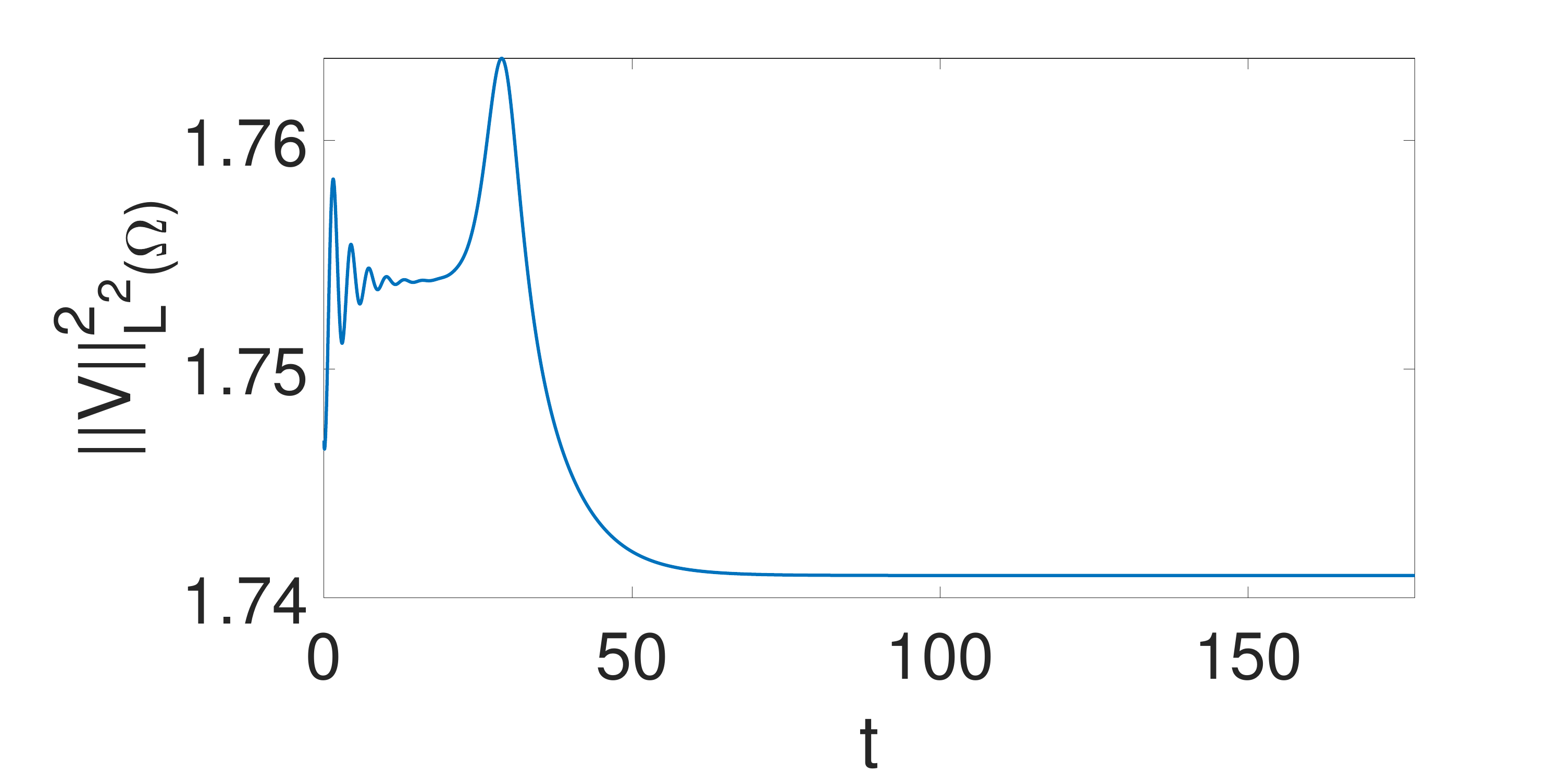} \hfill \includegraphics[width = 0.49 \textwidth]{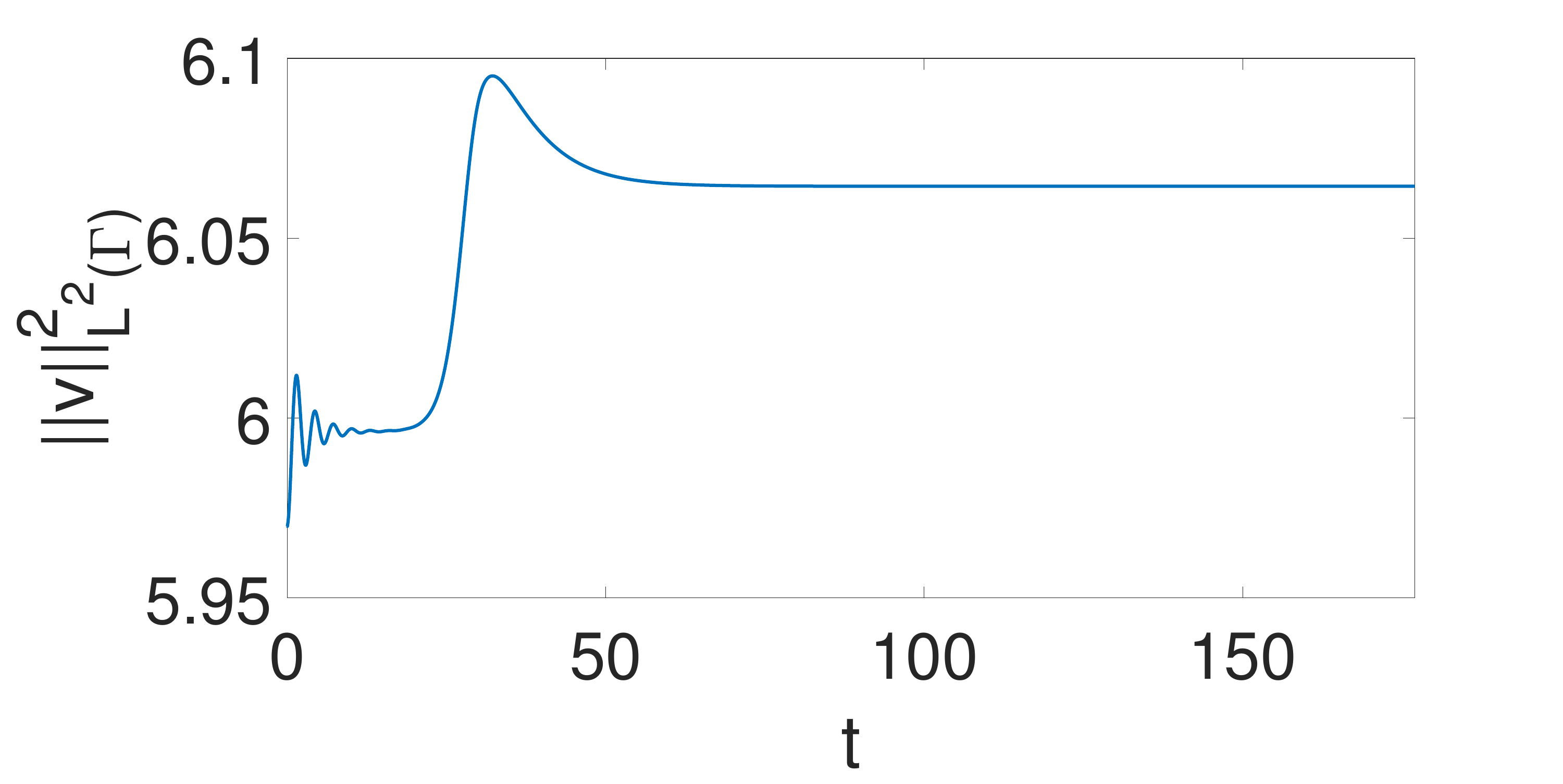}
        \\
        \includegraphics[width = 0.49 \textwidth]{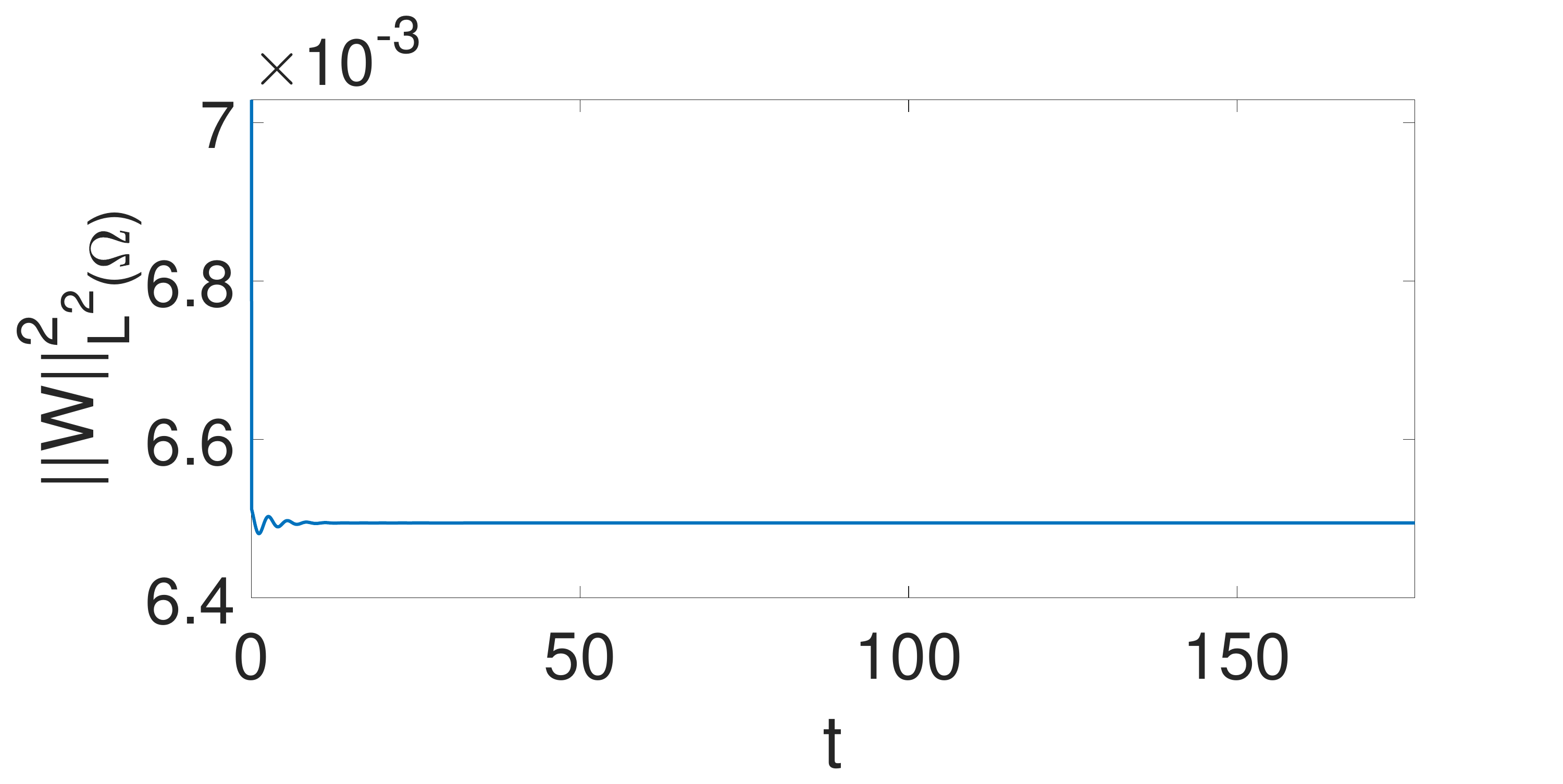} \hfill \includegraphics[width = 0.49 \textwidth]{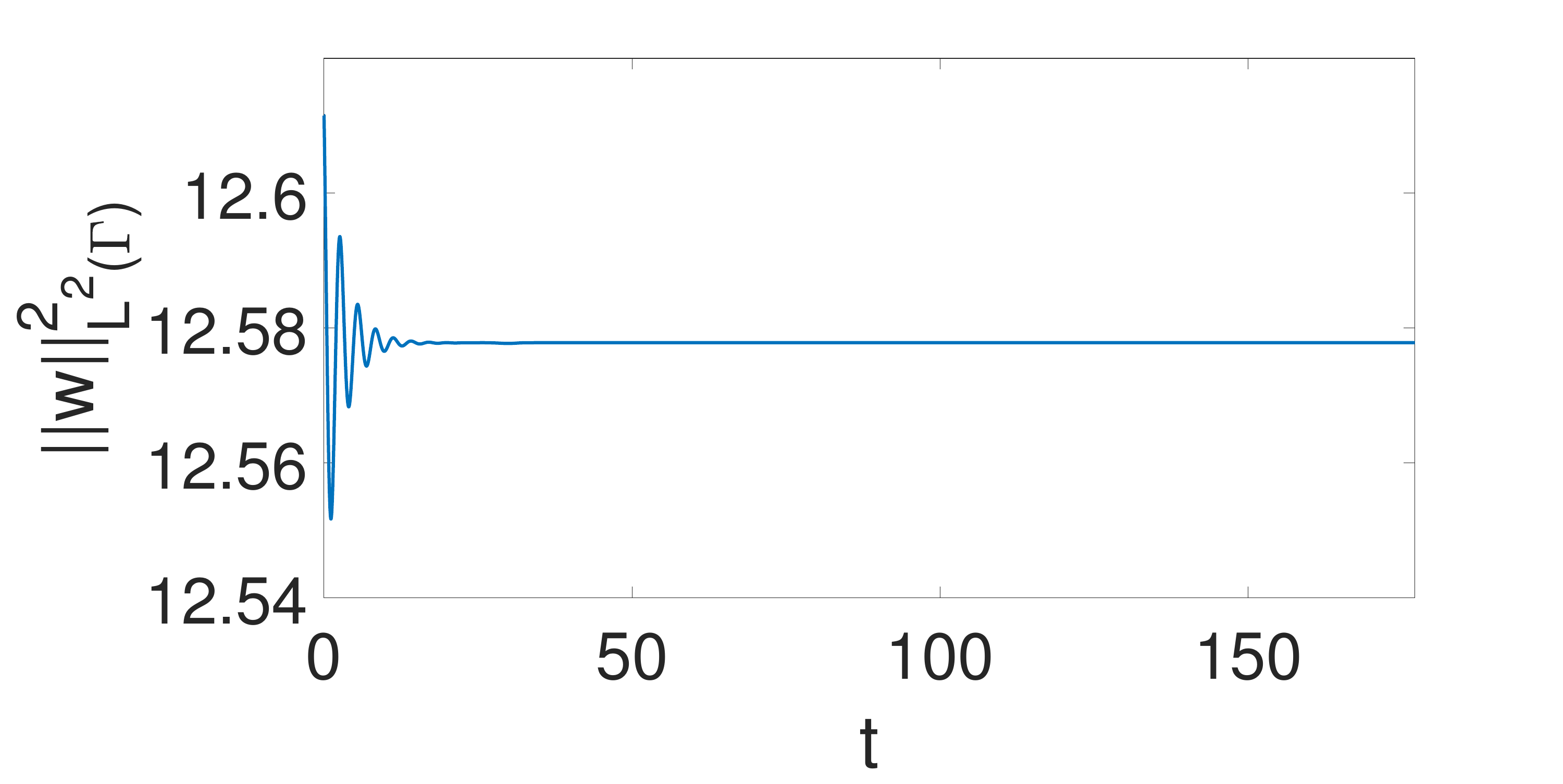}
        \\
        \includegraphics[width = 0.49 \textwidth]{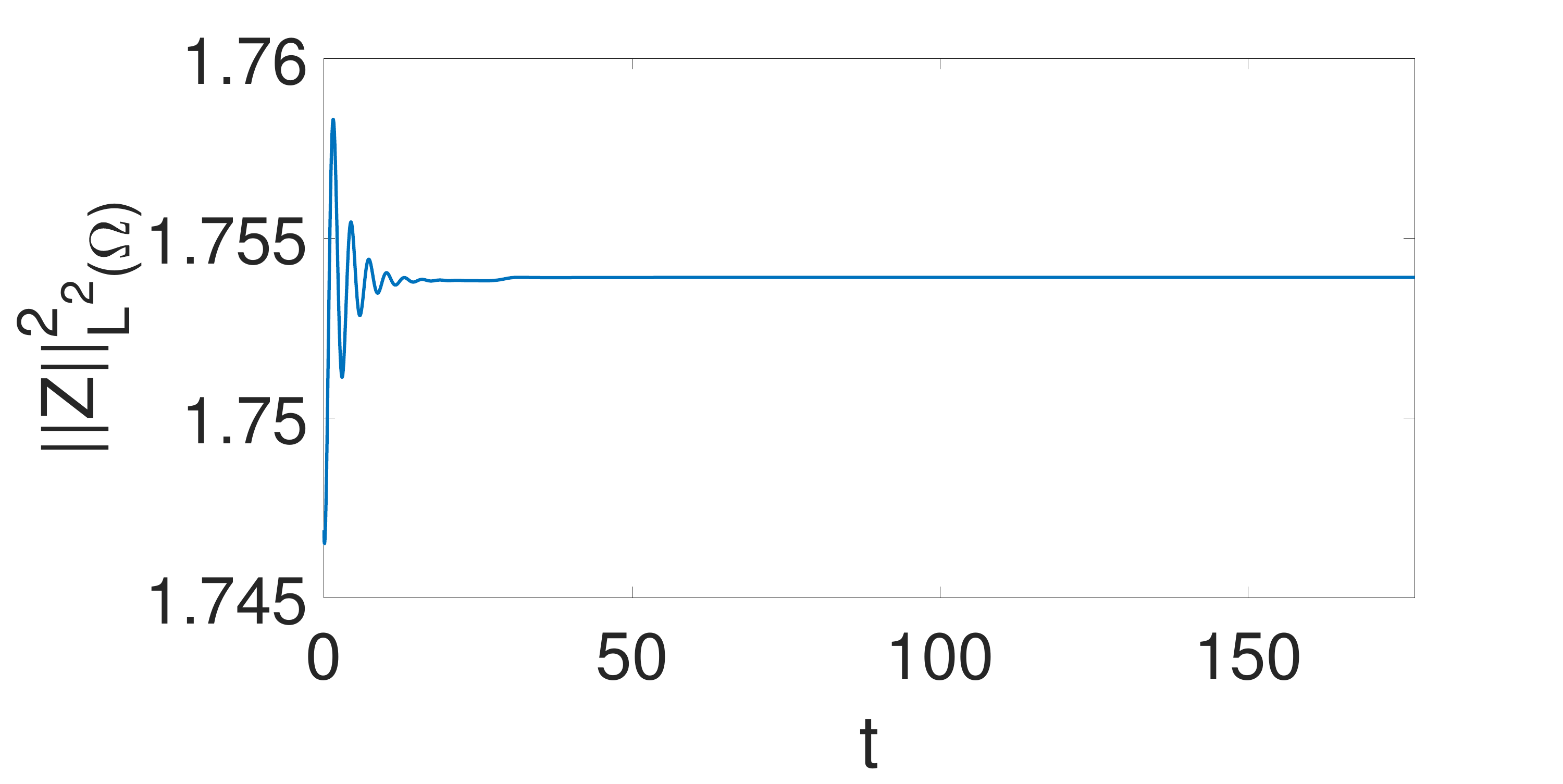} \hfill \includegraphics[width = 0.49 \textwidth]{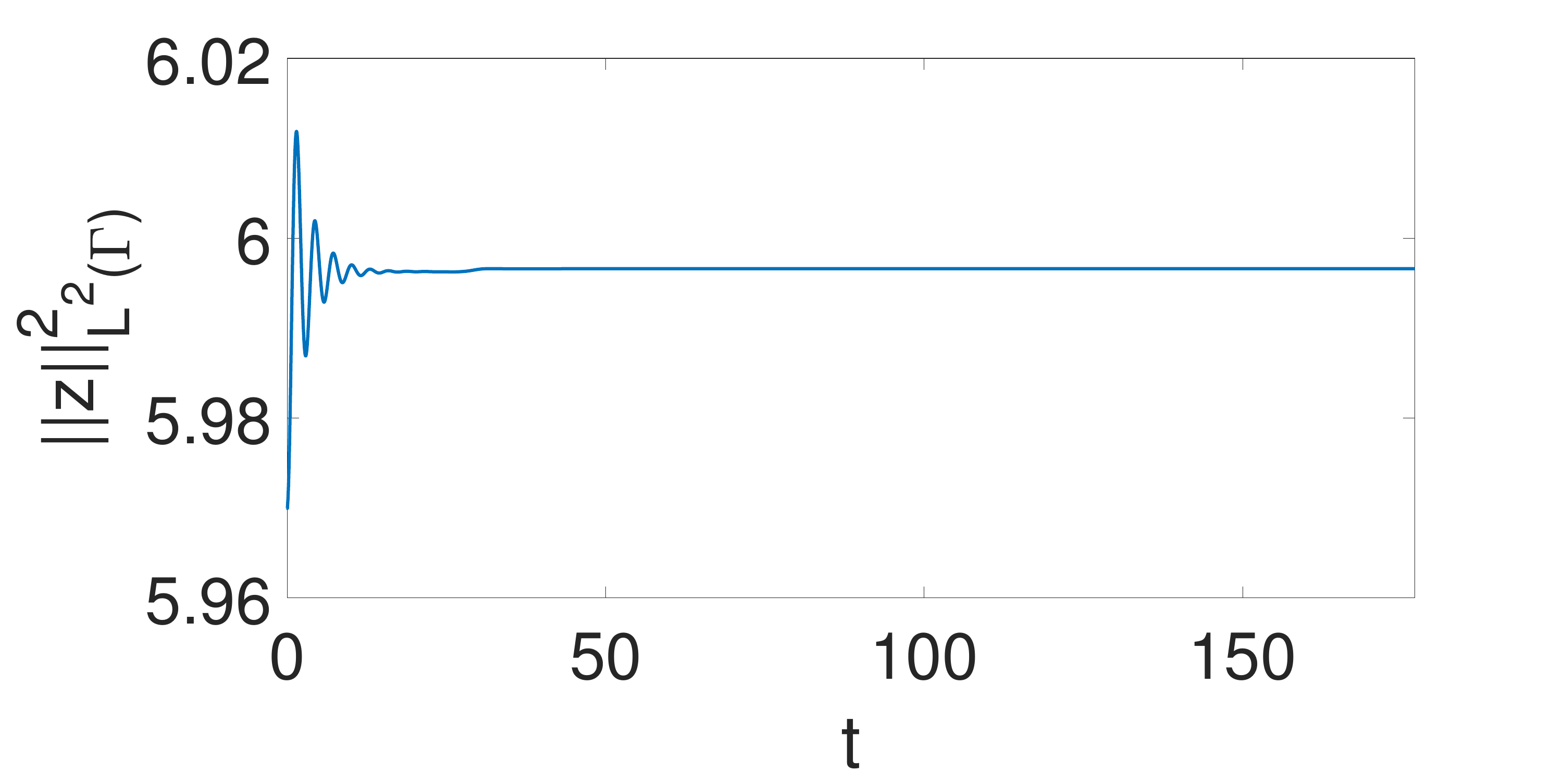}
        \caption{Cell-polarity model for $\ell = 2$.}
    \end{figure}

    \newpage

    \
    
    \begin{figure}[hb]
        \centering
        \includegraphics[width = 0.49 \textwidth]{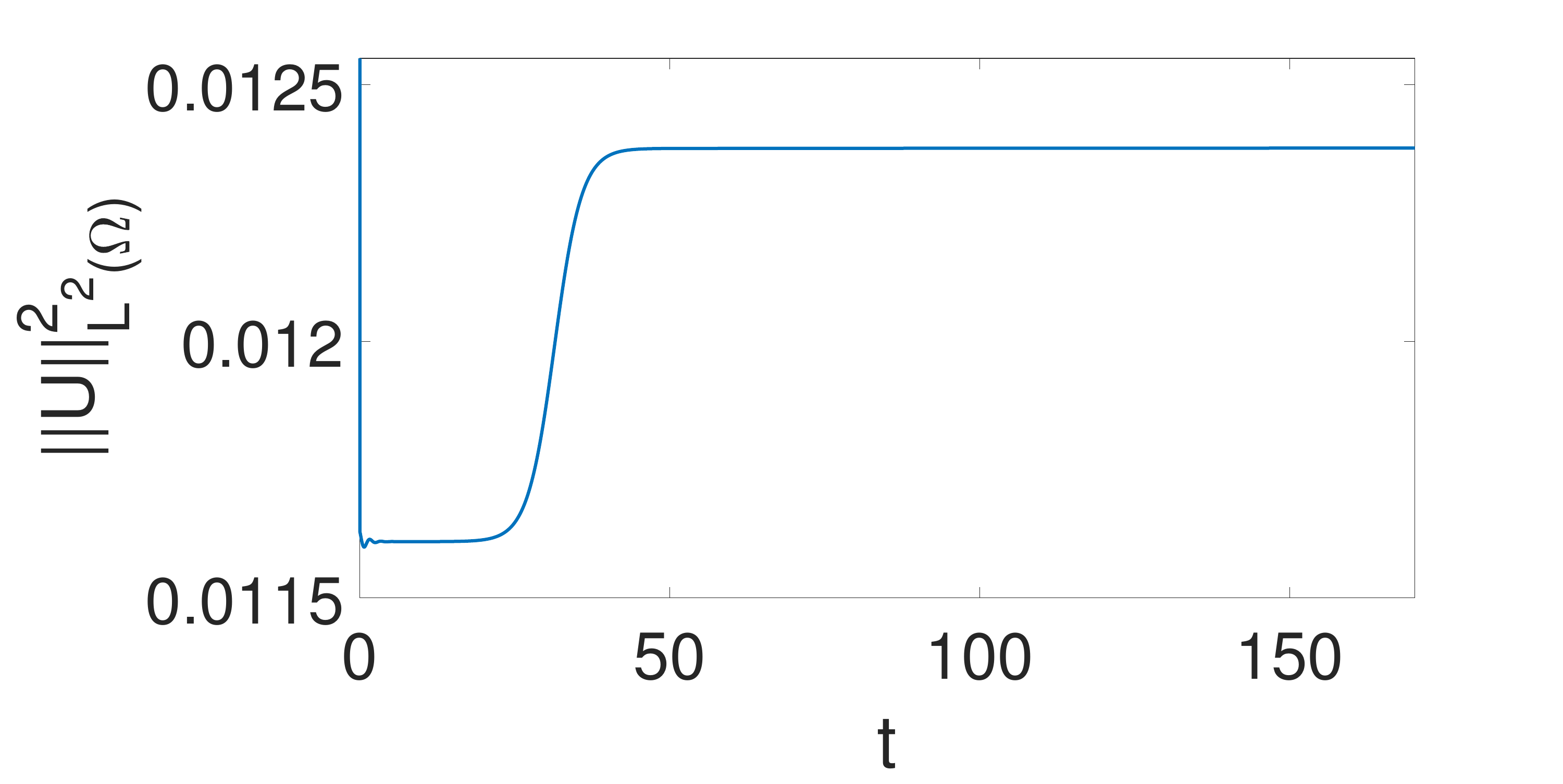} \hfill \includegraphics[width = 0.49 \textwidth]{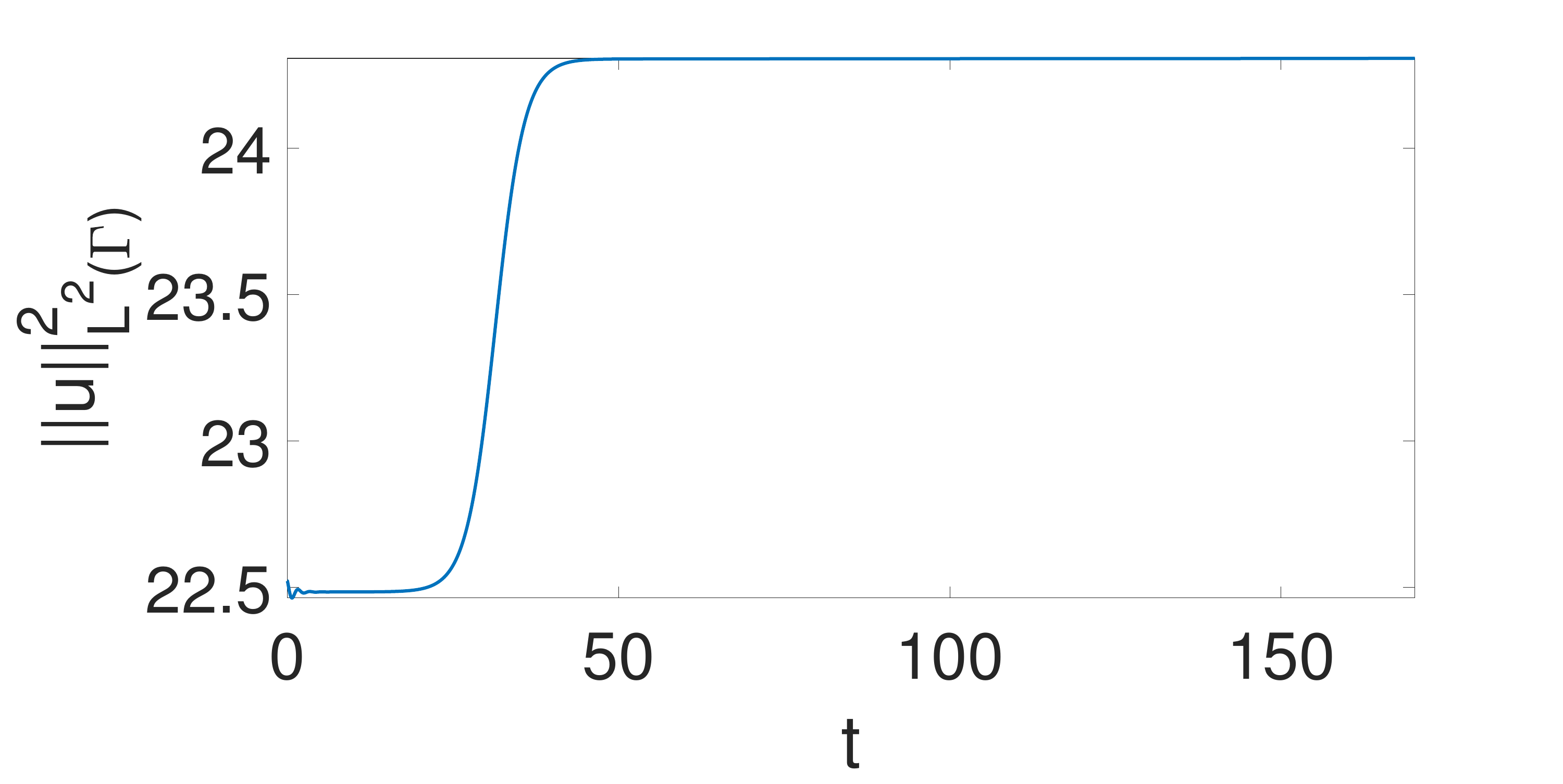}
        \\
        \includegraphics[width = 0.49 \textwidth]{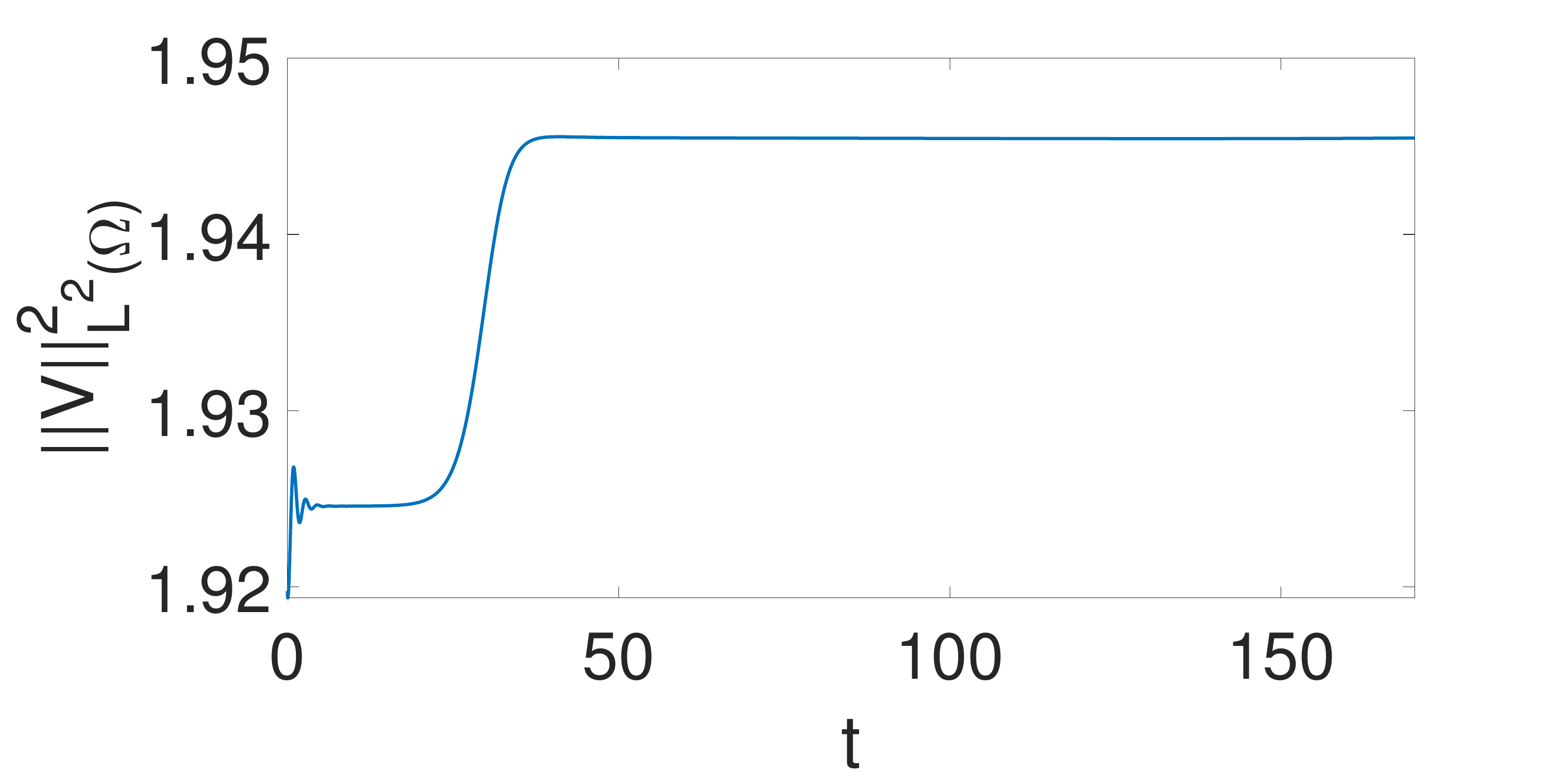} \hfill \includegraphics[width = 0.49 \textwidth]{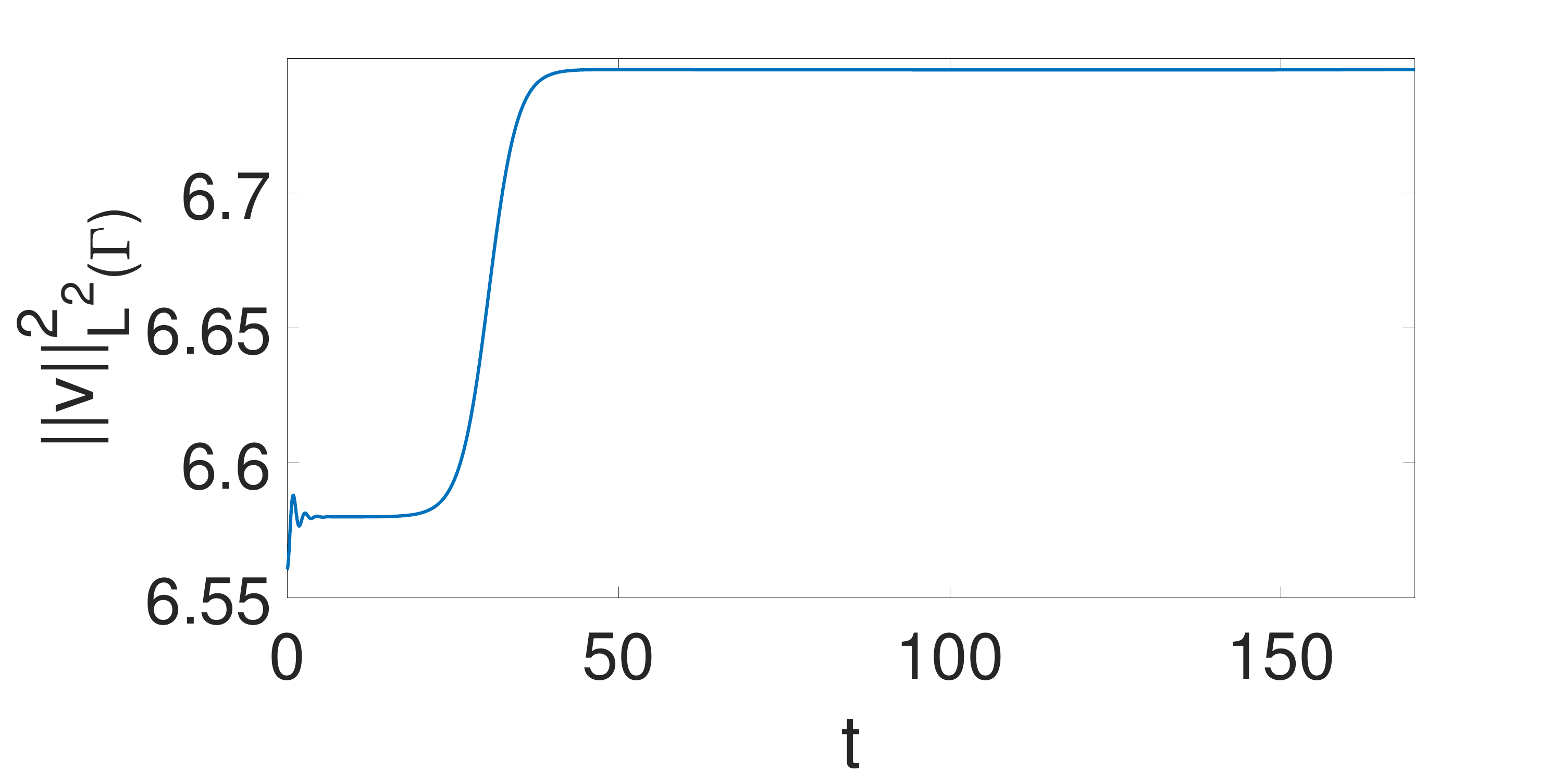}
        \\
        \includegraphics[width = 0.49 \textwidth]{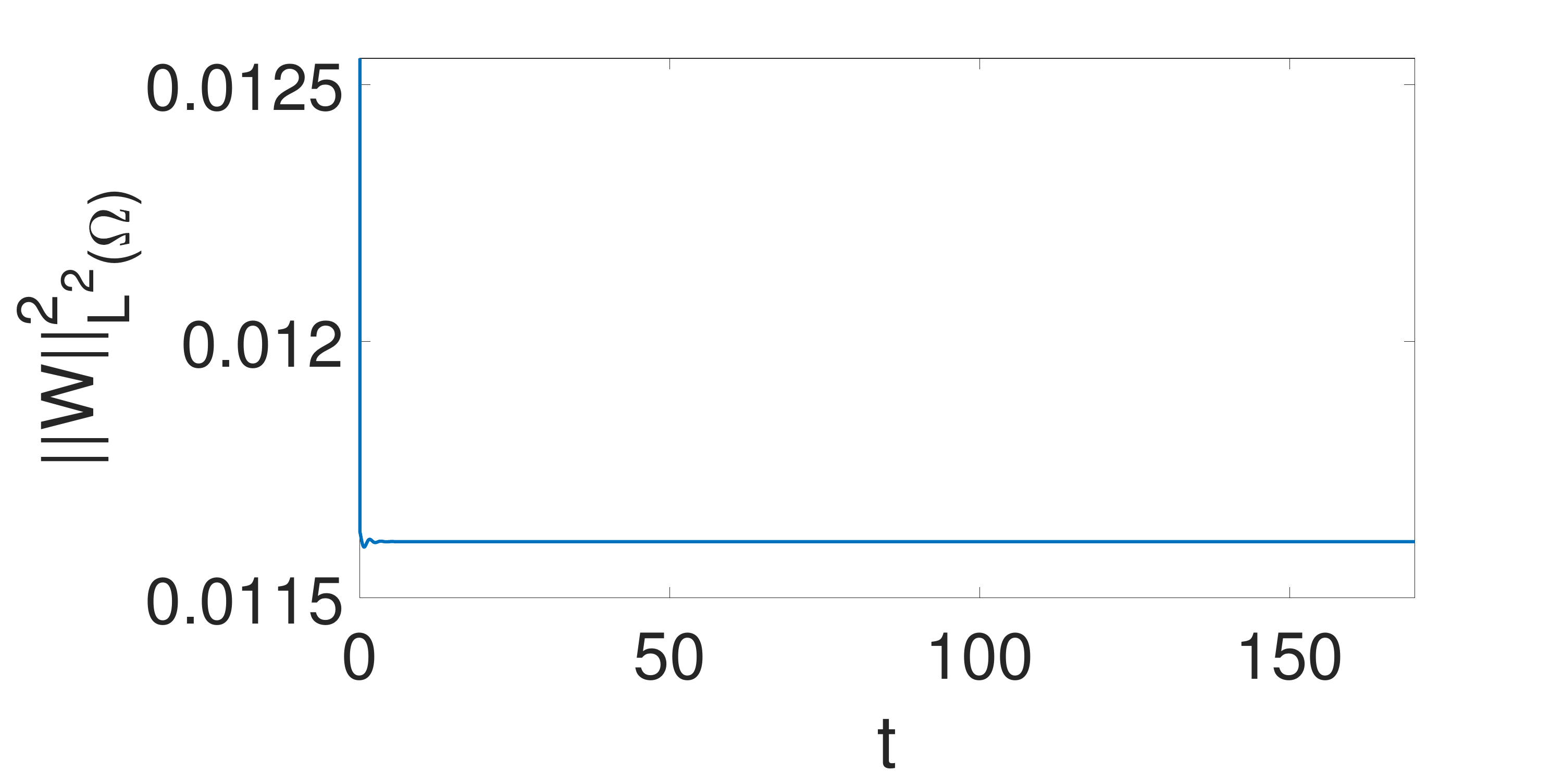} \hfill \includegraphics[width = 0.49 \textwidth]{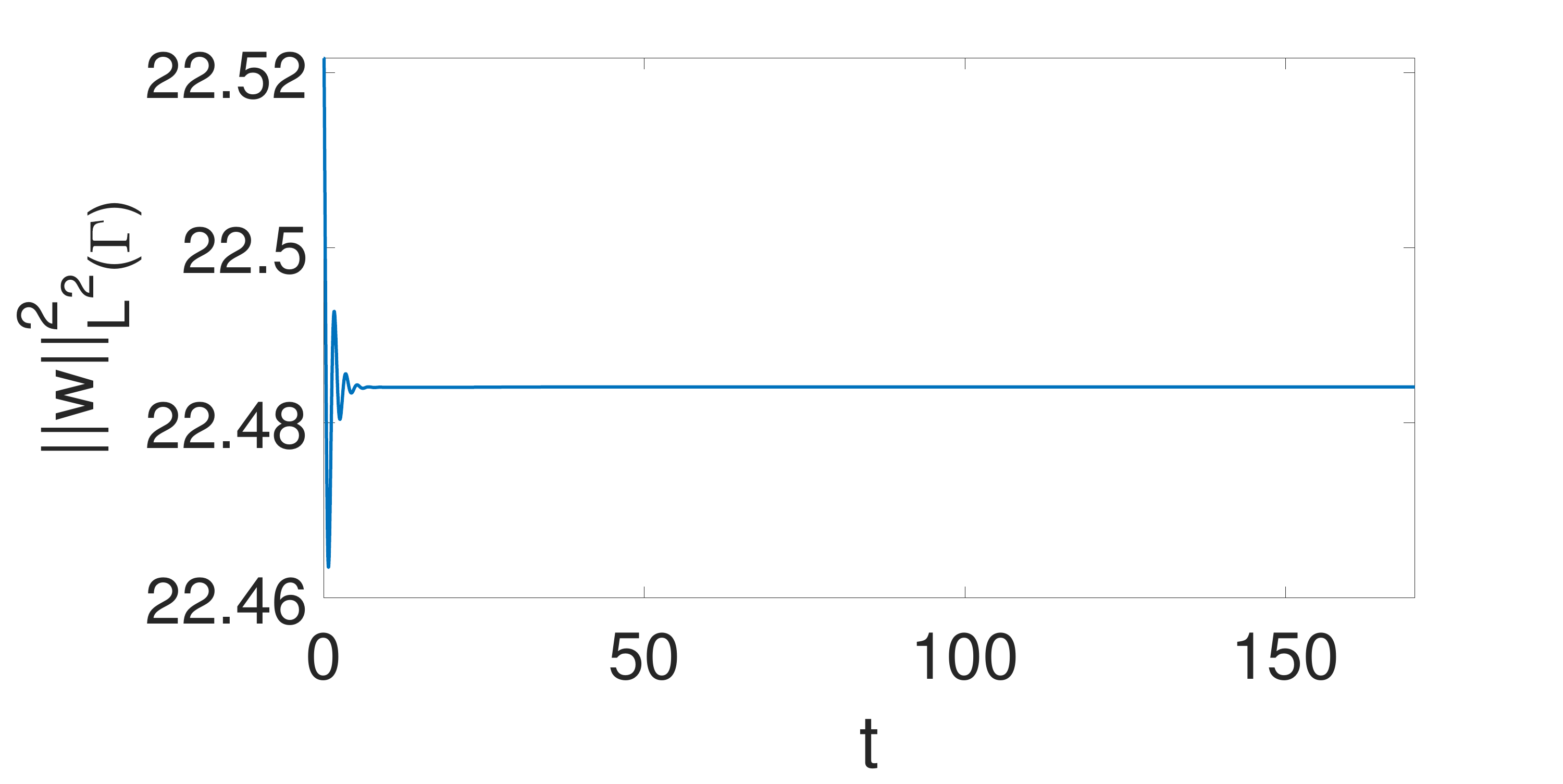}
        \\
        \includegraphics[width = 0.49 \textwidth]{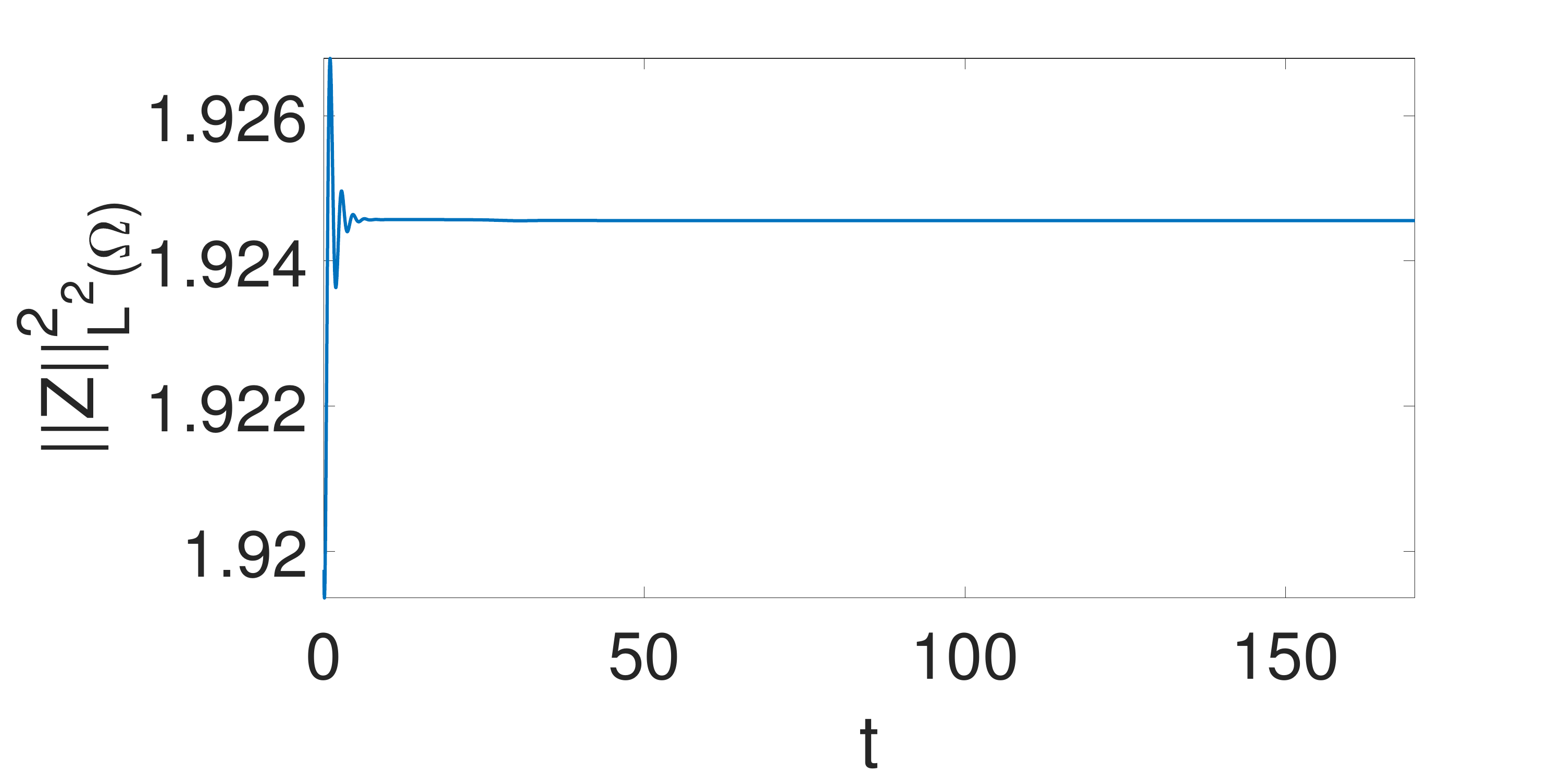} \hfill \includegraphics[width = 0.49 \textwidth]{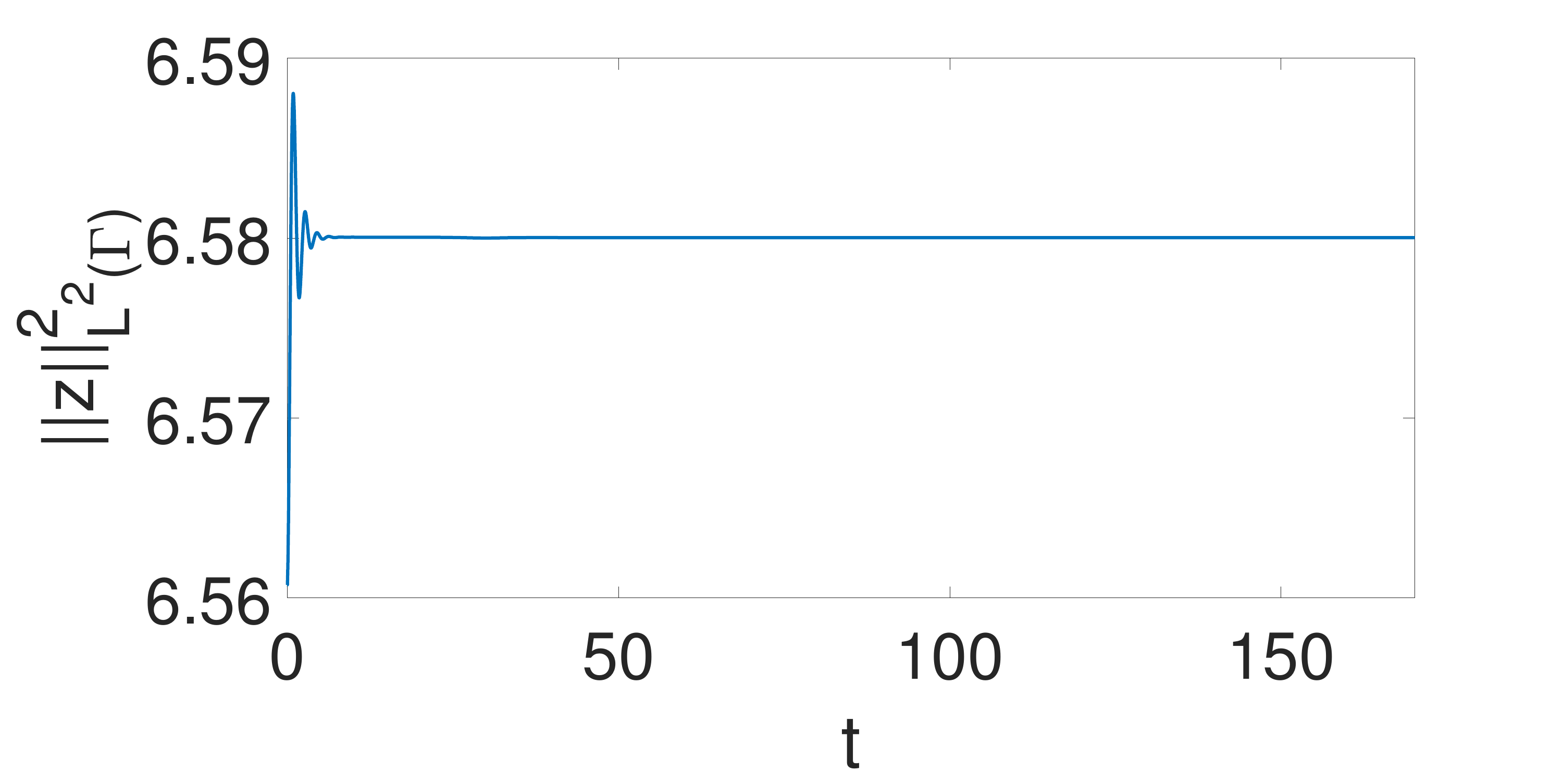}
        \caption{Cell-polarity model for $\ell = 3$.}
    \end{figure}

\end{document}